\RequirePackage{fix-cm}
\documentclass[smallextended]{svjour3}       
\smartqed  
\usepackage{graphicx}

\usepackage{amsmath,amssymb,array}
\usepackage{siunitx}
\usepackage{longtable,tabularx}
\usepackage{algorithm}
\usepackage{algorithmic}
\usepackage{multirow}
\usepackage{color}

\DeclareMathOperator{\Tr}{\mathrm{Tr}}

\newcommand{\argmax}{\mathop{\rm argmax}\limits}
\newcommand{\argmin}{\mathop{\rm argmin}\limits}

\def\Ker{\mathop{\rm Ker}}
\def\im{\mathop{\rm Im}}

\def\uni{{\rm uni}}

\def\Z{{\mathbb Z}}

\def\X{{\mathcal{X}}}
\def\Y{{\mathcal{Y}}}
\def\Z{{\mathcal{Z}}}

\def\P{{\mathcal{P}}}

\def\Pro{\mathop{\Gamma}\nolimits}

\newcommand\saa{\mathsf{a}}
\newcommand\sbb{\mathsf{b}}
\newcommand\scc{\mathsf{c}}
\newcommand\sdd{\mathsf{d}}
\newcommand\see{\mathsf{e}}
\newcommand\sff{\mathsf{f}}
\newcommand\sgg{\mathsf{g}}
\newcommand\shh{\mathsf{h}}

 \newenvironment{proofof}[1]{\vspace*{5mm} \par \noindent
{\it Proof of #1:\hspace{2mm}}}{\qed
}
\def\Label#1{\label{#1}\ [\ \text{#1}\ ]\ }
\def\Label{\label}
\begin{document}
\title{Reverse em-problem based on Bregman divergence
and its application to classical and quantum information theory
}
\titlerunning{Reverse em-problem based on Bregman divergence}
\author{Masahito~Hayashi}

\institute{M. Hayashi \at
School of Data Science, The Chinese University of Hong Kong, Shenzhen, Longgang District, Shenzhen, 518172, China,
International Quantum Academy (SIQA), Futian District, Shenzhen 518048, China,
and
Graduate School of Mathematics, Nagoya University, Chikusa-ku, Nagoya 464-8602, Japan.
\\
              \email{e-mail: hmasahito@cuhk.edu.cn, hayashi@iqasz.cn}           
}
\date{Received: date / Accepted: date}

\maketitle

\begin{abstract}
The recent paper (IEEE Trans. IT 69, 1680) introduced an analytical method for calculating the channel capacity without the need for iteration. 
This method has certain limitations that restrict its applicability. Furthermore, the paper does not provide an explanation as to why the channel capacity can be solved analytically in this particular case.
In order to broaden the scope of this method and address its limitations, we turn our attention to the reverse em-problem, proposed by 
Toyota (Information Geometry, 3, 1355 (2020)). 
This reverse em-problem involves iteratively applying the inverse map of the em iteration to calculate the channel capacity, which represents the maximum mutual information. However, several open problems remained unresolved in Toyota's work.
To overcome these challenges, we formulate the reverse em-problem based on Bregman divergence and provide solutions to these open problems. 
Building upon these results, we transform the reverse em-problem into em-problems and derive a non-iterative formula for the reverse em-problem. 
This formula can be viewed as a generalization of the aforementioned analytical calculation method. Importantly, this derivation sheds light on the information geometrical structure underlying this special case.
By effectively addressing the limitations of the previous analytical method and providing a deeper understanding of the underlying information geometrical structure, our work significantly expands the applicability of the proposed method for calculating the channel capacity without iteration.
\end{abstract}


\section{Introduction}\Label{S1}
The em-algorithm is widely recognized as a valuable tool in various domains, including machine learning and neural networks \cite{Amari,Fujimoto,Allassonniere}. 
This algorithm is typically formulated within the framework of information geometry, which encompasses important concepts such as exponential families and mixture families \cite{Amari-Nagaoka,Amari-Bregman}. 
This algorithm aims to solve the em-problem, i.e., the minimization of 
the divergence between an exponential family and a mixture family.
In other words, the goal is to identify an element in the mixture family that minimizes the divergence from the given exponential family. The algorithm achieves this by iteratively performing projections onto the exponential family and the mixture family.

Recently, Toyota \cite{Shoji} addressed the opposite problem related to the calculation of classical channel capacity, as depicted in Fig \ref{rev-fig}. 
Specifically, he aimed to find an element in the mixture family that maximizes the minimum divergence from the given exponential family. 
He observed that if the inverse operation of the combined projection exists, repeating it leads to the maximization mentioned above in the case of classical channel capacity \cite{Shannon}. 
Consequently, he proposed an alternative method for calculating the channel capacity, which has been extensively studied in existing literature 
\cite{Arimoto,Blahut,Matz,Yu,SSML,NWS}. 
This problem is referred to as the reverse em-problem. However, Toyota did not establish the existence or uniqueness of the inverse map, nor did he provide a method for computing the inverse of the map. Furthermore, his analysis was limited to the specific scenario of classical channel capacity. These issues remain open challenges in the field.

\begin{figure}[htbp]
\begin{center}
  \includegraphics[width=0.7\linewidth]{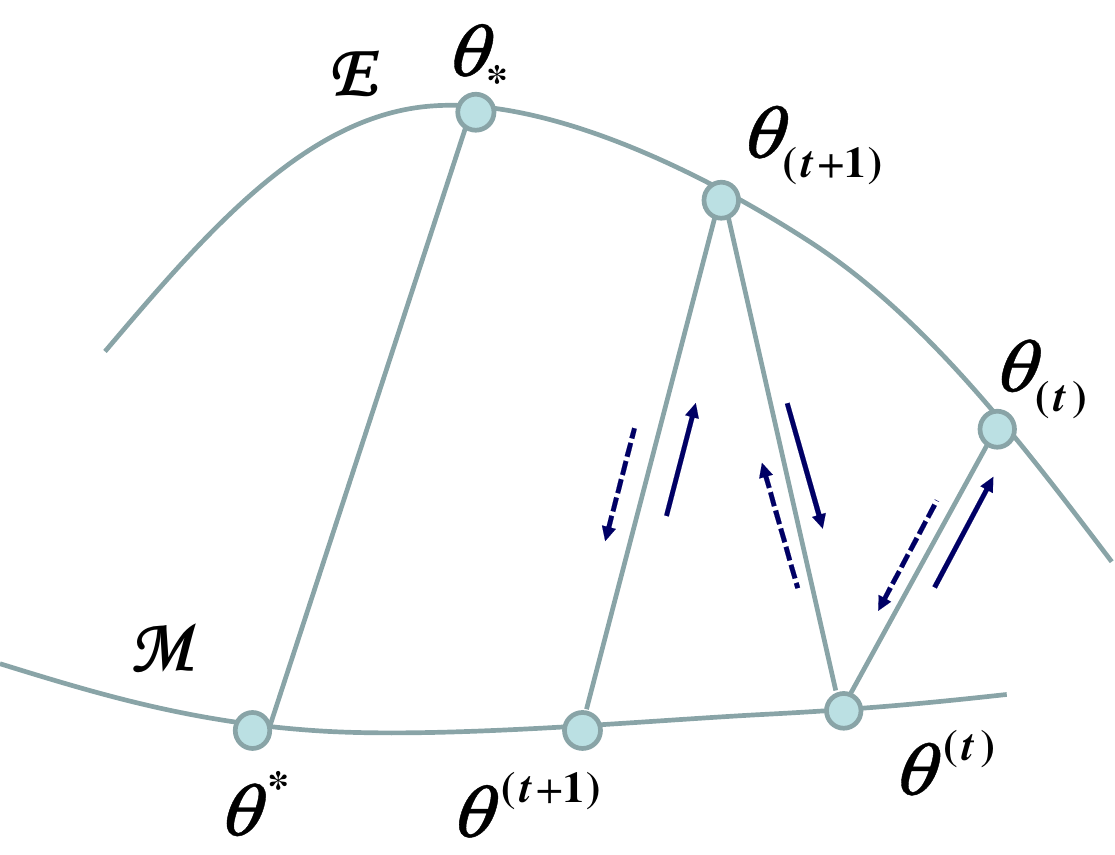}
  \end{center}
\caption{Brief idea of our maximization problem:
${\cal E}$ is an exponential family.
${\cal M}$ is a mixture family.
The solid line expresses the direction of the em-algorithm.
The dashed line expresses the direction of the reverse em-algorithm.
The pair of $\theta_* \in {\cal E}$ and $\theta^*\in {\cal M}$
realized the maximum.}
\Label{rev-fig}
\end{figure}   

Furthermore, a recent paper \cite{exact} introduced an analytical method for calculating the channel capacity without the need for iteration. 
However, this method has certain restrictions that limit its applicability. 
Additionally, the paper does not provide an explanation for why the channel capacity can be solved analytically in this specific case. 
Consequently, 
to expand the applicable range of the method proposed in the paper \cite{exact}, 
this paper aims to generalize this method and 
explore the information geometrical background for the algorithm by the paper \cite{exact}.

Surprisingly, these two problems can be resolved
by addressing the open problems in the reverse em-problem.
By leveraging the framework of Bregman divergence, we can effectively tackle these open problems. 
In this study, we formulate the maximization problem within the framework of Bregman divergence, 
following a similar approach as in the papers 
\cite{em-only,Fujimoto}, which is given in Section \ref{S4-1}.
Moreover, as Theorem \ref{theo:conv:BBem1},
we establish the uniqueness and existence of the inverse map under certain conditions in this general setting. 
Notably, the case of classical channel capacity satisfies these conditions, allowing us to successfully address the problem initially proposed by Toyota \cite{Shoji}. 
In this approach, we introduce a specific parameterization condition for the reverse em-problem and present the iteration process for each step. 
Additionally, we evaluate the convergence speed within this general framework.

In the subsequent step, using the aforementioned results, we convert the reverse em-problem into an em-problem. In Section \ref{S47}, we derive equivalent conditions that determine when an element of the mixture family becomes a fixed point for the iteration function. 
These equivalent conditions transform the reverse em-problem into a problem of finding the intersection between an exponential family and a mixture family, which can be effectively solved through an em-problem. 
Notably, in Section \ref{S47B}, we demonstrate that under certain conditions, the reverse em-problem can be further simplified into a non-iterative form, 
minimizing a particular convex function. 
This reduction results in a problem with fewer free parameters compared to the original reverse em-problem. Importantly, when the reverse em-problem satisfies specific conditions, 
it can be solved analytically without resorting to a minimization problem.
In summary, our approach not only generalizes the analytical calculation method proposed in the paper \cite{exact} but also provides insights into the information geometrical structure underlying the algorithm.
By addressing the open problems in the reverse em-problem, we make significant advancements in the field, enabling more efficient and comprehensive solutions for calculating the channel capacity without iteration.

In the case of the classical channel capacity \cite{Shannon,Arimoto,Blahut,Matz,Yu,SSML,NWS},
the above conditions are satisfied. 
Consequently, the calculation of the channel capacity can be transformed into a minimization problem of a specific convex function. 
This transformation yields a new calculation algorithm for the classical channel capacity. 
Notably, this algorithm can be viewed as a generalization of the analytical algorithm proposed in the paper \cite{exact} because it coincides with the analytical algorithm when the classical channel satisfies the same condition as described in \cite{exact}. 
Moreover, this reduction to the result presented in \cite{exact} provides insight into the information geometrical background explaining why the channel capacity can be solved analytically in this special case.
Furthermore, even when the condition from \cite{exact} does not hold, 
our calculation algorithm still exhibits advantages. Specifically, under certain conditions, the obtained algorithm has a reduced number of free parameters compared to the original problem of the classical channel capacity. 
It is worth noting that a similar method was previously derived by Muroga \cite{Muroga}. However, our approach offers slight improvements over Muroga's method, as elucidated in Remark \ref{Rem3}.
Additionally, we extend the application of our results to two other scenarios: the capacity of classical wire-tap channels \cite{Wyner,CK79} and the capacity of classical-quantum channels \cite{Holevo,SW}. These maximization problems have been explored in numerous papers 
\cite{Yasui,Nagaoka,Dupuis,Sutter,Li,RISB}. 

The remaining part of this paper is organized as follows.
Section \ref{Section.IG-structure} formulates general basic properties for Bregman divergence.
Section \ref{S3} explains how the set of probability distributions and the set of quantum states 
satisfy the condition for Bregman divergence.
We omit the proofs of statements in Sections \ref{Section.IG-structure} and \ref{S3}, and their proofs are given in the paper \cite{em-only}.
Section \ref{Sec:BBem} formulates the reverse em-problem, and studies its various properties. 
Section \ref{S10} applies these results to the capacity of a classical channel.
Section \ref{S11} applies these results to the secrecy capacity of a degraded wiretap channel.
Section \ref{S12} applies these results to the capacity of a classical-quantum channel.

\section{Bregman divergence system}\Label{Section.IG-structure}
In this section, we formulate the Bregman divergence system
as a preparation for our maximization problem.
We omit the proofs of statements in this section and their proofs are given in the paper \cite{em-only}.
The contents of this section will be used in the main body and the appendices.

\subsection{Legendre transform}\Label{S2-0}
In this paper, 
a sequence $a= (a^i)_{i=1}^k$ with an upper index expresses
a vertical vector 
and 
a sequence $b= (b_i)_{i=1}^k$ with a lower index expresses
a horizontal vector as
\begin{align}
a= \left(
\begin{array}{c}
a^1 \\
a^2 \\
\vdots \\
a^k
\end{array}
\right), \quad
b= (b_1, b_2,\ldots, b_k).
\end{align}

We choose an open convex $\Theta$ set in $\mathbb{R}^d$ and 
a $C^\infty$-class strictly convex function $F:\Theta \rightarrow \mathbb{R}$.
Using the convex function $F$, we introduce another parametrization 
$\eta =(\eta_1, \ldots, \eta_d)\in \mathbb{R}^d$ as
\begin{align}
\eta_j := \partial_j F(\theta),\Label{du1}
\end{align}
where $\partial_j$ expresses the partial derivative for the $j$-th variable $\partial_j$.
We also use the notation for the vector $\nabla^{(e)} [F](\theta):=(\partial_j F(\theta))_{j=1}^d$.
Hence, the relation \eqref{du1} is rewritten as
\begin{align}
\eta = \nabla^{(e)} [F](\theta).\Label{M1}
\end{align}
Therefore, $\nabla^{(e)}$ can be considered as a horizontal vector.

Since $F$ is a $C^\infty$-class strictly convex function,
this conversion is one-to-one.
The parametrization 
$\eta_j$ is called the mixture parameter
while the original parameter  ${\theta}=({\theta}^1, \ldots, {\theta}^d)$
is called the natural parameter.
In the following, $\Xi$ expresses the open set of vectors $\eta(\theta) =(\eta_1, \ldots, \eta_d)$ given in \eqref{du1}.
For $\eta \in \Xi$,
we define the {\it Legendre transform} $F^*={\cal L}[F]$ of $F$ 
\begin{align}
F^*(\eta)=\sup _{\theta \in \Theta} \langle \eta,\theta\rangle -F(\theta).
\Label{MN1}
\end{align}

We denote the partial derivative for the $j$-th variable under the mixture parameter
by $\partial^j$.
The partial derivative of $F^*$ is given as
\cite[Section 3]{Fujimoto}\cite[Section 2.2]{hayashi}
\begin{align}
\partial^j F^*
(\eta(\theta) )=\theta^j. \Label{du2}
\end{align}
In the same way as the above, we use the notation
$\nabla^{(m)} [F^*](\eta):=(\partial^j F^*(\eta))_{j=1}^d$.
The relation \eqref{du2} is rewritten as
\begin{align}
\theta = \nabla^{(m)} [F^*] (\eta(\theta) ).\Label{M2}
\end{align}

In the following discussion, we address subfamilies related to
$m$ vectors $v_1, \ldots, v_m \in \mathbb{R}^d$.
For preparation for such cases, we prepare 
the following two equations, which  will be used for calculations based on mixture parameters.
The $d \times m $ matrix $V$ is defined as $(v_1 \ldots v_m )$.
The multiplication function of V from the left (right) hand side  
is denoted by $L[V]$ ($R[V]$).
The relation
 \begin{align}
\partial_j (F\circ L[V])(\theta) =
\frac{\partial F}{\partial \theta^j}(V \theta) =
\sum_{i} v_{j}^i \partial_i F(V \theta) =
(R[V]\circ (\nabla^{(e)} [F] )\circ L[V] (\theta))_j,
\end{align}
implies that
\begin{align}
\nabla^{(e)} [F\circ L[V]] &=   R[V]\circ (\nabla^{(e)} [F] )\circ L[V].\Label{E12}
\end{align}

Similarly, the relation 
\begin{align}
\nabla^{(m)} [ F^* \circ R[V]]
=L[V] \circ \nabla^{(m)} [ F^*] \circ R[V]
\Label{MDF}
\end{align}
holds. Also, we have
\begin{align}
& (F^* \circ R[V])^* (\theta')
=
\sup_{\eta} \Big(\langle \eta ,\theta'\rangle
- \sup _{\theta \in \Theta} \Big(\langle \eta V,\theta\rangle -F( \theta) \Big)
\Big)\nonumber\\
=&
\sup_{\eta} \inf_{\theta \in \Theta}
\big(\langle \eta ,\theta'- V\theta\rangle
 +F( \theta) \big)
=\inf_{\theta: \theta'= V\theta} F( \theta) .
\Label{VCA}
\end{align}

\subsection{Exponential subfamily}\Label{S2-1}
Next, we introduce an exponential subfamily, and discuss its properties.
We say that a subset $\mathcal{E} \subset \Theta$ is an {\it exponential subfamily} 
generated by $l$ linearly independent vectors $v_1,\ldots, v_l \in \mathbb{R}^d$
at $\theta_0 \in \Theta$ 
when the subset $\mathcal{E}$ is given as
\begin{align}
    \mathcal{E} &= \left\{ \phi_{{\cal E}}^{(e)}(\bar{\theta})  \in \Theta \left| 
    \bar{\theta} \in \Theta_{{\cal E}}    \right. \right\}.
\end{align}
In the above definition,
$\phi_{{\cal E}}^{(e)} (\bar{\theta})$ is defined 
for $ \bar{\theta} =( \bar{\theta} ^1, \ldots,  \bar{\theta}^l ) \in 
\mathbb{R}^l$ as
\begin{align}
\phi_{{\cal E}}^{(e)}(\bar{\theta}) &:= \theta_0 + \sum_{j=1}^l \bar{\theta}^j v_j  
\end{align}
and the set $\Theta_{{\cal E}}$ is defined as 
\begin{align}
\Theta_{{\cal E}} := \{ \bar{\theta} \in \mathbb{R}^l |  
\phi_{{\cal E}}^{(e)}(\bar{\theta}) \in \Theta\}.
\end{align}
The set $\Theta_{{\cal E}}$ is an open set because $\Theta$ is an open set.
In the following, we restrict the domain of $\phi_{{\cal E}}^{(e)}$
to $\Theta_{{\cal E}}$.
We define the inverse map 
$\psi_{{\cal E}}^{(e)}:=(\phi_{{\cal E}}^{(e)})^{-1}: {\cal E}\to \Theta_{{\cal E}}$.

For an exponential subfamily ${\cal E}$, we define 
the function $F_{{\cal E}}$ as
\begin{align}
F_{{\cal E}}(\bar{\theta}):= F (\phi_{{\cal E}}^{(e)}(\bar{\theta})).
\Label{CXO}
\end{align}
In fact, even in an exponential subfamily ${\cal E}$,
we can employ the mixture parameter
${\psi}_{{\cal E},j}^{(m)} (\phi_{{\cal E}}^{(e)}(\bar{\theta})):= 
\partial_j F_{{\cal E}} (\bar{\theta})$
because the map $\bar{\theta}\mapsto F_{{\cal E}} (\bar{\theta})$ is also
a $C^\infty$-class strictly convex function.
For the latter discussion, 
we prepare the set $\Xi_{{\cal E}}:= \{ (\partial_j F_{{\cal E}} (\bar{\theta}))_{j=1}^l \}_{\bar{\theta} \in \Theta_{{\cal E}}}$, and the inverse map 
$\phi_{{\cal E}}^{(m)}:=(\psi_{{\cal E}}^{(m)})^{-1}: 
\Xi_{{\cal E}}\to {\cal E}$.

\subsection{Mixture subfamily}\Label{S2-2}
Next, we introduce a mixture subfamily, and discuss its properties.
For $d$ linearly independent vectors $u_1, \ldots, u_d \in \mathbb{R}^d$, and 
a vector $a=(a_1, \ldots, a_{d-k} )^T \in \mathbb{R}^{d-k}$, we say that
a subset $\mathcal{M} \subset \Theta$ is a {\it mixture subfamily} 
generated by the constraint 
\begin{align}
\sum_{i=1}^d u^i_{k+j} \partial_i F(\theta) =a_j\Label{const1}
\end{align}
 for $j=1, \ldots,d-k$
when the subset $\mathcal{M}$ is written as
\begin{align}
    \mathcal{M} = \left\{ \theta  \in \Theta \left| \hbox{ Condition \eqref{const1} holds.} \right. \right\} .
\end{align}
The $d \times d $ matrix $U$ is defined as $(u_1 \ldots u_d )$.
To make a parametrization in the above mixture subfamily ${\cal M}$,
we set the new natural parameter $\bar{\theta}=(\bar{\theta}^1, \ldots, \bar{\theta}^d)$ as
$\theta=U \bar{\theta} $,
and introduce the new mixture parameter
\begin{align}
\bar{\eta}_i
=\partial_j (F \circ U) (\bar{\theta}).
\Label{Dif}
\end{align}
Since the relation $\bar{\eta}_{k+i}= a_i$ holds for $i=1, \ldots, d-k$ in ${\cal M}$,
the initial $k$ elements $\bar{\eta}_1, \ldots, \bar{\eta}_{k}$
give a parametrization for ${\cal M}$.
To make the parametrization, we define the map $\psi_{\cal M}^{(m)}$ 
as $\psi_{\cal M}^{(m)} ( U \bar{\theta}):= (\partial_j (F \circ U) (\bar{\theta}))_{j=1}^k$.
The set $\Xi_{{\cal M}}:= 
\{ \psi_{\cal M}^{(m)} (\theta) |  {\theta} \in {\cal M}\}$
works as the range of the new mixture parameters,
and we also employ the inverse map 
$\phi_{{\cal M}}^{(m)}:=(\psi_{{\cal M}}^{(m)})^{-1}: 
\Xi_{{\cal M}}\to {\cal M}$.
Since $\Theta$ is an open set, 
the set $\Xi_{{\cal M}}$ is an open subset of $\mathbb{R}^k$.
When an element $\bar{\eta} \in \Xi_{{\cal M}}$ satisfies 
$\bar{\eta}_j =\partial_j (F \circ U) (\bar{\theta})$ for $j=1, \ldots, k$,
we have
\begin{align}
\partial^i (F\circ U)^*( \bar{\eta},a) = \bar{\theta}^i\Label{CO1}
\end{align}
 for $i=1, \ldots, d$.
The strict convexity of the map $\bar{\eta} \mapsto(F\circ U)^*( \bar{\eta},a) $ guarantees that
the map $\bar{\eta} \mapsto (\partial^i (F\circ U)^*( \bar{\eta},a))_{i=1}^k$
is one-to-one. 
Hence, the initial $k$ elements $\bar{\theta}^1, \ldots, \bar{\theta}^{k}$
form a parametrization for ${\cal M}$.
In other words, the relation 
\begin{align}
( (U^{-1} \theta)^i )_{i=1}^k
= (\partial^i (F\circ U)^*( \psi_{\cal M}^{(m)} ( {\theta}) ,a))_{i=1}^k \Label{NAY}
\end{align}
holds.
We define the set $\Theta_{\cal M}:= \{ ( (U^{-1} \theta)^i )_{i=1}^k | {\theta} \in {\cal M} \}$, which is rewritten as
\begin{align}
\Theta_{\cal M}= 
\left\{(\theta^1, \ldots, \theta^k) \in \mathbb{R}^k \left|
\begin{array}{l}
\exists (\theta^{k+1}, \ldots, \theta^d) \in \mathbb{R}^{d-k} \hbox{ such that}  \\
\sum_{i=1}^d u^i_{k+j} \partial_i F( U(\theta^1, \ldots, \theta^d)) =a_j \\
\hbox{ for } j=1, \ldots,d-k.
\end{array}
\right.\right\}.
\Label{const1-U}
\end{align}

When the mixture subfamily ${\cal M}$ forms an exponential subfamily generated by $u_1, \ldots, u_k$,
it is possible to retake $\theta_0$ such that $(U^{-1} \theta_0)^i=0$ for $i=1, \ldots, k$.
Therefore, the subsets $\Theta_{{\cal M}}$ and $\Xi_{{\cal M}}$ are the same subsets defined in Subsection \ref{S2-1}.

\subsection{Bregman Divergence and $m$- and $e$- projections}
Next, we introduce the concept of  Bregman Divergence, which is a generalization of the conventional divergence.
\begin{definition}[Bregman Divergence]
We choose an open set $\Theta$ in $\mathbb{R}^d$ and 
a $C^\infty$-class strictly convex function $F:\Theta \rightarrow \mathbb{R}$.
We define the Bregman divergence $D^F$ as
\begin{equation}
    D^{F}(\theta_1 \| \theta_2):= 
    \langle \nabla^{(e)}[F](\theta_1), \theta_1 - \theta_2\rangle 
    - F(\theta_1)+F(\theta_2)      
    ~ (\theta_1, \theta_2 \in \Theta).
\Label{XZL}
\end{equation}
\end{definition}
Our Bregman divergence system is defined as the triplet $(\Theta,F,D^F)$.
Given a one-variable convex function $\mu(t)$, we have
\begin{align}
\mu'(\bar{t}) (\bar{t}-\tilde{t} )-\mu(\bar{t}) +\mu(\tilde{t}) 
=
\int^{\bar{t}}_{\tilde{t}} 
\mu''({t}) ({t}-\tilde{t} )dt.\Label{MNT2}
\end{align}
Now, we use the Hesse matrix $J_{i,j}(\theta):=
\frac{\partial^2 F}{\partial \theta^i \partial \theta^j}(\theta)$.
We substitute 
$F(\theta_2 +t(\theta_1 - \theta_2)$ into $\mu(t)$
in \eqref{MNT2} with $\bar{t}=1$ and $\tilde{t}=0$.
this quantity can be written as
\begin{equation}
    D^{F}(\theta_1 \| \theta_2)= 
\int_{0}^1 \sum_{i,j}(\theta_1^i - \theta_2^i)(\theta_1^j - \theta_2^j)
J_{i,j} 
(\theta_2 +t(\theta_1 - \theta_2))t dt.\Label{KPOT}
\end{equation}
In addition, since the relations \eqref{du1} and \eqref{MN1} imply
\begin{align}
F^*(\eta)
=\sum_{i=1}^d\theta^i \eta(\theta_i)-F(\theta)
=\langle  \eta(\theta),\theta\rangle -F(\theta),
\end{align}
the relations
\begin{align}
&    D^{F^*}(\nabla^{(e)} [F](\theta_2) \| \nabla^{(e)} [F](\theta_1))
=    D^{F^*}(\eta(\theta_2) \| \eta(\theta_1))\nonumber  \\
=& \langle \eta(\theta_2)-\eta(\theta_1), \theta_2\rangle -F^*(\eta(\theta_2))
+F^*(\eta(\theta_1))  \nonumber \\
=& \langle \eta(\theta_1),\theta_1-\theta_2\rangle -F(\theta_1)+F(\theta_2)   
=    D^{F}(\theta_1 \| \theta_2)
\Label{XI1}
\end{align}
hold.

In fact, when we restrict both inputs into elements of an exponential subfamily ${\cal E}$, 
the characterization
\begin{align}
D^{F}(\phi_{{\cal E}}^{(e)}(\bar{\theta}_1) \| \phi_{{\cal E}}^{(e)}(\bar{\theta}_2))
=
D^{F_{{\cal E}}}(\bar{\theta}_1 \| \bar{\theta}_2)\Label{NBSO}
\end{align}
holds for $\bar{\theta}_1 , \bar{\theta}_2 \in \Theta_{{\cal E}}$.
Therefore, the restriction of the Bregman divergence system $(\Theta,F,D^F)$ to ${\cal E}$
can be considered as the Bregman divergence system $(\Theta_{{\cal E}},F_{{\cal E}},D^{F_{\cal E}})$.
A simple calculation shows the following proposition.

\begin{proposition}[Pythagorean Theorem\cite{Amari-Nagaoka}]\Label{MNL}
Given a vector $(a_j)_{j=1}^l$,
we consider an
exponential subfamily $\mathcal{E} \subset \Theta$ generated by $l$ vectors $v_1,\ldots, v_l 
\in \mathbb{R}^d$ at $\theta_0 \in \Theta$, 
and 
a mixture subfamily $\mathcal{M} \subset \Theta$ 
generated by the constraint $\sum_{i=1}^d v^i_j \eta_i(\theta)=a_j$ for 
$j=1, \ldots, l$.
Assume that an intersection $\theta^*$ of $\mathcal{E}$ and $\mathcal{M}$ exists.
Any pair of $\theta \in \mathcal{E}$ and $\theta' \in \mathcal{M}$ satisfies 
\begin{align}
D^F(\theta\|\theta')=D^F(\theta\|\theta^*)+D^F(\theta^*\|\theta').\Label{AKO9}
\end{align}
\end{proposition}

\begin{lemma}\Label{LA1}
We consider an exponential family $\mathcal{E}$
generated by $l$ vectors $v_1,\ldots, v_l 
\in \mathbb{R}^d$.
The following conditions are equivalent for 
an exponential subfamily $\mathcal{E}$, $\theta^* \in \mathcal{E}$,
and $\theta_0 \in \Theta$.
\begin{description}
\item[(E0)]
The element $\theta^*\in \mathcal{E}$ achieves a local minimum
for the minimization $\min_{\hat{\theta} \in \mathcal{E} }  D^F (\theta_0 \| \hat{\theta})$.
\item[(E1)]
The element $\theta^*\in \mathcal{E}$ achieves the minimum value  
for the minimization $\min_{\hat{\theta} \in \mathcal{E} }  D^F (\theta_0 \| \hat{\theta})$.
\item[(E2)]
Let $\mathcal{M} \subset \Theta$ be the mixture subfamily 
generated by the constraint $\sum_{i=1}^d v^i_j \eta_i(\theta)=
\sum_{i=1}^d v^i_j \eta_i(\theta_0)$ for $j=1, \ldots, l$.
The element $\theta^*\in \mathcal{E}$ belongs to
the intersection 
$\mathcal{M}\cap \mathcal{E}$.
\end{description}
Further, when an element $\theta^*\in \mathcal{E}$ with the above condition exists,
it is unique. 
\end{lemma}

In the following, we denote the above mixture family $\mathcal{M}$ by $\mathcal{M}_{\theta_0\to \mathcal{E}}$.
Then, 
$\theta^* \in \mathcal{E}$ is called
the $m$-{\it projection} of $\theta$ onto an exponential subfamily $\mathcal{E}$,
and is denoted by $\Pro^{(m),F}_{\mathcal{E}} (\theta)$
because the points $\theta$ and $\theta^* $ are connected via the mixture family 
$\mathcal{M}_{\theta_0\to \mathcal{E}}$.
The minimum value $\min_{\hat{\theta} \in \mathcal{E} }  D^F (\theta \| \hat{\theta})$
is called the projected Bregman divergence
between $\theta$ and $\mathcal{E}$.

Exchanging the roles of the exponential family and the mixture family leads the following lemma.

\begin{lemma}\Label{LA2}
We choose $l$ vectors $v_1,\ldots, v_l \in \mathbb{R}^d$.
Let $\mathcal{M}$
be a mixture family generated by 
the constraint $\sum_{i=1}^d v^i_j \eta_i(\theta)=
\sum_{i=1}^d v^i_j \eta_i(\theta_0)$ for $j=1, \ldots, l$.
The following conditions are equivalent for 
the mixture family $\mathcal{M}$, $\theta^{\dagger} \in \mathcal{M}$,
and $\theta_0 \in \Theta$.
\begin{description}
\item[(M0)]
The element $\theta^{\dagger}\in \mathcal{M}$ achieves a local minimum
for the minimization $\min_{\hat{\theta} \in \mathcal{M} }  D^F ( \hat{\theta}\|\theta_0 )$.
\item[(M1)]
The element $\theta^{\dagger}\in \mathcal{M}$ achieves the minimum value  
for 
the minimization $\min_{\hat{\theta} \in \mathcal{M} }  D^F ( \hat{\theta}\|\theta_0 )$.
\item[(M2)]
Let $\mathcal{E} \subset \Theta$ be the mixture subfamily 
generated by $l$ vectors $v_1,\ldots, v_l \in \mathbb{R}^d$ at $\theta_0 \in \Theta$.
The element $\theta^{\dagger}\in \mathcal{M}$ belongs to
the intersection 
$\mathcal{M}\cap \mathcal{E}$.
\end{description}
Further, when there exists an element $\theta^{\dagger}\in \mathcal{M}$ to satisfy the above condition,
such an element is unique. 
\end{lemma}

In the following, 
the symbol $\mathcal{E}_{\theta_0\to \mathcal{M}}$ expresses the above exponential family $\mathcal{E}$.
Then, an element
$\theta^{\dagger} \in \mathcal{M}$ is called
the {\it $e$-projection} of $\theta$ onto a mixture subfamily $\mathcal{M}$,
and is denoted by $\Pro^{(e),F}_{\mathcal{M}} (\theta)$
because the points $\theta$ and $\theta^{\dagger} $ are connected via 
the exponential family $\mathcal{E}_{\theta_0\to \mathcal{M}}$.
When $ \mathcal{M}$ is an exponential subfamily and a mixture subfamily,
we can define both projections
$\Pro^{(m),F}_{\mathcal{M}}$ and $\Pro^{(e),F}_{\mathcal{M}}$, and these projections are different maps.
Hence, the subscripts $(e)$ and $(m)$ are needed.

\begin{lemma}\Label{Th5}
Consider an
exponential subfamily $\mathcal{E} \subset \Theta$ generated by $l$ vectors $v_1,\ldots, v_l 
\in \mathbb{R}^d$ at $\theta_0 \in \Theta$.
For $\theta_* \in \Theta$,
the element $\Pro^{(m),F}_{\mathcal{E}} (\theta_*)=\theta^* \in {\cal E}$ 
is uniquely characterized as
$\sum_{j=1}^d v_i^j \partial_j F(\theta^*)
= \sum_{j=1}^d v_i^j \partial_j F(\theta_*)$, i.e.,
$R[V] \circ \nabla^{(e)}[F](\theta^*)=R[V] \circ \nabla^{(e)}[F](\theta_*)$.
That is, 
the mixture parameter of the element $\Pro^{(m),F}_{\mathcal{E}} (\theta_*)=\theta^* \in {\cal E}$ 
is given by the above condition.
\end{lemma}

\begin{lemma}\Label{Th6}
Let $d$ vectors $u_1,\ldots, u_d \in \mathbb{R}^d$ be linearly independent.
We consider a mixture subfamily $\mathcal{M} \subset \Theta$  
generated by the constraint 
\begin{align}
\sum_{i=1}^d u^i_j \partial_i F (\theta)=a_j \Label{BO1}
\end{align}
for $j=k+1, \ldots, d$.
For an element $\theta_{\dagger}\in \Theta $,
the existence of the maximum $\max_{\theta\in {\cal M}} D^F(\theta\|\theta_{\dagger})$ yields
the following characterizations for $\Pro^{(e),F}_{\mathcal{M}} (\theta_{\dagger})$.
\begin{description}
\item[(C1)]
The point $\Pro^{(e),F}_{\mathcal{M}} (\theta_{\dagger})=\theta^{\dagger} \in {\cal M}$ 
is uniquely characterized as
\begin{align}
(U^{-1}{\theta}^{\dagger})^i
=(U^{-1}{\theta}_{\dagger})^i\Label{MFA2}
\end{align}
 for $i=1, \ldots, k$,
 where $U$ is defined in the same way as Subsection \ref{S2-2}.
\item[(C2)]
We choose the exponential subfamily ${\cal E}$ 
generated by $d-k$ vectors $u_{k+1},\ldots, u_{d} \in \mathbb{R}^{d}$
at $\theta_{\dagger}$.
The intersection between
${\cal M}$  and ${\cal E}$ is composed of 
the unique element $\Pro^{(e),F}_{\mathcal{M}} (\theta_{\dagger})$. 

\item[(C3)]
The point $\Pro^{(e),F}_{\mathcal{M}} (\theta_{\dagger})=\theta^{\dagger} \in {\cal M}$ 
is uniquely characterized as
$\theta_{\dagger}+ \sum_{j'=1}^{d-k} \bar{\tau}^{j'} u_{k+j'}$,
where $(\bar{\tau}^1, \ldots, \bar{\tau}^{d-k})$ is the unique element to satisfy 
\begin{align}
\frac{\partial}{\partial \tau^{j}} F \Big(\theta_{*}+ \sum_{j'=1}^{l} \tau^{j'} u_{k+j'} \Big) =a_j
\Label{const1-T}
\end{align}
for $j=1, \ldots, d-k$.
\end{description}
\end{lemma}

Lemmas \ref{LA1} and \ref{LA2} show the importance to
find a sufficient condition for (E2) and (M2).
To seek such a condition with a convex function $F$ and $\Theta$,
we consider the following conditions with
$l$ linearly independent fixed vectors $v_1,\ldots, v_l \in \mathbb{R}^d$;
\begin{description}
\item[(M3)]
We denote the exponential family generated by 
the $l$ linearly independent vectors $v_1,\ldots, v_l \in \mathbb{R}^d$ at $\theta_0 \in \Theta$ by $\mathcal{E}(\theta_0) $.
The set $\Xi_{\mathcal{E}(\theta_0)}$ does not depend on $ \theta_0 \in \Theta$.
In this case, this set is denoted by $\Xi(v_1,\ldots, v_l)$.
Notice that the set $\Xi_{\mathcal{E}(\theta_0)}$ is defined after \eqref{CXO}.
\item[(E3)]
We denote the mixture family generated by the constraint
$\sum_{i=1}^d v^i_{j} \partial_i F(\theta) =a_j
$ for $j=1, \ldots,l$
by $\mathcal{M}(a_1, \ldots, a_l)$.
When the set $\Theta_{\mathcal{M}(a_1, \ldots, a_l)}$ is defined 
in the way as \eqref{const1-U},
it does not depend on $ (a_1, \ldots, a_l) \in \mathbb{R}^l$ unless
$\mathcal{M}(a_1, \ldots, a_l)$ is empty.
In this case, this set is denoted by $\Theta(v_1,\ldots, v_l)$.
\end{description}

Under the above condition, we have the following lemmas. 
\begin{lemma}\Label{Lem7}
Suppose that 
the $l$ linearly independent vectors $v_1,\ldots, v_l \in \mathbb{R}^d$
satisfy Condition (M3).
Given $(a_1, \ldots, a_l)\in \Xi(v_1,\ldots, v_l)$, the mixture family 
$\mathcal{M}(a_1, \ldots, a_l)$ is defined by using the condition \eqref{BO1}.
Then, for $\theta_0 \in \Theta$,
the projected point $\Pro^{(e),F}_{\mathcal{M}(a_1, \ldots, a_l)} (\theta_{0}) $ exists.
\end{lemma}

\begin{lemma}\Label{Lem8}
Suppose that 
the $l$ linearly independent vectors $v_1,\ldots, v_l \in \mathbb{R}^d$
satisfy Condition (E3).
Then, for $ (b^1, \ldots, b^{d-l}) \in \mathbb{R}^{d-l}$ and
$\theta_0 \in \Theta$,
the projected point $\Pro^{(m),F}_{\mathcal{E}(b^1, \ldots, b^{d-l})} (\theta_{0}) $ exists
unless $\mathcal{E}(b^1, \ldots, b^{d-l})$ is empty
where the exponential family $\mathcal{E}(b^1, \ldots, b^{d-l})$ is defined
as $\{  (\sum_{i=1}^{d-l} u_i^j b^i+ \sum_{i=1}^{l} u_i^j \theta^i)_{j=1}^d |
(\theta^1,\ldots,\theta^{l}) \in \mathbb{R}^{l} 
\}\cap \Theta$. 
\end{lemma}

Therefore, to consider the existence of both types of projections universally, we introduce the following conditions for 
the Bregman divergence system $(\Theta,F,D^F)$.
\begin{description}
\item[(M4)]
Any $l$ linearly independent vectors $v_1,\ldots, v_l \in \mathbb{R}^d$
satisfy the condition (M3) for $l=1, \ldots, d-1$.
\item[(E4)]
Any $l$ linearly independent vectors $v_1,\ldots, v_l \in \mathbb{R}^d$
satisfy the condition (E3) for $l=1, \ldots, d-1$.
\end{description}
When (M4) holds, the $e$-projection 
$\Pro^{(e),F}_{\mathcal{M}}$ 
can be defined for any mixture subfamily $\mathcal{M}$.
Also, 
when (E4) holds, the $m$-projection 
$\Pro^{(m),F}_{\mathcal{E}}$ 
can be defined for any exponential subfamily $\mathcal{E}$.
Therefore, these two conditions are helpful for the analysis of these projections.

\begin{table*}[htb]
\begin{center}
\caption{Summary of dimensions}
\begin{tabular}{|c|c|} \hline
Symbol & Space \\
\hline 
$d$ & Dimension of the whole space \\ 
\hline 
$l$ & Dimension of Exponential family ${\cal E}$ \\
\hline 
$k$ & Dimension of Mixture family ${\cal M}$ \\
\hline 
\end{tabular}
\Label{table1}
\end{center}
\end{table*}

\subsection{Evaluation of Bregman divergence without Pythagorean theorem}
Next, we evaluate Bregman divergence when we cannot use the Pythagorean theorem.
For this aim, we focus on $J(\theta)^{-1}$, i.e.,
the inverse of the Hesse matrix $J(\theta$) defined for the parameters of $\Theta$.
Then, we introduce the quantity
$\gamma(\hat{\Theta} |{\Theta})$
for a subset $\hat{\Theta}$ of $\Theta$.
\begin{align}
\gamma(\hat{\Theta}|{\Theta}):= &
\inf\{\gamma|
\gamma J(\theta_1)^{-1}\ge J(\theta_2)^{-1}
\hbox{ for }\theta_1,\theta_2 \in \hat{\Theta} \} .
\end{align}
We say that a subset $\hat{\Theta}$ of $\Theta$
is a {\it star subset} for an element $\theta_1 \in \hat{\Theta}$
when 
$\lambda \eta(\theta)+(1-\lambda)\eta(\theta_1) \in \eta(\hat{\Theta})$ for 
$\theta \in \hat{\Theta}$ and $\lambda \in (0,1)$.

Then, we have the following theorem.
\begin{theorem}\Label{XAM} 
We assume that the condition (M4) holds.
Then, for 
a star subset with $\hat{\Theta}$ for $ \theta_1 \in \hat{\Theta}$,
$\theta_2 \in \hat{\Theta}$, and $\theta_3 \in \Theta$, we have
\begin{align}
& D^F(\theta_1\|\theta_2)
\nonumber  \\
\le & 
D^F(\theta_1\|\theta_3)+\gamma (\hat{\Theta}|{\Theta}) D^F(\theta_2\|\theta_3)+ 
2 \gamma (\hat{\Theta}|{\Theta})\sqrt{ D^F(\theta_1\|\theta_3)D^F(\theta_2\|\theta_3)}.
\Label{BLT}
\end{align}
\end{theorem} 

\section{Examples of Bregman divergence}\Label{S3}
\subsection{Classical system}\Label{4A}
We consider the set of probability distributions on the finite set ${\cal X}=\{1, \ldots, n\}$.
We focus on $d$ linearly independent functions $f_1, \ldots, f_{d}$ defined on ${\cal X}$,
where the linear space spanned by $f_1, \ldots, f_{d}$ does not contain a constant function and $d \le n-1$.
Then, the $C^{\infty}$ strictly convex function 
$\mu$ on $\mathbb{R}^{d}$ is defined as
$\mu(\theta):= 
\log \big(\sum_{x \in {\cal X}} \exp (  \sum_{j=1}^{d} \theta^j f_j(x)  )\big)$, 
which yields
the Bregman divergence system $(\mathbb{R}^{d}, \mu, D^\mu)$. 
When $d=n-1$,
any probability distribution with full support on ${\cal X}$
can be written as $P_\theta$, which is defined as
$P_\theta(x):= 
\exp \Big(  (\sum_{j=1}^{n-1} \theta^j f_j(x)  ) - \mu(\theta) \Big)$.
It is known that the KL divergence 
equals the Bregman divergence of the potential function $\mu$ \cite[Section 3.4]{Amari-Nagaoka}, 
i.e., we have
\begin{align}
D^{\mu}(\theta\|\theta')= D(P_\theta\|P_{\theta'})
\Label{MGA}
\end{align}
for $\theta \in \mathbb{R}^{d}$, where
the KL divergence $D(q\| p)$ is defined as
\begin{align}
D(q\| p)=\sum_{\omega} p(\omega)
(\log p(\omega) - \log q(\omega)).
\end{align}

When the parameter $\theta$ is limited to $(\bar\theta,\underbrace{0, \ldots, 0}_{d-l})$ with 
$\bar\theta \in \mathbb{R}^l$,
the set of distributions $P_\theta$ forms an exponential subfamily.
Also, when the linear space spanned by $d-k$ linearly independent functions $g_1, \ldots, g_{d-k}$
does not contain a constant function,
for $d-k$ constants $a_1, \ldots, a_{d-k}$, 
the following set of distributions forms a mixture subfamily; 
\begin{align}
\Big\{P_\theta \Big|
\sum_{x \in {\cal X}}g_i(x) P_\theta(x)=a_i \hbox{ for } i=1, \ldots, d-k \Big\}.
\end{align}

\begin{example}
When $\X$ is given as $\X_1 \times \X_2$ with $n_i=|\X_i|$,
the set of distributions with full support on $\X$ forms 
a Bregman divergence system $(\mathbb{R}^{d}, \mu, D^\mu)$. 
When
$f_i$ is a function on $\X_1$ or $\X_2$ with  $i=1, \ldots, n_1+n_2-2$,
and they are linearly independent,
the exponential subfamily generated by $f_1, \ldots,  f_{n_1+n_2-2}$ forms 
the set ${\cal P}_{\X_1}\times {\cal P}_{\X_2}$ of independent distributions on  $\X_1 \times \X_2$.
\end{example}

\begin{example}
When $\X$ is given as $\X_1 \times \X_2\times \X_3$ with $n_i=|\X_i|$,
the set of distributions with full support on $\X$ forms 
a Bregman divergence system $(\mathbb{R}^{d}, \mu, D^\mu)$. 
When 
$f_i$ is a function on $\X_1,\X_2$ or $\X_2,\X_3$ with $i=1, \ldots, n_2 (n_1+n_3-1)-1$,
and they are linearly independent,
the exponential subfamily generated by 
$f_1, \ldots, f_{n_2 (n_1+n_3-1)-1}$ forms the set  
${\cal P}_{X_1-X_2-X_3}$ of distributions on 
$\X_1 \times \X_2 \times X_3$ to satisfy the Markovian condition $X_1-X_2-X_3$.
\end{example}

For the possibility of the projection, we have the following lemma.
For its proof, see \cite{em-only}.
\begin{lemma}\Label{LOS}
The Bregman divergence system $(\mathbb{R}^{d}, \mu, D^\mu)$ defined in this subsection satisfies
the conditions (E4) and (M4).
\end{lemma}

\subsection{Quantum system}\Label{4B}
In the quantum system, we focus on the $n$-dimensional Hilbert space ${\cal H}$ \cite{hayashi}.
We choose $d$ linearly independent Hermitian matrices $X_1, \ldots, X_{d}$ on ${\cal H}$,
where the linear space spanned by $X_1, \ldots, X_{d}$ does not contain the identity matrix.
Then, we define the $C^{\infty}$ strictly convex function 
$\mu$ on $\mathbb{R}^{d}$ as
$\mu(\theta):= \log (\Tr \exp (  \sum_{j=1}^{d} \theta^j X_j ) )$.
A quantum state on ${\cal H}$ is given as
a positive semi-definite Hermitian matrix $\rho$ with the condition $\Tr \rho=1$, which is called 
a density matrix.
We denote the set of density matrices by ${\cal S}({\cal H})$.
Any density matrix with full support on ${\cal H}$
can be written as $\rho_\theta$, which is defined as
$\rho_\theta := 
\exp \Big(  (\sum_{j=1}^{d} \theta^j X_j  ) - \mu(\theta) \Big)$.
It is known that the relative entropy 
equals the Bregman divergence of the potential function $\mu$ \cite[Section 7.2]{Amari-Nagaoka}, 
i.e., we have
\begin{align}
D^{\mu}(\theta\|\theta')= D(\rho_\theta\|\rho_{\theta'})
\Label{MGA2}
\end{align}
for $\theta \in \mathbb{R}^{d}$, where
the relative entropy $D(\rho\| \rho')$ is defined as
\begin{align}
D(\rho\| \rho')=\Tr \rho (\log \rho - \log \rho').
\end{align}

When the parameter $\theta$ is limited to 
$(\bar\theta,\underbrace{0, \ldots, 0}_{d-l})$ with 
$\bar\theta \in \mathbb{R}^l$,
the set of distributions $\rho_\theta$ forms an exponential family.
Also, when the linear space spanned by $d-k$ linearly independent Hermitian matrices
$Y_1, \ldots, Y_{d-k}$
does not contain a constant function,
for $d-k$ constants $a_1, \ldots, a_{d-k}$, 
the following set of distributions forms a mixture family; 
\begin{align}
\Big\{\rho_\theta \Big|
\Tr Y_i \rho_\theta=a_i \hbox{ for } i=1, \ldots, d-k \Big\}.
\end{align}

For the possibility of the projection, we have the following lemma.
For its proof, see \cite{em-only}.
\begin{lemma}\Label{LOS3}
The Bregman divergence system $(\mathbb{R}^{d}, \mu, D^\mu)$ defined in this section satisfies
the conditions (E4) and (M4).
\end{lemma}

\section{Reverse em-problem}\Label{Sec:BBem}
\subsection{General formulation}\Label{S4-1}
In this section, we address a maximization problem for
a pair of a $k$-dimensional mixture subfamily $\mathcal{M}$ and an $l$-dimensional  exponential subfamily $\mathcal{E}$.
Similar to Section IV of \cite{em-only}, 
we assume the following condition;
\begin{description}
\item[(B1)]
The Bregman divergence system $(\Theta,F,D^F)$
satisfies the conditions (E4) and (M4). 
\end{description}
The meaning of (B1) is clear.
In the general setting of Bregman, $m$- and $e$- projections do not necessarily exist.
To guarantee their existence, we assume condition (B1), which is satisfied
when they are given as probability distributions or density operators.

Hence, the minimum $\min_{\theta' \in \mathcal{E}} D^{F}(\theta \| \theta')$ 
exists.
As discussed in Section IV of \cite{em-only}, 
the em-algorithm is a method to minimize 
the divergence between 
two points in the mixture and exponential subfamilies ${\cal E}$ and ${\cal M}$, which
is formulated as the following minimization under the framework of Bregman divergence system:
\begin{equation}\Label{eq:em}
C_{\inf}(\mathcal{M},\mathcal{E})
:=
\inf_{\theta \in \mathcal{M}} 
D^{F}(\theta \| \Pro^{(m),F}_{\mathcal{E}} (\theta))
=
\inf_{\theta \in \mathcal{M}} 
\min_{\theta' \in \mathcal{E}} 
D^{F}(\theta \| \theta').
\end{equation}
For this problem,
the em-algorithm, Algorithm \ref{protocol1-0}, is known.

\begin{algorithm}
\caption{em-algorithm}
\Label{protocol1-0}
\begin{algorithmic}
\STATE {
Choose the initial value ${\theta}_{(1)} \in \mathcal{E}$;} 
\REPEAT 
\STATE {\bf m-step:}\quad 
Calculate ${\theta}^{(t+1)}:=\Pro^{(e),F}_{{\cal M}} ({\theta}_{(t)})$.
That is, ${\theta}^{(t+1)}$ is given as 
$\argmin_{\theta \in \mathcal{M}} D^{F}({\theta} \| {\theta}_{(t)})$, i.e.,
the unique element in ${\cal M}$ to realize the minimum
of the smooth convex function $\theta \mapsto D^{F}(\theta\| {\theta}_{(t)})$.
\STATE {\bf e-step:}\quad 
Calculate ${\theta}_{(t+1)}:=\Pro^{(m),F}_{{\cal E}} (\theta^{(t+1)})$.
That is, ${\theta}_{(t+1)}$ is given as 
$\argmin_{\theta' \in \mathcal{E}} D^{F}({\theta}^{(t+1)} \| \theta')$, i.e.,
the unique element in ${\cal E}$ to realize the minimum
of the smooth convex function $\theta'\mapsto D^{F}({\theta}^{(t+1)} \| \theta')$.
\UNTIL{convergence.} 
\end{algorithmic}
\end{algorithm}

Instead of the em-problem \eqref{eq:em}, 
we address the following maximization problem for
a pair of a mixture subfamily $\mathcal{M}$ and an exponential subfamily $\mathcal{E}$; 
\begin{equation}\Label{eq:Inf.st}
C_{\sup}(\mathcal{M},\mathcal{E})
:=
\sup_{\theta \in \mathcal{M}} D^{F}(\theta \| \Pro^{(m),F}_{\mathcal{E}} (\theta))
=
\sup_{\theta \in \mathcal{M}} 
\min_{\theta' \in \mathcal{E}} 
D^{F}(\theta \| \theta').
\end{equation}
Also, we need to characterize the following set;
\begin{equation}\Label{maxset}
\Theta^*(\mathcal{M},\mathcal{E})
:=
\{\theta \in \mathcal{M}|
C_{\sup}(\mathcal{M},\mathcal{E})
=D^{F}(\theta \| \Pro^{(m),F}_{\mathcal{E}} (\theta))\}.
\end{equation}
When the above set is not empty and is composed of a unique element,
we need to find the maximization point
\begin{equation}\Label{eq:Inf-2}
\theta^*(\mathcal{M},\mathcal{E})
:=
\argmax_{\theta \in \mathcal{M}} D^{F}(\theta \| \Pro^{(m),F}_{\mathcal{E}} (\theta)).
\end{equation}
Some of maximization problems in information theory can be written in the above form.
The above maximization asks to maximize the divergence between 
two points in the mixture and exponential subfamilies ${\cal E}$ and ${\cal M}$.
Hence, as pointed out in Toyota \cite{Shoji}, we can expect that
the reverse operation of the em-algorithm gives the solution of the maximization given in \eqref{eq:Inf.st}, which is illustrated in Fig. \ref{rev-fig}.
Since the minimum $\min_{\theta \in \mathcal{M}} D^{F}(\theta \| \theta')$ exists due to the condition (B1),
the em-algorithm repetitively applies the function $\Pro^{(e),F}_{{\cal M}}\circ \Pro^{(m),F}_{{\cal E}}|_{{\cal M}}$
for an element $\theta \in {\cal M}$.
Therefore, when the function $\Pro^{(e),F}_{{\cal M}}\circ \Pro^{(m),F}_{{\cal E}}|_{{\cal M}}$
is a surjective map from ${\cal M}$ to ${\cal M}$,
there exists its inverse map $(\Pro^{(e),F}_{{\cal M}}\circ \Pro^{(m),F}_{{\cal E}}|_{{\cal M}})^{-1}$.
Since the application of $\Pro^{(e),F}_{{\cal M}}\circ \Pro^{(m),F}_{{\cal E}}|_{{\cal M}}$
monotonically decreases the Bregman divergence, the application of the inverse map increases 
the Bregman divergence
\begin{align}
 & D^{F}(\theta \| \Pro^{(m),F}_{{\cal E}}  (\theta) ) 
\nonumber \\
 \leq & 
 D^{F}((\Pro^{(e),F}_{{\cal M}}\circ \Pro^{(m),F}_{{\cal E}}|_{{\cal M}})^{-1}(\theta)\| 
 \Pro^{(m),F}_{{\cal E}}  ((\Pro^{(e),F}_{{\cal M}}\circ \Pro^{(m),F}_{{\cal E}}|_{{\cal M}})^{-1} (\theta)) ) .
\end{align}
In this case, 
when we apply the updating rule ${\theta}^{(t+1)}:= 
(\Pro^{(e),F}_{{\cal M}}\circ \Pro^{(m),F}_{{\cal E}}|_{{\cal M}})^{-1}({\theta}^{(t)})$,
it is expected that the outcome ${\theta}^{(t)}$ of the repetitive application of the inverse map converges to 
$\theta^*(\mathcal{M},\mathcal{E})$.
Due to the above reason, we call the maximization \eqref{eq:Inf.st}
the reverse em-problem.

\subsection{Precision analysis}
For the analysis of the precision, we introduce the following condition for $\mathcal{M}$ and $\mathcal{E}$.
\begin{description}
\item[(B2)]
The relation \begin{align}
 D^{F}(\theta'\|\theta)
\le
 D^{F}( \Pro^{(m),F}_{{\cal E}}  (\theta')\|  \Pro^{(m),F}_{{\cal E}}  (\theta) ) \Label{MLA2Y}
\end{align}
holds for any $\theta,\theta' \in{\cal M} $.
\end{description}
For example, Condition (B2) holds in the case of classical and quantum channel coding, as explained later.
That is, when the exponential family ${\cal E}$ is given as the product of two exponential families
${\cal E}_1$ and ${\cal E}_2$,
and there is a data processing between the mixture family ${\cal M}$ and the exponential family ${\cal E}_1$,
Condition (B2) is satisfied.

In the following, 
we restrict the domain of $e$- and $m$- projections 
into ${\cal M}$ and ${\cal E}$.
We use the notations:
\begin{align}
\Pro^{(e),F}_{{\cal E}\to {\cal M}}  
:=\Pro^{(e),F}_{{\cal E}\to {\cal M}}  |_{{\cal E}}, \quad
\Pro^{(m),F}_{{\cal M}\to {\cal E}}  
:=\Pro^{(m),F}_{{\cal M}\to {\cal E}}  |_{{\cal M}}.
\end{align}
Then, we have the following theorem, which is proven in Appendix \ref{A1}. 
\begin{theorem}\Label{theo:conv:BBem}
Assume that the conditions (B1) and (B2) hold,
the initial point $\theta^{(1)} \in {\cal M} $ satisfies
the relation $\sup_{\theta \in \mathcal{M}} D^F(\theta \| \theta^{(1)})<\infty$, 
and its inverse map $(\Pro^{(e),F}_{{\cal M}}\circ \Pro^{(m),F}_{
{\cal M}\to {\cal E}})^{-1}$ exists.
Then, the quantity $ D^{F}\big(\theta^{(t)}\big\| \Pro^{(m),F}_{{\cal E}}  (\theta^{(t)}) \big) 
$ converges to the supremum $C_{\sup}(\mathcal{M},\mathcal{E})$ 
with the speed 
\begin{align}
C_{\sup}(\mathcal{M},\mathcal{E})- D^{F}\big(\theta^{(t)}\big\| \Pro^{(m),F}_{{\cal E}}  (\theta^{(t)}) \big) 
=o(\frac{1}{t}).\Label{mma}
\end{align}
That is, the convergence point achieves the maximum in \eqref{eq:Inf.st}.
Further, 
when 
$t \ge \frac{ \sup_{\theta \in \mathcal{M}}D^F(\theta \| \theta^{(1)})}{\epsilon}$,
the parameter $ \theta^{(t)}$ satisfies
\begin{align}
C_{\sup}(\mathcal{M},\mathcal{E})- D^{F}\big(\theta^{(t)}\big\| \Pro^{(m),F}_{{\cal E}}  (\theta^{(t)}) \big) 
\le \epsilon.
\Label{NHG}
\end{align}
\end{theorem}

\begin{lemma}
When the set $\Theta^*(\mathcal{M},\mathcal{E})$ is not empty, 
it is a mixture subfamily.
\end{lemma}

As a strengthened version of (B2), we introduce the following condition 
for $\mathcal{M}$, $\mathcal{E}$, and $\theta'\in \mathcal{M}$;
\begin{description}
\item[(B2+)]
The maximizer $\theta^*=\theta^*(\mathcal{M},\mathcal{E})$  exists.
There exists a constant $\alpha(\theta')>0$ such that
the relation 
\begin{align}
 (1+\alpha(\theta')) D^{F}(\theta^*\|\theta)
\le
 D^{F}( \Pro^{(m),F}_{{\cal E}}  (\theta^*)\|  \Pro^{(m),F}_{{\cal E}}  (\theta) ) \Label{MLA2Y+}
\end{align}
holds when an element $\theta \in{\cal M} $ satisfies the condition
$ D^{F}(\theta^*\|\theta)\le D^{F}(\theta^*\|\theta')$.
\end{description}

When the condition (B2+) holds, 
we have a better evaluation.
\begin{theorem}\Label{theo:conv:BBem+}
Assume that the conditions (B1) and (B2+) hold
for $\mathcal{M}$, $\mathcal{E}$, and $\theta'\in \mathcal{M}$,
and
there exists its inverse map 
$(\Pro^{(e),F}_{{\cal M}}\circ \Pro^{(m),F}_{{\cal M}\to {\cal E}}
)^{-1}$.
Then,
the quantity $ D^{F}\big(\theta^{(t)}\big\| \Pro^{(m),F}_{{\cal E}}  (\theta^{(t)}) \big) 
$ converges to the supremum $C_{\sup}(\mathcal{M},\mathcal{E})$ 
with the speed 
\begin{align}
C_{\sup}(\mathcal{M},\mathcal{E})- D^{F}\big(\theta^{(t)}\big\| \Pro^{(m),F}_{{\cal E}}  (\theta^{(t)}) \big) 
\le 
(1+\alpha(\theta^{(1)}))^{-t+1}
D^F(\theta^* \| \theta^{(1)}).
\Label{mma+}
\end{align}
Further, 
when 
$t-1 \ge \frac{ \log  D^F(\theta^* \| \theta^{(1)}) -\log \epsilon}{\log (1+\alpha(\theta^{(1)}))}$,
the parameter $ \theta^{(t)}$ satisfies
\begin{align}
C_{\sup}(\mathcal{M},\mathcal{E})- D^{F}\big(\theta^{(t)}\big\| \Pro^{(m),F}_{{\cal E}}  (\theta^{(t)}) \big) 
\le \epsilon.
\Label{NHG+}
\end{align}
\end{theorem}
Theorem \ref{theo:conv:BBem+} is proven in Appendix \ref{A2}. 
Here, we consider the case 
$\alpha(\theta^{(1)})$ can be chosen as a non-negligible value
when $\theta^{(1)}$ is close to $\theta^*(\mathcal{M},\mathcal{E})$.
In this case, the convergence speed in \eqref{mma+} increases
when $t$ is larger.

\subsection{Algorithm based on minimization under mixture parameter}
In the rest of this paper, 
we use the subscript $\saa,\sbb,\scc,\sdd,\see$ to express elements
of $\mathbb{R}^{k}$, $\mathbb{R}^{l-k}$,
$\mathbb{R}^{l}$, $\mathbb{R}^{d-k}$, $\mathbb{R}^{l-2k}$, respectively, as Table \ref{tableB}.

\begin{table*}[htb]
\begin{center}
\caption{Summary of subscripts}
\begin{tabular}{|c|c|c|c|c|c|} \hline
Subscript &  $\saa$&$\sbb$&$\scc$&$\sdd$&$\see$\\
\hline 
Vector space &  $\mathbb{R}^{k}$& $\mathbb{R}^{l-k}$&
$\mathbb{R}^{l}$& $\mathbb{R}^{d-k}$& 
$\mathbb{R}^{l-2k}$ \\ 
\hline 
\multirow{2}{*}{Examples} &
$\Theta_{{\cal M}}$ &
\multirow{2}{*}{$\Theta_{{\cal E},\sbb}, \Theta_{{\cal E},\sbb}$}
& \multirow{2}{*}{$\Theta_{{\cal E}}$} &&\\
&
$\Theta_{{\cal E},\saa}, \Theta_{{\cal E},\saa}$
& &&&\\
\hline 
\end{tabular}
\par
\vspace{1ex}
\begin{flushleft}
In this paper, there are many types of vector spaces. 
An element of each vector space has a subscript to identify the vector space.
This table shows the relation between the vector space and the subscript.
\end{flushleft}
\Label{tableB}
\end{center}
\end{table*}

To handle these maps, we employ natural parameters and mixture 
parameters.
We use the following notions.
\begin{align}
\Pro^{(e),F}_{{\cal E}\to \Xi_{\cal M}}
&:=
(\psi_{{\cal M}}^{(m)})^{-1} \circ
\Pro^{(e),F}_{{\cal E}\to {\cal M}}  \\
\Pro^{(e),F}_{{\cal E}\to \Theta_{\cal M}}
&:=
(\psi_{{\cal M}}^{(e)})^{-1} \circ
\Pro^{(e),F}_{{\cal E}\to {\cal M}}  \\
\Pro^{(e),F}_{\Xi_{\cal E}\to {\cal M}}
&:=
\Pro^{(e),F}_{{\cal E}\to {\cal M}} \circ \psi_{{\cal E}}^{(m)} \\
\Pro^{(e),F}_{\Xi_{\cal E}\to \Xi_{\cal M}}
&:=
(\psi_{{\cal M}}^{(m)})^{-1} \circ
\Pro^{(e),F}_{{\cal E}\to {\cal M}} \circ \psi_{{\cal E}}^{(m)} \\
\Pro^{(e),F}_{\Xi_{\cal E}\to \Theta_{\cal M}}
&:=
(\psi_{{\cal M}}^{(e)})^{-1} \circ
\Pro^{(e),F}_{{\cal E}\to {\cal M}} \circ \psi_{{\cal E}}^{(m)} .
\end{align}
In the same way, we define the maps
$\Pro^{(e),F}_{\Theta_{\cal E}\to {\cal M}}$,
$\Pro^{(e),F}_{\Theta_{\cal E}\to \Xi_{\cal M}}$,
$\Pro^{(e),F}_{\Theta_{\cal E}\to \Theta_{\cal M}}$,
$\Pro^{(m),F}_{{\cal M}\to \Xi_{\cal E}}$,
$\Pro^{(m),F}_{{\cal M}\to \Theta_{\cal E}}$,
$\Pro^{(m),F}_{\Xi_{\cal M}\to {\cal E}}$,
$\Pro^{(m),F}_{\Xi_{\cal M}\to \Xi_{\cal E}}$,
$\Pro^{(m),F}_{\Xi_{\cal M}\to \Theta_{\cal E}}$,
$\Pro^{(m),F}_{\Theta_{\cal M}\to {\cal E}}$,
$\Pro^{(m),F}_{\Theta_{\cal M}\to \Xi_{\cal E}}$,
$\Pro^{(m),F}_{\Theta_{\cal M}\to \Theta_{\cal E}}$.

To characterize $e$- and $m$-projections, we introduce the following condition, which is also useful for 
the characterization of the inverse map of the map $\Pro^{(e),F}_{{\cal M}}\circ \Pro^{(m),F}_{{\cal M}\to{\cal E}}$.

\begin{description}
\item[(B3)]
Let $u_1, \ldots, u_d $ be a basis of $\mathbb{R}^d$.
$v_1, \ldots, v_{l}$ are linearly independent vectors in  $\mathbb{R}^d$.
Let $\mathcal{E} \subset \Theta$ be an
exponential subfamily generated by $l$ vectors 
$v_1,\ldots, v_{l} \in \mathbb{R}^d$ at $\theta_0 \in \Theta$, 
and 
$\mathcal{M} \subset \Theta$ be the mixture subfamily 
generated by the constraint $\sum_{i=1}^d u^i_{k+j} \partial_i F(\theta)=0$ 
for $j=1, \ldots,d- k$ with $k \le l$.
Also, $\mathcal{M} \subset \Theta$ is an exponential subfamily generated by
$u_1,\ldots, u_{k} \in \mathbb{R}^d$. 
That is, there exists $\theta_{\sdd}^*=
(\theta^{k+1,*}, \ldots, \theta^{d,*})$ such that
$\mathcal{M}=\{(\theta_{\saa}, \theta_{\sdd}^*)| \theta \in \mathbb{R}^k\} 
\cap \Theta$.
\end{description}

When the condition (B3) holds, 
for $\theta \in \mathcal{M}$, we denote its natural parameter and its mixture parameter
by 
$\hat{\theta}_{\saa}(\theta)\in \Theta_{\mathcal{M}}$ and 
$\hat{\eta}_{\saa}(\theta)\in \Xi_{\mathcal{M}}$, respectively.
Therefore, we use the notation $\hat{\theta}_{\saa}^{(t)}
:=\hat{\theta}_{\saa}(\theta^{(t)})
\in \Theta_{\mathcal{M}}$ to identify an element in 
$\mathcal{M}$ instead of $\theta^{(t)}$.
Then, we define 
the $d\times d$ matrix $U$ and the $d\times l$ matrix $V$ as
$U=(u_1, \ldots, u_d)$ and 
$V=(v_1, \ldots, v_l)$, and define the $k \times l$ matrix $V_1$ and 
the $(d-k) \times l$ matrix $V_2$ as
$ \left(
\begin{array}{c}
V_1 \\
V_2 
\end{array}
\right) =V U^{-1}$.
Condition (B3) brings the following useful characterization of 
$e$- and $m$-projections.

\begin{lemma}\Label{BXPR}
Assume Condition (B3).
For $\eta_\saa \in \Xi_{\mathcal{M}} $,  we have
\begin{align}
\Pro^{(m),F}_{\Xi_{{\cal M}}\to \Xi_{\cal E}}  
(\eta_\saa)= 
\eta_\saa V_1 .\Label{MAF}
\end{align}
\end{lemma}
\begin{lemma}\Label{BXP}
Assume Condition (B3).
The following conditions for elements 
$\eta_\saa \in \Xi_{\mathcal{M}} $ and 
$\theta_\saa \in \Theta_{\mathcal{M}} $
are equivalent.
\begin{description}
\item[(i)]
$\Pro^{(e),F}_{\Xi_{{\cal E}}\to \Theta_{\cal M}}  
 (\eta_\saa V_1)= 
\theta_\saa$.
\item[(ii)]
The element $\psi_{\mathcal{E}}^{(m)} (\eta_\saa V_1) 
\in \mathcal{E} $ 
belongs to 
the exponential subfamily $
\psi_{\mathcal{E}}^{(e)} (
\{ \theta_\scc \in \Theta_\mathcal{E}
| \theta_\saa=V_1 \theta_\scc\})$.
\item[(iii)]
 The following relation holds.
\begin{align}
\theta_\saa
=V_1 \nabla^{(m)}[F_{\mathcal{E}}^*]( \eta_\saa V_1)  
=\nabla^{(m)}[ F_{\mathcal{E}}^* \circ R[V_1 ]]( \eta_\saa ).\Label{NM7}
\end{align}
Here, the second equation always holds.
\item[(iv)]
 The following relation holds.
\begin{align}
\eta_\saa 
=\nabla^{(e)}[ (F_{\mathcal{E}}^* \circ R[V_1 ])^*]
(\theta_\saa ).\Label{NM7Y}
\end{align}
\item[(v)]
 The following relation holds.
\begin{align}
\eta_\saa 
&=\argmin_{\eta_\saa' \in \mathbb{R}^{k} }
F_{\mathcal{E}}^* ( \eta_\saa' V_1 )- 
\langle \eta_1' , \theta_\saa \rangle \Label{NADK}.
\end{align}
\end{description}
\end{lemma}
The proofs of the above lemmas are given in Appendix \ref{AL3}.

These lemmas give the following meaning of Condition (B3), which assumes that
the mixture family ${\cal M}$ has the structure of an exponential family.
Due to Lemma \ref{BXPR},
a mixture parameter $\eta_{\saa}$ in ${\cal E}$ is mapped to 
the mixture parameter $\eta_{\saa}V_1$
in ${\cal M}$ by multiplying the matrix $V_1$, which also characterizes
the $m$-projection.
Due to (iii) of Lemma \ref{BXP},
a natural parameter $\theta_{\scc} $ 
in ${\cal E}$ is mapped to 
the natural parameter $V_1\theta_{\scc} $in ${\cal M}$ by multiplying the matrix $V_1$.
This map also characterizes
the $m$-projection when $\theta_{\scc} $ is 
$\nabla^{(m)}[F_{\mathcal{E}}^*]( \eta_\saa V_1)$.
These mappings take a central role in the latter discussion.

In addition, the equivalence between (i) and (iii) in Lemma \ref{BXP} implies
\begin{align}
\Pro^{(e),F}_{\Xi_{\cal E}\to\Theta_{\cal M}}  (\eta_\saa V_1)
= 
V_1 \nabla^{(m)}[F_{\mathcal{E}}^*]( \eta_\saa V_1) .
\end{align}
Combining \eqref{MAF} of Lemma \ref{BXPR}, we have 
\begin{align}
\Pro^{(e),F}_{{\cal E}\to \Theta_{\cal M}}  \circ
\Pro^{(m),F}_{\Xi_{\cal M}\to{\cal E}}   (\eta_\saa)
= 
\Pro^{(e),F}_{\Xi_{\cal E}\to \Theta_{\cal M}}  \circ
\Pro^{(m),F}_{\Xi_{\cal M}\to\Xi_{\cal E}}   (\eta_\saa)
= 
V_1 \nabla^{(m)}[F_{\mathcal{E}}^*]( \eta_\saa V_1) .
\end{align}
The following theorem characterizes the inverse map of
$\Pro^{(e),F}_{{\cal M}}  \circ
\Pro^{(m),F}_{{\cal M}\to{\cal E}}  $.

\begin{theorem}\Label{theo:conv:BBem1}
When the condition (B3) holds, we have the following two statements;
(i) The map $\Pro^{(e),F}_{{\cal M}}\circ \Pro^{(m),F}_{{\cal M}\to{\cal E}}$ has a unique
inverse map $\nabla^{(m)} [F_{\cal M}^*]\circ \nabla^{(e)} [(F_{\cal E}^*\circ R[V_1])^*]$
under the natural parameter of ${\cal M}$.
(ii) In addition, 
for $\hat{\theta}_{\saa} \in \Theta_{{\cal M}}\subset \mathbb{R}^k$,
we have
\begin{align}
\nabla^{(e)} [(F_{\cal E}^*\circ R[V_1])^*](\hat{\theta}_{\saa})  
=&\argmin_{\hat\eta_{\saa} \in \mathbb{R}^{k}} 
F_{{\cal E}}^*( \hat{\eta}_{\saa}V_1 ) 
-\langle\hat\eta_{\saa} , \hat\theta \rangle. \Label{VBJ}
 \end{align}
\end{theorem}
Theorem \ref{theo:conv:BBem1} is proven in Appendix \ref{A3}. 

Therefore, when the conditions (B1) and (B3) hold, 
our algorithm is given as Algorithm \ref{protocol2M}, which is based on the minimization under
 the mixture parameter $\hat{\eta}$.
Further, when the condition (B2) holds additionally, and
we set $\hat{\theta}^{(0)}$ as an arbitrary element in $\Theta_{\cal M}$ and
update it as the rule 
$\hat{\theta}^{(t+1)}:= \nabla^{(m)} [F_{\cal M}^*]\circ \nabla^{(e)} [(F_{\cal E}^*\circ R[V_1])^*](\hat{\theta}^{(t)})$,
then we obtain the maximum value  
$C_{\sup}(\mathcal{M},\mathcal{E})$ as the limit of 
$D^{F} \Big( \psi_{\cal M}^{(e)} (\hat\theta^{(t)}) \Big\| 
\Pro^{(m),F}_{{\cal E}} \big(
 \psi_{\cal M}^{(e)} (\hat\theta^{(t)})\big)  \Big)$.

\begin{algorithm}
\caption{Reverse em-algorithm with mixture parameter under conditions (B1) and (B3)}
\Label{protocol2M}
\begin{algorithmic}
\STATE {Choose the initial value $\hat{\theta}^{(1)}_{\saa} 
\in \Theta_{\mathcal{M}} \subset \mathbb{R}^{k}$;} 
\REPEAT 
\STATE Calculate $\hat\eta^{(t+1)}_{\saa}:=
\argmin_{\hat\eta \in \mathbb{R}^{k}} 
F_{{\cal E}}^*( \hat{\eta}_{\saa}V_1) 
-\langle\hat\eta_{\saa} , \hat{\theta}_{\saa}^{(t)} \rangle$;
\STATE 
Calculate $\hat{\theta}_{\saa}^{(t+1)}
:=\nabla^{(m)} [F_{\cal M}^*](\hat\eta_{\saa}^{(t+1)})
\in \Theta_{\mathcal{M}}\subset \mathbb{R}^k$;
\UNTIL{convergence.} 
\end{algorithmic}
\end{algorithm}

Further, we have the following corollary
of Theorems \ref{theo:conv:BBem} and \ref{theo:conv:BBem1}.

\begin{corollary}\Label{CorT4}
Assume that the conditions (B1), (B2), and (B3) hold.
An invariant point of the map 
$\Pro^{(e),F}_{{\cal M}}\circ \Pro^{(m),F}_{{\cal M}\to {\cal E}}$,
i.e., 
an invariant point of the inverse map 
$(\Pro^{(e),F}_{{\cal M}}\circ \Pro^{(m),F}_{{\cal M}\to{\cal E}})^{-1}$,
is a maximizer in \eqref{eq:em}.
Hence, no minimizer exists in \eqref{eq:em}.
\end{corollary}

\begin{proof}
Theorem \ref{theo:conv:BBem1} guarantees the existence of 
the inverse map $(\Pro^{(e),F}_{{\cal M}}\circ 
\Pro^{(m),F}_{{\cal M}\to{\cal E}})^{-1}$.
Applying Algorithm \ref{protocol2M} by setting 
an invariant point is the initial point,
we find that it is the global maximizer
because Theorem \ref{theo:conv:BBem} guarantees that 
the algorithm asymptotically achieves the global maximum.

Since a minimizer in \eqref{eq:em} is also
an invariant point,
no minimizer exists in \eqref{eq:em}.
\end{proof}

\begin{remark}
The proof technique of Theorem \ref{theo:conv:BBem} is inspired by the proof of 
\cite[Theorem 11]{Shoji}.
In contrast, 
Theorem \ref{theo:conv:BBem+} employs a different technique, which is close to 
\cite[Eq. (25)]{RISB}.
\end{remark}

\subsection{Algorithm with approximate minimization}
However, the minimization $\min_{\hat\eta_{\saa} \in \mathbb{R}^{k}} 
F_{{\cal E}}^*( \hat{\eta}_{\saa}V_1) 
-\langle\hat\eta_{\saa} , \hat{\theta}_{\saa}^{(t)}\rangle$
cannot be solved perfectly in general.
That is, it can be solved only approximately.
Hence, we propose an alternative algorithm, Algorithm \ref{protocol2F},  
by replacing the minimization by $\epsilon$-approximation.
To evaluate the error of Algorithm \ref{protocol2F}, we have Theorem \ref{conv:BBem}.

\begin{algorithm}
\caption{Reverse em-algorithm with $\epsilon$ approximation under mixture parameter under conditions (B1) and (B3)}
\Label{protocol2F}
\begin{algorithmic}
\STATE {Choose the initial value 
$\hat{\theta}_{\saa}^{(1)} \in \Theta_{\mathcal{M}} 
\subset \mathbb{R}^{k}$;} 
\REPEAT 
\STATE Choose an element 
$\hat\eta_{\saa}^{(t+1)} \in \mathbb{R}^{k}$ such that
\begin{align}
&F_{{\cal E}}^*( \hat{\eta}_{\saa}^{(t+1)}V_1)
-\langle\hat\eta_{\saa}^{(t+1)} , \hat\theta_{\saa}^{(t)} \rangle
\le
 \min_{\hat\eta_{\saa} \in \mathbb{R}^{k}} 
F_{{\cal E}}^*( \hat{\eta}_{\saa}V_1) 
-\langle\hat\eta_{\saa} , \hat\theta_{\saa}^{(t)} \rangle
+\epsilon;
\Label{AM8}
\end{align}
\STATE 
Calculate $\hat{\theta}_{\saa}^{(t+1)}:=
\nabla^{(m)} [F_{\cal M}^*](\hat\eta_{\saa}^{(t+1)})
\in \Theta_{\mathcal{M}}\subset \mathbb{R}^k$;
\UNTIL{$t=t_1-1$.}
\STATE {\bf final step:}\quad 
We output the final estimate 
$\hat\theta_{{\saa},f}^{(t_1)} :=\hat\theta_{\saa}^{(t_2)} \in \Theta_{\mathcal{M}}$
by using  $t_2:= \argmin_{t=2, \ldots, t_1} 
D^F(\theta^{(t)} \| \theta_{(t)}) $, where
$\theta_{(t)}:=\psi_{\mathcal{M}}^{(e)}(\theta_{\saa}^{(t)})$
and $\theta^{(t)}:=\Pro^{(m),F}_{{\cal E}} (\theta_{(t)})$.
\end{algorithmic}
\end{algorithm}

\begin{theorem}\Label{conv:BBem}
Assume that the conditions (B1), (B2) and (B3) hold
for a pair of a $k$-dimensional mixture subfamily $\mathcal{M}$ and an $l$-dimensional  exponential subfamily $\mathcal{E}$
and
the maximizer $\theta^*:=\theta^*(\mathcal{M},\mathcal{E})$ in \eqref{eq:Inf-2} exists.
We define the set $\mathcal{M}_0:=\{\theta \in \mathcal{M} |
D^F(\theta^* \| \theta ) \le D^F(\theta^*\| \theta^{(1)} ) \} \subset \mathcal{M}$.
Then, in Algorithm \ref{protocol2F},
the quantity $ D^{F}(\theta^{(t)}\| \Pro^{(m),F}_{{\cal E}}  (\theta^{(t)}) ) 
$ converges to the minimum $C_{\sup}(\mathcal{M},\mathcal{E})$ with the speed 
\begin{align}
&C_{\sup}(\mathcal{M},\mathcal{E})
- D^{F}(\theta_f^{(t_1)}\| \Pro^{(m),F}_{{\cal E}}  (\theta_f^{(t_1)}) ) \nonumber \\
\le &
\max \Big(
 \frac{D^F(\theta^{*}\| \theta^{(1)}  )}{t_1-1}
 + 2\gamma \sqrt{D^F(\theta^{*}\|\theta^{(1)})
\epsilon}+
(\gamma+1) \epsilon ,
2\gamma \sqrt{D^F(\theta^{*}\|\theta^{(1)})
\epsilon}+
(\gamma+1) \epsilon
\Big),
\Label{ANC}
\end{align}
where $\gamma:=\gamma(\mathcal{M}_0|\mathcal{M})$.
Further, when 
$t_1-1 \ge \frac{2 D^F(\theta^{*} \| \theta^{(1)})}{\epsilon'}$
and
$\epsilon \le 
\frac{{\epsilon'}^2}{4
(3\gamma+1)^2 {D^F(\theta^{*}\|\theta^{(1)})}}$,
the parameter $ \theta_f^{(t_1)}$ satisfies
\begin{align}
C_{\sup}(\mathcal{M},\mathcal{E})
- D^{F}(\theta_f^{(t_1)}\| \Pro^{(m),F}_{{\cal E}}  (\theta_f^{(t_1)}) ) 
 \le \epsilon'  .
\Label{NAC}
\end{align}
\end{theorem}
Theorem \ref{conv:BBem} is proven in Appendix \ref{A4}. 

\subsection{Algorithm based on quadratic approximation}\Label{S45}
Now, we apply the quadratic approximation in the minimization in 
Algorithm \ref{protocol2M}.
We define $\bar{\eta}_{\saa} \in \mathbb{R}^k$ 
and $\bar{\theta}_{\scc} \in \mathbb{R}^l $ as 
\begin{align}
\bar{\eta}_{\saa}= \nabla^{(e)}[F_{{\cal M}}] (\hat{\theta}_{\saa}),\quad
\bar{\eta}_{\saa}V_1= \nabla^{(e)}[F_{{\cal E}}](\bar{\theta}_{\scc}) ,
~(\hbox{i.e.,}~
\bar{\theta}_{\scc}=\nabla^{(m)}[F_{{\cal E}}^*]( \bar{\eta}_{\saa}V_1)).
\end{align}
Then, we have
\begin{align}
&F_{{\cal E}}^*( \hat{\eta}_{\saa}V_1) \nonumber \\
\cong &
F_{{\cal E}}^*( \bar{\eta}_{\saa}V_1)
+(\hat{\eta}_{\saa}-\bar{\eta}_{\saa}) 
V_1 \nabla^{(m)}[F_{{\cal E}}^*]( \bar{\eta}_{\saa}V_1)
 \nonumber \\
&+\frac{1}{2}(\hat{\eta}_{\saa}-\bar{\eta}_{\saa}) V_1 
(J[F_{{\cal E}}^*](\bar{\eta}_{\saa}V_1))
V_1^T (\hat{\eta}_{\saa}-\bar{\eta}_{\saa})^T\nonumber  \\
= &
F_{{\cal E}}^*( \bar{\eta}_{\saa}V_1)
+(\hat{\eta}_{\saa}-\bar{\eta}_{\saa}) V_1 \bar{\theta}_{\scc}
+\frac{1}{2}(\hat{\eta}_{\saa}-\bar{\eta}_{\saa}) 
D^F(\bar\theta_{\scc}) (\hat{\eta}_{\saa}-\bar{\eta}_{\saa})^T ,
\end{align}
where
\begin{align}
D(\bar\theta_{\scc}):= 
V_1 (J[F_{{\cal E}}^*](\bar{\eta}_{\saa}V_1)) V_1^T 
=V_1 (J[F_{{\cal E}}](\bar{\theta}_{\scc}))^{-1} V_1^T.
\end{align}
Hence, we have
\begin{align}
&F_{{\cal E}}^*( \hat{\eta}_{\saa}V_1) 
-\langle\hat\eta_{\saa} , \hat{\theta}_{\scc}^{(t)}\rangle
\nonumber \\
\cong &
F_{{\cal E}}^*( \bar{\eta}_{\saa}V_1)
-\langle\bar\eta_{\saa} , \hat{\theta}_{\saa}^{(t)}\rangle
+(\hat{\eta}_{\saa}-\bar{\eta}_{\saa}) 
(V_1 \bar{\theta}_{\scc}-\hat{\theta}_{\saa})
+\frac{1}{2}(\hat{\eta}_{\saa}-\bar{\eta}_{\saa}) 
D(\bar\theta_{\scc}) (\hat{\eta}_{\saa}-\bar{\eta}_{\saa})^T \nonumber \\
=&
F_{{\cal E}}^*( \bar{\eta}_{\saa}V_1)
-\langle\bar\eta_{\saa} , \hat{\theta}_{\saa}^{(t)} \rangle\nonumber \\
&+\frac{1}{2}
\big(\hat{\eta}_{\saa}-\bar{\eta}_{\saa}
- D(\bar\theta_{\scc})^{-1}(V_1 \bar{\theta}_{\scc}
-\hat{\theta}_{\saa})\big)
 D(\bar\theta_{\scc})
\big(\hat{\eta}_{\saa}-\bar{\eta}_{\saa}
- D(\bar\theta_{\scc})^{-1}(V_1 \bar{\theta}_{\scc}
-\hat{\theta}_{\saa})
\big)^T.
\end{align}
The minimum in Algorithm \ref{protocol2M} is approximately achieved when 
\begin{align}
\hat{\eta}_{\saa}=\bar{\eta}_{\saa}
+ D(\bar\theta_{\scc})^{-1}(V_1 \bar{\theta}_{\scc}
-\hat{\theta}_{\saa}).
\Label{ACP}
\end{align}
This approximation is effective when 
$\hat{\theta}_{\saa}$ is close to the minimizer $\theta^*$.

\subsection{Algorithm based on minimization under natural parameter}\Label{S46}
The above algorithms are based on the mixture parameter of 
$\mathcal{E}$ for the calculation of $ \nabla^{(e)} [(F_{\cal E}^*\circ R[V_1])^*]$.
To make an alternative algorithm based on the natural parameter of 
$\mathcal{E}$,
we introduce additional conditions.
\begin{description}
\item[(B4)]
The $k \times l$ matrix $V_1$ has the following form;
$V_1=(I ,V_3)$ with a $k \times (l-k) $ matrix $V_3$.
\item[(B5)]
The relation $\Theta_{\mathcal{E}}=
\Theta_{\mathcal{E},\saa}\times \Theta_{\mathcal{E},\sbb}$ holds with 
$\Theta_{\mathcal{E},\saa}= \mathbb{R}^k$
and $\Theta_{\mathcal{E},\sbb}= \mathbb{R}^{l-k}$.
$F_{\cal E}$ has the following form;
$F_{\cal E}(\theta_\saa,\theta_\sbb)=F_{{\cal E},\saa}(\theta_\saa)+F_{{\cal E},\sbb}(\theta_\sbb)$
with $(\theta_\saa ,\theta_\sbb) \in \Theta_{\mathcal{E}}$.
\end{description}

The meaning of Conditions (B4) and (B5) are the following.
Condition (B5) means that the exponential family is given as the product of 
two exponential families ${\cal E}_\saa$ and ${\cal E}_\sbb$.
Condition (B3) gives 
a linear map from a natural parameter in ${\cal M}$
to a natural parameter in ${\cal E}$
and 
a linear map from a mixture parameter in ${\cal E}$
to a mixture parameter in ${\cal M}$ via $V_1$.
When Condition (B4) holds, the above linear maps can be simplified.

Then, we have the following theorem.
\begin{theorem}\Label{VCF}
When the conditions (B3) and (B4) hold, for $\hat{\theta} \in \Theta_{{\cal M}}\subset \mathbb{R}^k$,
we have
\begin{align}
\nabla^{(e)} [(F_{\cal E}^*\circ R[V_1])^*](\hat{\theta})  =&
  \nabla^{(e)}[F_{\cal E}](\bar{\theta})
 \left( 
\begin{array}{c}
 I_k \\
 0
\end{array}
 \right)\Label{VBI},
 \end{align}
where $\bar{\theta}_{\scc}:=
(\hat{\theta}_{\saa}-V_3 \theta_\sbb^*(\hat{\theta}_{\saa}),\theta_\sbb^*(\hat{\theta}_{\saa}))^T$
and $\theta_\sbb^*(\hat{\theta}_{\saa}):=
\argmin_{\theta_{\sbb}\in \mathbb{R}^{l-k}} 
F_{{\cal E}}( \hat{\theta}_{\saa}-V_3 \theta_\sbb,\theta_\sbb) $.
\end{theorem}
Theorem \ref{VCF} is proven in Appendix \ref{A5}. 

Additionally, when the condition (B5) holds, we can use the following corollary.

\begin{corollary}\Label{VCF-C}
When conditions (B3), (B4), and (B5) hold, for 
$\hat{\theta}_{\saa} \in \Theta_{{\cal M}}\subset \mathbb{R}^k$,
we have
\begin{align}
\nabla^{(e)} [(F_{\cal E}^*\circ R[V_1])^*](\hat{\theta}_{\saa})  
=&
  \nabla^{(e)}[F_{{\cal E},1}](
  \hat{\theta}_{\saa}-V_3 \theta_\sbb^*(\hat{\theta}_{\saa})  )
\Label{VBI-C},
 \end{align}
where $\theta_\sbb^*(\hat{\theta}_{\saa}):=
\argmin_{\theta_\sbb\in \mathbb{R}^{l-k}} 
F_{{\cal E},\saa} ( \hat{\theta}_{\saa}-V_3 \theta_\sbb)
+F_{{\cal E},\sbb}(\theta_\sbb) $.
\end{corollary}

Therefore, 
thanks to Theorem \ref{VCF}, Corollary \ref{VCF-C},
and (i) of Theorem \ref{theo:conv:BBem1},
we can use Algorithm \ref{protocol2} to calculate 
$\hat{\theta}_{\saa}^{(t+1)}$ from $\hat{\theta}_{\saa}^{(t)}$
instead of Algorithm \ref{protocol2M}.
To implement Algorithm \ref{protocol2},
we need to calculate 
the minimization 
\begin{align}
\min_{\theta_\sbb\in \mathbb{R}^{l-k}} 
F_{{\cal E}}( \hat{\theta}_{\saa}-V_3 \theta_\sbb,\theta_\sbb) .\Label{XPZ}
\end{align}
The merit of our method is determined by 
whether the minimization \eqref{XPZ} is easier than the original maximization \eqref{eq:Inf.st}.
Since $F_{{\cal E}}$ is a convex function, 
the minimization \eqref{XPZ} can be solved by the convex optimization.
However, 
there is a case that the maximization \eqref{eq:Inf.st}
is also given as the minimization of a concave function.
Hence, this type of comparison depends on the target problem.

\begin{algorithm}
\caption{Reverse em-algorithm with natural parameter under conditions (B1), (B3), and (B4)}
\Label{protocol2}
\begin{algorithmic}
\STATE {Choose the initial value $\hat{\theta}_{\saa}^{(1)} 
\in \Theta_{\mathcal{M}} \subset \mathbb{R}^{k}$;} 
\REPEAT 
\STATE {Set $\hat{\theta}_{\saa}=\hat{\theta}_{\saa}^{(t)} 
\in \Theta_{\mathcal{M}}$;} 
\STATE Calculate $\bar\theta_\sbb^*(\hat{\theta}):=
\argmin_{\theta_\sbb\in \mathbb{R}^{l-k}} 
F_{{\cal E}}( \hat{\theta}_{\saa}-V_3 \theta_\sbb,\theta_\sbb) $
and $\bar{\theta}_{\saa}:=(\hat{\theta}_{\saa}
-V_3 \bar\theta_\sbb^*(\hat{\theta}_{\saa}),\bar\theta_\sbb^*(\hat{\theta}_{\saa}))^T$;
\STATE 
Calculate $\hat{\theta}_{\saa}^{(t+1)}
:=\nabla^{(m)} [F_{\cal M}^*]\left(
  \nabla^{(e)}[F_{\cal E}](\bar{\theta}_{\saa})
 \left( 
\begin{array}{c}
 I_k \\
 0
\end{array}
 \right)
\right)\in \mathbb{R}^k$;
\UNTIL{convergence.} 
\end{algorithmic}
When the condition (B5) holds additionally,
the calculation of $\hat{\theta}_{\saa}^{(t+1)}$
can be simplified as 
$\hat{\theta}_{\saa}^{(t+1)}:=\nabla^{(m)} [F_{\cal M}^*]\left(
  \nabla^{(e)}[F_{{\cal E},\saa}](
  \hat{\theta}_{\saa}-V_3 \theta_\sbb^*(\hat{\theta}_{\saa})  ) \right)$
with $\theta_\sbb^*(\hat{\theta}_{\saa})
:=\argmin_{\theta_\sbb \in \mathbb{R}^{l-k}} 
F_{{\cal E},\saa} ( \hat{\theta}_{\saa}-V_3 \theta_\sbb)
+F_{{\cal E},\sbb}(\theta_\sbb) $.
\end{algorithm}

\subsection{Conversion to em-problem}\Label{S47}
Next, we convert the reverse em-problem \eqref{eq:Inf.st} to the em-problem \eqref{eq:em}. 
We focus on the fixed point in Algorithm \ref{protocol2M}.
Theorem \ref{theo:conv:BBem}
guarantees that 
the convergence point is the maximizer of the maximization \eqref{eq:Inf.st}.
Since the fixed point equals the convergence point,
the fixed point is the maximizer of the maximization \eqref{eq:Inf.st}.
Therefore, characterizing the fixed point by Theorem \ref{theo:conv:BBem1}, we have the following theorem. 
\begin{theorem}\Label{CPZ}
Assume Conditions (B1) 
and (B3). 
Then, the following three conditions 
for $\theta_{\saa} \in \Theta_{{\cal M}}$ are equivalent.
\begin{description}
\item[(D1)] 
$\theta_{\saa} \in \Theta_{{\cal M}}$ 
is
an invariant point of the map 
$\Pro^{(e),F}_{{\cal M}}\circ \Pro^{(m),F}_{{\cal M}\to{\cal E}}$,
i.e., 
an invariant point of 
the inverse map $(\Pro^{(e),F}_{{\cal M}}\circ 
\Pro^{(m),F}_{{\cal M}\to{\cal E}})^{-1}$.
\item[(D2)] The relation
$V_1 \nabla^{(m)} [F_{\cal E}^*] (
\nabla^{(e)} [F_{\cal M}] (\theta_{\saa}) V_1) =\theta$ holds.
\item[(D3)] 
The mixture parameter 
$\eta_{\saa}=\nabla^{(e)} [F_{\cal M}](\theta_{\saa})$ satisfies
\begin{align}
V_1 \nabla^{(m)} [F_{\cal E}^*] (\eta_{\saa} 
 V_1) =\nabla^{(m)} [F_{\cal M}^*] (\eta_{\saa}).\Label{CCA}
\end{align}
\end{description}

When Condition (B2) holds in addition to (B1) and (B3),
the following two conditions 
for the pair of ${\cal E}$ and ${\cal M}$ 
are equivalent.
\begin{description}
\item[(D4)] There exists an element 
$\theta_{\saa} \in \Theta_{{\cal M}}$ to 
satisfy the condition (D1), (D2), or (D3).
\item[(D5)] The set $\Theta^*(\mathcal{M},\mathcal{E})$ is not empty.
\end{description}
\end{theorem}
\begin{proof}
First, we show the equivalence among (D1), (D2), and (D3).
When $\theta_{\saa}$ satisfies Condition (D1), $\theta_{\saa}$ 
is a fixed point for the iteration given in Theorem \ref{theo:conv:BBem1}, which is equivalent to the condition:
\begin{align}
\nabla^{(m)} [F_{\cal M}^*]\circ \nabla^{(e)} 
[(F_{\cal E}^*\circ R[V_1])^*](\theta_{\saa})=\theta_{\saa}.
\Label{XHYP}
\end{align}
We choose the mixture parameter $\eta_{\saa}=\nabla^{(e)} [F_{\cal M}](\theta_{\saa})$, which implies 
\begin{align}
\nabla^{(m)} [F_{\cal M}^*](\eta_{\saa})=\theta_{\saa}.
\Label{XHYPB}
\end{align}
Hence, the condition \eqref{XHYP} is equivalent to 
\begin{align}
\nabla^{(e)} [(F_{\cal E}^*\circ R[V_1])^*](\theta_{\saa})=\eta_{\saa}.
\Label{XHYP2}
\end{align}
Due to \eqref{VBJ}, the condition \eqref{XHYP2} is equivalent to
\begin{align}
V_1\nabla^{(m)} [F_{\cal E}^*](\eta_{\saa} V_1)=\theta_{\saa}.
\Label{XHYP3}
\end{align}
The combination of \eqref{XHYPB} and \eqref{XHYP3} implies the equivalence between 
(D1) and (D3).
Also, substituting $\nabla^{(e)} [F_{\cal M}](\theta_{\saa})$ into 
$\eta_{\saa}$ at \eqref{XHYP3},
we obtain the equivalence between (D1) and (D2).

Under Conditions (B1), (B2), and (B3),
Corollary \ref{CorT4} guarantees that
an invariant point is limited to an element of 
the set $\Theta^*(\mathcal{M},\mathcal{E})$.
Hence, we obtain
the equivalence between (D4) and (D5) for 
$\theta_{\saa} \in \Theta_{{\cal M}}$.
\end{proof}

We define the exponential family
$\Theta_{{\cal M},{\cal E}}:= 
\Theta_{{\cal M}}\times \Theta_{{\cal E}}$
with the potential function
$F_{{\cal M},{\cal E}}((\theta_{\saa},\theta_{\scc})^T):= 
F_{{\cal M}}(\theta_{\saa})+F_{{\cal E}}(\theta_{\scc})$.
This exponential family has the mixture parameter
$(\eta_{\saa},\eta_{\scc})$ as
\begin{align}
\nabla^{(e)}[F_{{\cal M},{\cal E}}]
= 
(\nabla^{(e)}[F_{{\cal M}}],\nabla^{(e)}[F_{{\cal E}}]).\Label{XOA}
\end{align}

We define 
the mixture subfamily $\hat{\cal M}$ and 
the exponential subfamily $\hat{\cal E}$ as
\begin{align}
\hat{\cal M}:= &\left\{
\left(\begin{array}{c}
\theta_{\saa}\\ \theta_{\scc}
\end{array}\right)
 \left| (\eta_{\saa},\eta_{\scc}) =\nabla^{(e)}[F_{{\cal M},{\cal E}}]
 \left(\begin{array}{c}
\theta_{\saa}\\ \theta_{\scc}
\end{array}\right)
, ~(\eta_{\saa}, \eta_{\scc} ) 
\left(\begin{array}{c}
V_1 \\
-I
\end{array}\right)
=0 \right.\right\} \\
\hat{\cal E}:=& \left\{
\left.\left(\begin{array}{c}
\theta_{\saa}\\ \theta_{\scc}
\end{array}\right)
\right| V_1 \theta_{\scc}=\theta_{\saa}  \right\} 
=  \left\{
\left(\begin{array}{c}
V_1 \theta_{\scc}\\ \theta_{\scc}
\end{array}\right)
\right\} 
=  \left\{
\left(\begin{array}{c}
V_1 \\
I
\end{array}\right)
\theta_{\scc} \right\} .
\end{align}
By using $\theta_{\scc}=\nabla^{(m)} [F_{\cal E}^*] (\eta_{\saa} V_1) \in \Theta_{{\cal E}}$,
(D3) of Theorem \ref{CPZ} is rewritten as
\begin{align}
\eta_\saa V_1= \nabla^{(e)} [F_{\cal E}] (\theta_{\scc}),\quad
\eta_\saa= \nabla^{(e)} [F_{\cal M}](V_1\theta_{\scc}),
\end{align}
which implies that
\begin{align}
\nabla^{(e)} [F_{\cal M}](V_1\theta_{\scc}) V_1= \nabla^{(e)} [F_{\cal E}] (\theta_{\scc}).\Label{XZA}
\end{align}
Since the condition \eqref{XZA} for $\theta_{\scc}$
is equivalent to the condition that
the element $ \left(\begin{array}{c}
V_1 \theta_{\scc}\\ \theta_{\scc}
\end{array}\right) \in \hat{\cal E}$ belongs to $\hat{\cal M}$,
we have the following corollary of Theorem \ref{CPZ}.
\begin{corollary}\Label{Cor3}
Assume that Conditions (B1) and (B3) hold and
the intersection $\hat{\cal M}\cap \hat{\cal E}$ is not empty.
For an element $(\theta_{\saa,*},\theta_{\scc,*}) \in \hat{\cal M}\cap \hat{\cal E}$,
$\theta_{\saa,*} \in \Theta_{{\cal M}}$ 
is
an invariant point of the map 
$\Pro^{(e),F}_{{\cal M}}\circ \Pro^{(m),F}_{{\cal M}\to{\cal E}}$,
i.e., 
an invariant point of 
the inverse map $(\Pro^{(e),F}_{{\cal M}}\circ \Pro^{(m),F}_{{\cal M}\to{\cal E}})^{-1}$.

When Condition (B2) holds additionally,
the maximization \eqref{eq:Inf.st} is written 
as follows.
\begin{align}
C_{\sup}({\cal M},{\cal E})=D^F(\phi^{(e)}_{{\cal M}}(\theta_{\saa,*}) \|
\phi^{(e)}_{{\cal E}}(\theta_{\scc,*})).
\end{align}
\end{corollary}

Therefore, the reverse em-problem \eqref{eq:Inf.st} is reduced to 
finding the element $(\theta_{\saa,*},\theta_{\scc,*}) \in \hat{\cal M}\cap \hat{\cal E}$.
This element can be found by solving the following minimization problem;
\begin{align}
\argmin_{(\theta_{\saa},\theta_{\scc}) \in \hat{\cal M}} 
\min_{(\theta_{\saa}',\theta_{\scc}') \in \hat{\cal E}} 
D^{F_{{\cal M}}}(\theta_{\saa} \| \theta_{\saa}')
+D^{F_{{\cal E}}}(\theta_{\scc} \| \theta_{\scc}').\Label{AML}
\end{align}
Since 
$\hat{\cal E}$ is an exponential family and $\hat{\cal M}$ is a mixture family,
the above minimization problem \eqref{AML} is a special case of 
the em-problem \eqref{eq:em}.
Therefore, to solve \eqref{AML}, we can employ the em-algorithm, Algorithm \ref{protocol1-0}. 

As another method to characterize the intersection $\hat{\cal M}\cap\hat{\cal E}$,
we assume Condition (B4), and introduce the parameterizations
$\theta_{\scc}=(\theta_{\saa},\theta_{\sbb})$ and 
$\eta_{\scc}=(\eta_{\saa},\eta_{\sbb})$.
Additionally, we introduce the following new condition.

\begin{description}
\item[(B6)]
Condition (B4) and 
the relation $k \ge 2l$ hold.
The rank of  $V_3$ is $l$.
The vector 
\begin{align}
\theta_{\saa,*}:=\nabla^{(m)}[F_{{\cal M}}^*](\eta_{\saa})-
\nabla_\saa^{(m)}[F_{{\cal E}}^*](\eta_{\saa}(I, V_3))
\Label{MMZ}
\end{align}
does not depend on  $\eta_{\saa} \in \Xi_{{\cal M}} \subset \mathbb{R}^l$, where
\begin{align}
\nabla^{(m)}[F_{{\cal E}}^*](\eta_{\saa}(I, V_3))
=
\left(\begin{array}{c}
\nabla^{(m)}_{\saa}[F_{{\cal E}}^*](\eta_{\saa}(I, V_3)) \\
\nabla^{(m)}_{\sbb}[F_{{\cal E}}^*](\eta_{\saa}(I, V_3))
\end{array}
\right).
\end{align}
\end{description}

When Condition (B6) holds,
the first $l$ natural parameters of $\eta_{\saa}(I,V_3)$
in ${\cal E}$
can be calculated from the natural parameters of $\eta_{\saa}$
in ${\cal M}$.

We choose an element $\theta_{\sbb,*} \in \mathbb{R}^{k-l}$
and a $(k-l)\times (k-2l)$ matrix $V_4 $ such that 
\begin{align}
\theta_{\saa,*}=& V_3 \theta_{\sbb,*} \Label{ZMP}\\
\Ker V_3 =& \im V_4.\Label{ZMP2}
\end{align}
The existence of $\theta_{\sbb,*}$ is guaranteed by Condition (B6) (the rank condition for $V_3$).
Then,
we define the following exponential and mixture subfamilies of 
${\cal E}$ as
\begin{align}
\overline{\cal E}:=& \{ (\theta_{\saa},\theta_{\sbb,*}
+ V_4 \theta_\see)^T |
\theta_\saa \in \mathbb{R}^l, \theta_\see \in \mathbb{R}^{k-2l} \} \Label{XC3}\\
\overline{\cal M}:=& \{ (\theta_\saa,\theta_\sbb)^T |
(\eta_\saa,\eta_\sbb)= \nabla^{(e)}[F_{{\cal E}}]((\theta_\saa,\theta_\sbb)^T),~
\eta_\saa V_3- \eta_\sbb=0 \} .\Label{XC4}
\end{align}
We have the following corollary of Corollary \ref{Cor3}.

\begin{corollary}\Label{Cor4}
Assume Conditions (B1), (B3), (B4), and (B6).
The following two conditions for an element 
$(\theta_{\saa},\theta_{\sbb})^T \in \Theta_{{\cal E}} $ are equivalent.
\begin{description}
\item[(F1)] 
The point
$\Pro^{(e),F}_{\Theta_{\cal E}\to {\cal M}}
((\theta_{\saa},\theta_{\sbb})^T)$
is invariant for the map
$\Pro^{(e),F}_{{\cal M}}\circ \Pro^{(m),F}_{{\cal M}\to{\cal E}}$.
\item[(F2)] The element $(\theta_{\saa},\theta_{\sbb})^T$ belongs to 
the intersection $\overline{\cal E} \cap \overline{\cal M}$.
\item[(F3)]
There is an element $\eta_{\saa}'\in \Xi_{\cal M}$  such that
$\psi_{\cal M}^{(m)}(\eta_{\saa}')$
is invariant for the map
$\Pro^{(e),F}_{{\cal M}}\circ \Pro^{(m),F}_{{\cal M}\to{\cal E}}$
and 
$(\theta_{\saa},\theta_{\sbb})^T
=\Pro^{(m),F}_{\Xi_{\cal M}\to\Theta_{\cal E}}(\eta_{\saa}')$.
\end{description}

When Condition (B2) holds additionally,
(F1) is equivalent to the following condition.
\begin{description}
\item[(F1')] 
The maximum exists in \eqref{eq:Inf.st}, i.e., 
\begin{align}
C_{\sup}({\cal M},{\cal E})=
D^F(
\Pro^{(e),F}_{\Theta_{\cal E}\to {\cal M}}(\theta_{\saa},
\theta_{\sbb})^T
) \|
\phi^{(e)}_{{\cal E}}((\theta_{\saa},\theta_{\sbb})^T).
\Label{SXO}
\end{align}
\end{description}
\end{corollary}
Therefore,
when the intersection $\overline{\cal E} \cap \overline{\cal M}$ is not empty
and Conditions (B1), (B2), (B3), (B4), and (B6) hold,
the maximization \eqref{eq:Inf.st} is written 
by the element of $\overline{\cal E} \cap \overline{\cal M}$ as \eqref{SXO}.

\begin{proof}
We choose $(\eta_{\saa},\eta_{\sbb})= 
\nabla^{(e)}[F_{{\cal E}}]((\theta_{\saa},\theta_{\sbb})^T)$.
The equivalent between (F1) and (F3) is trivial.

In order that $\psi_{{\cal E}}^{(e)} (\eta_{\saa},\eta_{\sbb})$
satisfies the condition (F1),
$(\eta_{\saa},\eta_{\sbb})$ needs  to be written as
$\Pro^{(m),F}_{\Xi_{\cal M}\to\Xi_{\cal E}}(\eta_{\saa}')$
with $\eta_{\saa}' \in \Xi_{\cal M}$.
In addition, 
Lemma \ref{BXPR} guarantees that
$\Pro^{(m),F}_{\Xi_{\cal M}\to\Xi_{\cal E}}(\eta_{\saa}')
= \eta_{\saa}'(I, V_3)$ and $\eta_{\saa}= \eta_{\saa}'$.
That is, 
the condition (i)
$(\eta_{\saa},\eta_{\sbb})=\eta_{\saa}(I, V_3)$, i.e., 
$\eta_{\saa} V_3- \eta_{\sbb}=0$,
is a necessary condition for (F1).
In the following, we discuss the equivalent condition for (F1) under this necessary condition (i).

Condition (F1) is equivalent to each of  
the following conditions.
\begin{align}
\theta_{\saa,*}
=&V_3 \nabla^{(m)}_\sbb[F_{{\cal E}}^*](\eta_{\saa}(I, V_3))\Label{CAP} \\
\nabla^{(m)}[F_{{\cal M}}^*](\eta_{\saa})
=& V_1\nabla^{(m)}[F_{{\cal E}}^*](\eta_{\saa}(I, V_3))
\Label{CAP2} 
\end{align}
because 
\eqref{CAP2} is the same as \eqref{CCA}, which is
equivalent to Condition (E1),
and we have
\begin{align}
&\theta_{\saa,*}-V_3 \nabla^{(m)}_\sbb[F_{{\cal E}}^*](\eta_{\saa}'(I, V_3))\nonumber \\
\overset{(a)}=& \nabla^{(m)}[F_{{\cal M}}^*](\eta_{\saa}')
-
\nabla^{(m)}_{\saa}[F_{{\cal E}}^*](\eta_{\saa}'(I, V_3))
-V_3 \nabla^{(m)}_\sbb[F_{{\cal E}}^*](\eta_{\saa}'(I, V_3)) \nonumber \\
\overset{(b)}=& \nabla^{(m)}[F_{{\cal M}}^*](\eta_{\saa}')
- V_1\nabla^{(m)}[F_{{\cal E}}^*](\eta_{\saa}'(I, V_3)) ,\Label{CAP3}
\end{align}
where 
$(a)$ and $(b)$ follow from (B6) and the relation $V_1=(I,V_3))$,
respectively.

The condition \eqref{CAP}
is equivalent to the condition
$\theta_{\saa,*}=V_3 \theta_{\sbb}$.
This condition is equivalent to the condition (ii) 
that $\theta_{\sbb}$ is written as 
$\theta_{\sbb,*}+ V_4 \theta_{\see}$.
Since the conditions (i) and (ii) correspond to the sets 
$\overline{\cal M}$ and $\overline{\cal E}$, respectively.
Therefore, (F1) implies (F2).

Conversely, when Condition (F2) holds,
the conditions (i) and (ii) hold.
Due to \eqref{CAP3}, 
under the condition (i), the condition (ii), i.e., \eqref{CAP} implies 
\eqref{CAP2}, which is equivalent to (F1).
Therefore, (F2) implies (F1).
The desired equivalence is obtained.
\end{proof}

Therefore, the reverse em-problem \eqref{eq:Inf.st} is reduced to 
finding the element 
$(\bar\theta_{\saa,*},\bar\theta_{\sbb,*})^T \in
\overline{\cal E} \cap \overline{\cal M}$.
This element can be found by solving the following minimization problem;
\begin{align}
\argmin_{(\theta_{\saa},\theta_{\sbb}) \in \overline{\cal M}} 
\min_{(\theta_{\saa}',\theta_{\sbb}') \in \overline{\cal E}} 
D^{F_{{\cal E}}}((\theta_{\saa},\theta_{\sbb}) \| 
(\theta_{\saa}',\theta_{\sbb}')).
\Label{AML2}
\end{align}
Since 
$\overline{\cal E}$ is an exponential family and $\overline{\cal M}$ is a mixture family,
the above minimization problem \eqref{AML2} is another special case of 
the em-problem \eqref{eq:em}.
Therefore, to solve \eqref{AML2}, we can employ the em-algorithm, Algorithm \ref{protocol1-0}. 
The minimization problem \eqref{AML2} has a smaller number of free parameters
than the minimization problem \eqref{AML}.

The following is an alternative method to find an element of 
$\Pro^{(e),F}_{{\cal M}}\circ \Pro^{(m),F}_{{\cal M}\to{\cal E}}$.
Find an element $\eta_{\saa,*}$ to realize an
extremal value of the following function;
\begin{align}
\kappa(\eta_{\saa}):=
F_{{\cal E}^*}(\eta_{\saa}(I,V_3))
-F_{{\cal M}^*}(\eta_{\saa})
-\langle\eta_{\saa} , \theta_{\saa,*}\rangle
-\langle\eta_{\saa} , V_3 \theta_{\sbb,*}\rangle.
\end{align}

\begin{lemma}
Assume Conditions (B1), (B3), (B4), and (B6).
The condition 
\begin{align}
\nabla [\kappa] (\eta_{\saa,*})=0\Label{MZY}
\end{align}
is equivalent to 
$\nabla^{(m)} [F_{{\cal E}^*}]
(\eta_{\saa,*}(I,V_3))\in 
\overline{\cal M}\cap \overline{\cal E}$.
\end{lemma}
That is, the above extremal value gives 
the solution \eqref{SXO}.

\begin{proof}
We  have
\begin{align*}
&\nabla [\kappa] (\eta_{\saa,*})\\
=&(I,V_3) 
(\nabla^{(m)} [F_{{\cal E}^*}](\eta_{\saa,*}(I,V_3))
-
\nabla^{(m)} [F_{{\cal M}^*}](\eta_{\saa,*})
-\theta_{\saa,*}
-V_3 \theta_{\sbb,*} \\
\overset{(a)}=&V_3 
(\nabla^{(m)}_{\sbb} [F_{{\cal E}^*}](\eta_{\saa,*}(I,V_3))
-\theta_{\sbb,*}),
\end{align*}
where $(a)$ follows from \eqref{MMZ} in (B6).
Hence, \eqref{MZY}
is equivalent to $\nabla^{(m)} [F_{{\cal E}^*}]
(\eta_{\saa,*}(I,V_3))\in 
\overline{\cal E}$.
\end{proof}

\subsection{Non-iterative method}\Label{S47B}
We directly characterize the maximizer of the maximization \eqref{eq:Inf.st}
without iterations. 
For this aim, we assume Condition (B5) in addition to (B1), (B3), and (B4).
When Condition (B5) holds, Condition (B6) is rewritten as follows.
\begin{description}
\item[(B6+)]
Condition (B4) and 
the relation $k \ge 2l$ hold.
The rank of  $V_3$ is $l$.
The vector $
\theta_{\saa,*}:=\nabla^{(m)}[F_{{\cal M}}^*](\eta_\saa)-
\nabla^{(m)}[F_{{\cal E},\saa}^*](\eta_{\saa})$
does not depend on $\eta_{\saa} \in \Xi_{{\cal M}} \subset \mathbb{R}^l$.
\end{description}

Using the solution $\theta_{\sbb,*}$ of \eqref{ZMP}, 
we define the following exponential and mixture subfamilies of 
${\cal E}_{\sbb}$ as
\begin{align}
\overline{\cal E}_{\sbb}:=& \{ \theta_{\sbb,*}+ 
V_4 \theta_{\see} \in \Theta_{{\cal E},\sbb}|
 \theta_{\see} \in \mathbb{R}^{k-2l} \} \Label{XC1}\\
\overline{\cal M}_{\sbb}:=& 
\{ \theta_{\sbb} \in \Theta_{{\cal E},\sbb}|
\eta_{\sbb}=\nabla^{(e)}[F_{{\cal E},\sbb}](\theta_{\sbb}),~
\eta_{\sbb} V_4=0 \} .\Label{XC2}
\end{align}
We have the following corollary of Corollary \ref{Cor4}.
\begin{corollary}\Label{Cor5}
Assume Conditions (B1), (B3), (B4), (B5), (B6+), and $\Theta=\mathbb{R}^d$.
The following condition (E4) for an element $(\bar\theta_{\saa,*},\bar\theta_{\sbb,*})^T \in \Theta_{{\cal E}} $ is equivalent to (F1), (F2), and (F3) in Corollary \ref{Cor4}.
\begin{description}
\item[(F4)]
The following relations hold.
\begin{align}
\bar\theta_{\sbb,*} &\in  \overline{\cal E}_{\sbb} \cap \overline{\cal M}_{\sbb}\Label{XOP}\\
\nabla^{(e)}[F_{{\cal E},\sbb}](\bar\theta_{\sbb,*})
&=  \nabla^{(e)}[F_{{\cal E},\saa}](\bar\theta_{\saa,*})V_3.\Label{XOP2}
\end{align}
 \end{description}
\end{corollary}
\begin{proof}
Condition (F2) 
element $(\bar\theta_{\saa,*},\bar\theta_{\sbb,*})^T \in \Theta_{{\cal E}} $
is equivalent to the pair of the following conditions. 
(i) $\bar\theta_{\sbb,*}$ has the form $\theta_{\sbb,*}+ V_4 \theta_{\see}$, which corresponds to the condition 
$(\bar\theta_{\saa,*},\bar\theta_{\sbb,*})^T\in 
\overline{\cal E}$.
(ii) The pair $(\bar\theta_{\saa,*},\bar\theta_{\sbb,*})$ 
satisfies the condition 
\eqref{XOP2}, 
which corresponds to the condition 
$(\bar\theta_{\saa,*},\bar\theta_{\sbb,*})^T\in 
\overline{\cal M}$.
To satisfy \eqref{XOP2}, $\nabla^{(e)}[F_{{\cal E}_\sbb}]
(\bar\theta_{\sbb,*})
$ needs to have the form $\bar\eta_{\saa,*} V_3$ with 
$\bar\eta_{\saa,*} \in \mathbb{R}^l$,
which is equivalent to the condition (iii); 
$\nabla^{(e)}[F_{{\cal E}_\sbb}](\bar\theta_{\sbb,*}) V_4=0$, 
i.e., 
$(\bar\theta_{\saa,*},\bar\theta_{\sbb,*})^T\in 
\overline{\cal M}_\sbb$.
Since the conditions (i), (ii), and (iii) are equivalent to Condition (F4),
we obtain the desired statement.
\end{proof}

Although $\overline{\cal E}_\sbb$ is an exponential family and 
$\overline{\cal M}_\sbb$ is a mixture family,
we do not need to employ the em-algorithm, Algorithm \ref{protocol1-0},
because it can be solved directly as follows.
Since the generating vectors of $\overline{\cal E}_\sbb$ is the same as that of
$\overline{\cal M}_\sbb$,
 the intersection \eqref{XOP} can be calculated by solving
the following minimization.
That is, the following method finds the element in $\overline{\cal M}_\sbb$ among elements in $\overline{\cal E}_\sbb$.
Define $\bar\theta_{\see,*}$ as
\begin{align}
\bar\theta_{\see,*}:=\argmin_{\theta_3\in \mathbb{R}^{k-2l}
}F_{{\cal E},\sbb}(\theta_{\sbb,*}+ V_4 \theta_\see),\Label{XOE}
\end{align}
where $\theta_{\sbb,*}$ is defined by \eqref{MMZ} and \eqref{ZMP}.
Then, we have 
$\nabla^{(e)}[F_{{\cal E},\sbb}](\theta_{\sbb,*}+ 
V_4 \bar\theta_{\see,*}) V_4=0$, which implies 
$\theta_{\sbb,*}+ V_4 \bar\theta_{\see,*} 
\in \overline{\cal M}_\sbb$.
Thus,
\begin{align}
\bar\theta_{\sbb,*}:=\theta_{\sbb,*}+ V_4 \bar\theta_{\see,*}
\in  \overline{\cal E}_\sbb \cap \overline{\cal M}_\sbb.
\Label{MZSS}
\end{align}
Therefore, the statement of Corollary \ref{Cor5} is rewritten as follows.

\begin{theorem}\Label{Cor6-1}
Assume Conditions (B1), (B3), (B4), (B5), (B6+), and $\Theta=\mathbb{R}^d$.
We choose $\bar\theta_{\sbb,*}$ 
by combining \eqref{MZSS} and the solution of \eqref{XOE}.
Also, we choose 
$\bar\eta_{\saa,*} \in \mathbb{R}^l$ as
$\nabla^{(e)}[F_{{\cal E},\sbb}](\bar\theta_{\sbb,*})
=\bar\eta_{\saa,*} V_3$. 
When $\bar\eta_{\saa,*} \in \mathbb{R}^l$ 
belongs to the image of
$\nabla^{(e)}[F_{{\cal E},\saa}]$,
there exists $\bar\theta_{\saa,*}\in \Theta_{{\cal E},\saa}$ 
to satisfy the condition \eqref{XOP2}, i.e.,
$\nabla^{(e)}[F_{{\cal E},\sbb}](\bar\theta_{\sbb,*})
=  \nabla^{(e)}[F_{{\cal E},\saa}](\bar\theta_{\saa,*})V_3$, and 
the parameter $(\bar\theta_{\saa,*},\bar\theta_{\sbb,*})$ 
is invariant for the map
$\Pro^{(e),F}_{{\cal M}}\circ \Pro^{(m),F}_{{\cal M}\to{\cal E}}$.
When Condition (B2) holds additionally,
the parameter $(\bar\theta_{\saa,*},\bar\theta_{\sbb,*})$ 
is the solution of 
the maximum in \eqref{eq:Inf.st}.
 \end{theorem}

Due to Corollary \ref{Cor5},
the existence of the maximum in \eqref{eq:Inf.st}, Condition (E1), 
is equivalent to the existence of 
$\bar\eta_{\saa,*} \in \mathbb{R}^l$ that belongs to the image of $\nabla^{(e)}[F_{{\cal E}_\saa}]$.
That is, although an element $\bar\theta_{\sbb,*} \in 
\Theta_{{\cal E}_\sbb}$ exists,
there is a possibility that no element $\bar\theta_{\saa,*}
\in \Theta_{{\cal E}_\saa}$ 
 satisfies the condition \eqref{XOP2} with $\bar\theta_{\sbb,*}$.
Therefore, the method of this subsection works only when 
the maximum in \eqref{eq:Inf.st} exists.
That is, when the maximum does not exist in \eqref{eq:Inf.st},
the non-iterative method does not work at all.
Instead of the non-iterative method, as proven in Theorem \ref{theo:conv:BBem},
the iterative algorithms in the previous subsection work 
even when the maximum does not exist in \eqref{eq:Inf.st}.

Now, we compare the minimization \eqref{XOE} with the original reverse em-problem \eqref{eq:Inf.st}.
The minimization \eqref{XOE} is given as the minimization of the convex function
$F_{{\cal E},\sbb}$.
This objective function $F_{{\cal E},\sbb}$ has a simpler form 
than the objective function of the original reverse em-problem \eqref{eq:Inf.st}
because it is a part of the potential function to define the exponential family 
${\cal E}$. 
Further,
the number of free parameters in the minimization \eqref{XOE} is $k-2l$.
When $k<3l $, the number of free parameters in this method is smaller than 
the number of free parameters of the original reverse em-problem.
Depending on the situation, this method reduces the number of free parameters.
In particular, when $k=2l$, 
the matrix $V_3$ is a square matrix of size $l$
and
we do not need to solve the minimization \eqref{XOE}
as follows.
In this case, when the rank of $V_3$ is $l$, $\Ker V_3$ is $\{0\}$, which implies 
$V_4=0$.
Hence, as the special case with $k=2l$, i.e., the case when the number of parameters in ${\cal E}$
is twice of that of ${\cal M}$, 
we have the following corollary.
\begin{corollary}\Label{Cor6}
Assume Conditions (B1), (B3), (B4), (B5), (B6+), $\Theta=\mathbb{R}^d$,
and $k=2l$.
Then, the intersection $\overline{\cal E} \cap \overline{\cal M}
$ is given as $\{\theta_{\sbb,*}\}$, where
$\theta_{\sbb,*}$ is defined by \eqref{MMZ} and \eqref{ZMP}.
In addition, when there exists $\bar\theta_{\saa,*} 
\in \Theta_{{\cal M}} $ such that
\begin{align}
\nabla^{(e)}[F_{{\cal E},\sbb}](\theta_{\sbb,*})
=  \nabla^{(e)}[F_{{\cal E},\saa}](\bar\theta_{\saa,*})V_3, \Label{XOP2T}
\end{align}
the relation \eqref{SXO} holds under the choice 
$\bar\theta_{\sbb,*}=\theta_{\sbb,*}$.
 \end{corollary}
That is, Corollary \ref{Cor6} shows a simple calculation method for $k=2l$.
However, it works when an element $\bar\theta_{\saa,*} \in \Theta_{{\cal M}} $ to
satisfy \eqref{XOP2T} exists.
In addition, as explained in Sections \ref{S10} and \ref{S12},
the algorithms in the reference \cite{exact} are special cases of 
the method based on Theorem \ref{Cor6-1}.
Hence, this method can be considered as an extension of algorithms in the reference \cite{exact}.

Here, we notice that 
Condition (B2) can be replaced by the unique existence of the solution of the maximization \eqref{eq:Inf.st}
in the discussions in Subsections \ref{S47} and \ref{S47B}.
When we drop this condition, 
the methods in Subsections \ref{S47} and \ref{S47B} work for finding the 
local maximizer of the maximization \eqref{eq:Inf.st}.

\section{Capacity of classical channel }\Label{S10}
\subsection{Problem setting}
Let $\X:=\{1,\ldots,n_1 \}$ and $\Y:= \{1,\ldots, n_2 \}$ be finite sets. 
We call a map  $ W  :\X  \rightarrow  {\cal P}_{\Y}$ a channel from $\X$ to $\Y$. 
We use the notation $W_x(y):=W(y|x)$.
For $q \in {\cal P}_{\X}$ and $r \in {\cal P}_{\Y}$, 
$W \cdot q \in {\cal P}_{\Y}$,
$W \times q \in {\cal P}_{\X \times \Y}$, and $q \times r \in \mathcal{P}_{\X \times \Y}$ 
are defined by 
$(W \cdot q)(x,y):=\sum_{x \in \X}W(y|x)q(x)$, 
$(W \times q)(x,y):=W(y|x)q(x)$, and $(q \times r)(x,y) := q(x)r(y)$ respectively. 
The channel capacity of a channel  $W$ is given by 
\begin{align}
\max_{q \in {{\cal P}_{\X}} }D(W \times q 
\| (W \cdot q ) \times q )
=\max_{q \in {{\cal P}_{\X}} } \min_{ q'\in {{\cal P}_{\X}} , q'' \in {{\cal P}_{\Y}} }
 D(W \times q \| q'' \times q' ).\Label{MOA}
\end{align}

As explained in Subsection \ref{S54},
the set of product distributions $q'' \times q'$ forms an exponential subfamily 
${\cal E}$ and 
the set of distributions $W \times q$ forms a mixture subfamily ${\cal M}$.
That, the maximization problem \eqref{MOA} is a special case of the maximization \eqref{eq:Inf.st}
with $k=n_1-1$, $l=n_1+n_2-2$, and $d=n_1 n_2-1$.
In the following, we apply Algorithm \ref{protocol2}.
For this aim, 
we need to choose a suitable coordinate to satisfy conditions (B1), (B3), (B4), and (B5)
and check Condition (B2).

\subsection{Constructions of vectors $u_1, \ldots, u_{n_1 n_2-1},v_1,\ldots, v_{n_1+n_2-2}$}
To choose a suitable coordinate to satisfy conditions (B3), (B4), and (B5), we need to choose 
suitable vectors $u_1, \ldots, u_{n_1n_2},v_1,\ldots, v_{n_1+n_2-1}$.
For this aim, we define various functions on ${\cal Y}$ and ${\cal X}\times {\cal Y}$.

First, we choose $n_2-1$ linearly independent functions $f_j$ on ${\cal Y}$ for $j=1, \ldots, n_2-1$ to satisfy the condition that
\begin{align}
\sum_{y \in {\cal Y}} f_j(y) W_{n_1}(y)=0
\end{align}
and the linear space spanned by $f_1, \ldots, f_{n_2-1}$ does not contain a constant function.
As a typical case, $f_j$ can be chosen as follows.
\begin{align}
f_j(y):=
\left\{
\begin{array}{ll}
W_{n_1}(j+1) & \hbox{when } y=j \\
-W_{n_1}(j) & \hbox{when } y=j+1 \\
0 & \hbox{otherwise}.
\end{array}
\right.
\end{align}
Then, we define the functions $\xi_1, \ldots, \xi_{n_1n_2-1}$ on ${\cal X}\times {\cal Y}$ as follows.
\begin{align}
\xi_i(x,y):=& \delta_i(x) \\ 
\xi_{n_1-1+(i-1)(n_2-1)+j }(x,y):=& (f_j(y)-h_{i,j})\delta_i(x) \\
\xi_{(n_1-1)n_2+j }(x,y):=& f_j(y)\delta_{n_1}(x)
\end{align}
for $i=1, \ldots, n_1-1$ and $j=1, \ldots, n_2-1$, 
where we define
\begin{align}
h_{i,j}:=\sum_{y} f_j(y)W_{i}(y)\Label{NUM2}
\end{align}
for $i=1, \ldots, n_1$ and $j=1, \ldots, n_2-1$.

Then, we define the $C^{\infty}-$strictly convex function 
$F$ on $\mathbb{R}^{n_1n_2-1}$ as
\begin{align}
F(\theta):= \log \sum_{x,y} e^{\sum_{i=1}^{n_1 n_2-1} \theta^i \xi_i(x,y) }
\Label{MMLA}
\end{align}
That is,
we consider the Bregman divergence system $(\mathbb{R}^{d}, F, D^F)$.
We define the distribution $P_{\theta,XY}, P_{\theta,X}, P_{\theta,Y}$ as
\begin{align}
P_{\theta,XY}(x,y) &:=e^{\sum_{i=1}^{n_1 n_2-1} \theta^i \xi_i(x,y)-F(\theta) }, \\
P_{\theta,X}(x) &:=\sum_{y} e^{\sum_{i=1}^{n_1 n_2-1} \theta^i \xi_i(x,y)-F(\theta) },\\
P_{\theta,Y}(y) &:=\sum_{x}e^{\sum_{i=1}^{n_1 n_2-1} \theta^i \xi_i(x,y)-F(\theta) }. 
\end{align}
Then, as a special case of \eqref{MGA}, we have
\begin{align}
D^F(\theta\|\theta')= D(P_{\theta,XY} \| P_{\theta',XY} ).
\end{align}

Next, we choose the matrix $U$ as the identity matrix,
and $u_1, \ldots, u_{n_1 n_2-1}$ are chosen as its $n_1 n_2-1$ column vectors.
Then, we define vectors $v_1, \ldots, v_{n_1+n_2-2}$ as follows, whereas $V=(v_1, \ldots, v_{n_1+n_2-2})$.
\begin{align}
v_i:=& u_i \hbox{ for }i=1, \ldots,n_1-1, \\
v_{n_1-1+j}:=& 
\sum_{i=1}^{n_1} u_{n_1-1+(i-1)(n_2-1)+j }+ \sum_{i=1}^{n_1-1} h_{i,j} u_i
\hbox{ for }j=1, \ldots,n_2-1.\Label{XL3}
\end{align}
Then, we have
\begin{align}
\sum_{i=1}^{n_1n_2-1}\xi_i (x,y) v_j^i=
\left\{
\begin{array}{ll}
\delta_j(x) & \hbox{when }j=1, \ldots, n_{1}-1 \\
f_{j-n_1+1 }(y) & \hbox{when }j=n_1, \ldots, n_{1}+n_2-1 .
\end{array}
\right.
\end{align}
The case with $j=n_1, \ldots, n_{1}+n_2-1 $ can be shown as follows.
For $j=1, \ldots, n_2-1 $, we have
\begin{align}
&\sum_{i=1}^{n_1n_2-1}\xi_i (x,y) v_{n_1-1+j}^i \nonumber \\
=&
\sum_{i=1}^{n_1n_2-1}\xi_i (x,y) 
\Big(\sum_{i'=1}^{n_1} u^i_{n_1-1+(i'-1)(n_2-1)+j }+ \sum_{i'=1}^{n_1-1} h_{i',j} u_{i'}^i\Big)\nonumber \\
=&
\sum_{i'=1}^{n_1}\xi_{n_1-1+(i'-1)(n_2-1)+j } (x,y) 
+
\sum_{i'=1}^{n_1-1}\xi_{i'} (x,y) 
 h_{i',j} \nonumber \\
=&
\xi_{(n_1-1) n_2+j } (x,y) 
+\sum_{i'=1}^{n_1-1}\xi_{n_1-1+(i'-1)(n_2-1)+j } (x,y) 
+
\sum_{i'=1}^{n_1-1}\xi_{i'} (x,y) 
 h_{i',j} \nonumber \\
=&
f_j(y) \delta_{n_1}(x)
+\sum_{i=1}^{n_1-1}(f_j(y)-h_{i,j})\delta_i(x)
+\sum_{i=1}^{n_1-1}\delta_{i} (x) h_{i,j} \nonumber \\
=&
\sum_{i=1}^{n_1} f_j(y) \delta_i(x)=f_j(y).
\end{align}

\subsection{Parameterizations of ${\cal E}$ and ${\cal M}$}
Using 
\begin{align}
F_{{\cal E},\saa}(\theta^1, \ldots,\theta^{n_1-1})
&:= \log \sum_{x} e^{\sum_{i=1}^{n_1 -1} \theta^i \delta_i(x)}\Label{NCV}\\
F_{{\cal E},\sbb}(\theta^{n_1}, \ldots,\theta^{n_1+n_2-2})
&:= 
\log \sum_{y} e^{\sum_{j=1}^{n_2 -1} \theta^{n_1-1+j} f_j(y) },\Label{NCV2}
\end{align}
we define the distributions on $\X$ and $\Y$ as
\begin{align}
\bar{P}_{\theta_\saa,X}(x):=&e^{\sum_{i=1}^{n_1 -1} \theta^i \delta_i(x) -
F_{{\cal E},1}(\theta_\saa)} \\
\bar{P}_{\theta_\sbb,Y}(y):=&e^{\sum_{j=1}^{n_2 -1} \theta^{n_1-1+j} f_j(y) 
-F_{{\cal E},\sbb}(\theta_\sbb)}\Label{BNL}
\end{align}
for 
$\theta_\saa:=(\theta^1, \ldots,\theta^{n_1-1})$
and 
$\theta_\sbb:=(\theta^{n_1}, \ldots,\theta^{n_1+n_2-2})$.
Then, we have
\begin{align}
P_{\sum_{j=1}^{n_1-1}\theta_\saa^j v_j
+\sum_{j'=1}^{n_2-1}\theta_\sbb^j 
v_{n_1-1+j},XY}
=\bar{P}_{\theta_\saa,X}\times \bar{P}_{\theta_\sbb,Y}.
\end{align}
Hence, the set of product distributions is written as the exponential subfamily
${\cal E}:=\{ P_{\sum_{j=1}^{n_1+n_2-2} \bar{\theta}^j v_j ,XY}\}$
generated by $v_1, \ldots, v_{n_1+n_2-2}$ 
at the point $(0, \ldots, 0)$.
Then, we have $F_{{\cal E}}(\theta_\saa,\theta_\sbb)=F_{{\cal E},\saa}(\theta_\saa)+
F_{{\cal E},\sbb}(\theta_\sbb)$.


We define the mixture family ${\cal M}$ by the constraint
$\sum_{i=1}^{n_1 n_2-1} u^i_{n_1-1+j'} \partial_i F(\theta)=0$ 
for $j'=1, \ldots, n_1 (n_2-1)$.
This constraint is equivalent to 
\begin{align}
\sum_{y }  (f_j(y)-h_{i',j})  P_{\theta,XY}(i',y)=0 , \quad
\sum_{y }  f_j(y) P_{\theta,XY}(n_1,y)=0 
\end{align}
for $i'=1, \ldots, n_1-1$ and $j=1, \ldots, n_2-1$.
Hence, the mixture family ${\cal M}$ is composed of distributions with the form $W \times q$.
Thus, the problem \eqref{MOA} is written as the problem \eqref{eq:Inf.st} with the above defined ${\cal E}$ and ${\cal M}$.
The conditional probability 
$P_{\theta,Y|X}(y|i)=
\frac{P_{\theta,XY}(i,y)}{\sum_{y'}P_{\theta,XY}(i,y')}$ depends only on 
$(\theta^{n_1-1+(i-1)(n_2-1)+j})_{j=1}^{n_2-1}$ for $i=1, \ldots, n_1$.
Since
\begin{align}
P_{\theta,Y|X}(y|i)=&
\frac{e^{\theta^i +\sum_{j=1}^{n_2-1}  \theta^{n_1-1+ (i-1) (n_2-1)+j} 
\xi_{n_1-1+ (i-1) (n_2-1)+j}(i,y)}}
{\sum_{y'}e^{\theta^i +\sum_{j'=1}^{n_2-1}  \theta^{n_1-1+ (i-1) (n_2-1)+j' } 
\xi_{n_1-1+ (i-1) (n_2-1)+j'}(i,y')}}\nonumber \\
=&
\frac{e^{\sum_{j=1}^{n_2-1}  \theta^{n_1-1+ (i-1) (n_2-1)+j} 
(f_j(y)-h_{i,j})
}}
{\sum_{y'} e^{\sum_{j'=1}^{n_2-1}  \theta^{n_1-1+ (i-1) (n_2-1)+j'} 
(f_j(y')-h_{i,j'})
}},
\end{align}
we choose $\theta_\sbb^{\dagger}=
(\theta^{n_1,\dagger}, \ldots, \theta^{n_1n_2-1,\dagger})$ as
\begin{align}
W_i(y)=
\frac{e^{\sum_{j=1}^{n_2-1}  \theta^{n_1-1+ (i-1) (n_2-1)+j,\dagger } 
(f_j(y)-h_{i,j})
}}
{\sum_{y'} e^{\sum_{j'=1}^{n_2-1}  \theta^{n_1-1+ (i-1) (n_2-1)+j',\dagger } 
(f_{j'}(y')-h_{i,j'})
}}.
\end{align}
In this choice, we have
\begin{align}
\log \sum_{y'} e^{\sum_{j'=1}^{n_2-1}  \theta^{n_1-1+ (i-1) (n_2-1)+j',\dagger } 
(f_{j'}(y')-h_{i,j'})
}=H(W_i)\Label{XMZ}
\end{align}
because 
\begin{align}
\sum_{y'} 
(f_{j}(y')-h_{i,j})
e^{\sum_{j'=1}^{n_2-1}  \theta^{n_1-1+ (i-1) (n_2-1)+j',\dagger } 
(f_{j'}(y')-h_{i,j'})
}=0
\end{align}
for $j=1, \ldots,n_2-1$.

Then, ${\cal M}$ is written as
$\{( \theta_\saa, \theta_\sbb^{\dagger})| \theta_\saa \in \mathbb{R}^{n_1-1}\}$.
That is,
${\cal M}$ forms an exponential subfamily generated by $u_1, \ldots, u_{n_1-1}$.
Using \eqref{MMLA}, the function $F_{{\cal M}}$ is written as 
\begin{align}
F_{{\cal M}}(\theta_\saa)
=F(\theta_\saa,\theta_\sbb^{\dagger})
= \log \sum_{x,y} e^{\sum_{i=1}^{n_1 -1}\theta^i \delta_i(x)
+\sum_{i=n_1}^{n_1 n_2-1} \theta^{i,\dagger} \xi_i(x,y) }
\end{align}

Hence, the maximization \eqref{MOA} is rewritten as 
\begin{align}
\sup_{q \in {{\cal P}_{\X}} }D(W \times q 
\| q \times (W \cdot q ))
=
\max_{\theta \in \mathcal{M}} D^{F}(\theta \| \Pro^{(m),F}_{\mathcal{E}} (\theta))
=
\max_{\theta \in \mathcal{M}} 
\min_{\theta' \in \mathcal{E}} 
D^{F}(\theta \| \theta').
\end{align}

\subsection{Check of Conditions (B1), (B2), (B3), (B4), and (B5)}
\Label{S54}
Lemma \ref{LOS} guarantees Condition (B1).
We define the $(n_1-1) \times (n_2-1)$ matrix $H:=(h_{i,j})$.
Then, the relation \eqref{XL3} guarantees that the $ (n_1-1)\times (n_1+n_2-2)$ matrix $V_1$ is $ (I,H)$.
That is, the $(n_1-1) \times (n_2-1)$ matrix $V_3$ is $H$.
Hence, Conditions (B3) and (B4) hold.
Since 
the exponential family ${\cal E}$ satisfies
$F_{{\cal E}}(\bar{\theta}_\saa,\bar{\theta}_\sbb)=
F_{{\cal E},1}(\bar{\theta}_\saa)
+F_{{\cal E},2}(\bar{\theta}_\sbb)$,
we obtain Condition (B5).
Therefore, we can apply 
Algorithm \ref{protocol2} with Condition (B5).
Therefore, we can apply Algorithms \ref{protocol2F} and \ref{protocol2}
to calculate the maximum \eqref{MOA}.

As we have
\begin{align}
P_{\Pro^{(m),F}_{{\cal E}}  (\theta),XY}= P_{\theta,X}\times P_{\theta,Y}
\end{align}
for any $\theta$, 
we have
\begin{align}
 D^{F}(\theta'\|\theta)
=&
 D(P_{\theta,XY}\|P_{\theta',XY})
=
 D(P_{\theta,X}\| P_{\theta',X})\nonumber \\
 \le &
 D( P_{\theta,X}\| P_{\theta',X})
+D(P_{\theta,Y}\| P_{\theta',Y}) \nonumber \\
=&
 D(P_{\theta,X}\times P_{\theta,Y}\| P_{\theta',X}\times P_{\theta',Y})
=
 D(P_{ \Pro^{(m),F}_{{\cal E}}  (\theta)}\| P_{ \Pro^{(m),F}_{{\cal E}}  (\theta')} ) 
\nonumber \\
=&
 D^{F}( \Pro^{(m),F}_{{\cal E}}  (\theta)\|  \Pro^{(m),F}_{{\cal E}}  (\theta') ) 
\end{align}
for $\theta,\theta' \in {\cal M}$.
Thus, condition (B2) holds.
Therefore,
Theorem \ref{theo:conv:BBem} guarantees the global convergence.
When $\theta^{(1)}$ is 
$W \times P_{\uni,X}$ with 
the uniform distribution $P_{\uni,X}$ on ${\cal X}$,
we have 
\begin{align}
\sup_{\theta \in \mathcal{M}} D^F(\theta \| \theta^{(1)})
=&\sup_{q \in {\cal P}({\cal X})}D(W \times q\|W \times P_{\uni,X}) \nonumber \\
=&\sup_{q \in {\cal P}({\cal X})}D(q\|P_{\uni,X})
= \log n_1.\Label{CKP}
\end{align}
Therefore, 
when Theorem \ref{theo:conv:BBem} is applied,
we obtain the precision \eqref{NHG}
with $ \frac{\log n_1}{\epsilon}$ iterations.
Also, we can apply Theorem \ref{conv:BBem} to the error evaluation in Algorithm \ref{protocol2}.

With the above choice of $\theta^{(1)}$,
we consider the case when the distributions $\{W_x\}_{x}$ are linearly independent.
We have $ D^{F}( \Pro^{(m),F}_{{\cal E}}  (\theta^*)\|  \Pro^{(m),F}_{{\cal E}}  (\theta) ) 
- D^{F}(\theta^*\|\theta)=D(W\cdot P_{\theta^*,X} \| W\cdot P_{\theta,X})$.
Since the set $\{\theta \in \mathcal{M} | 
D( P_{\theta^*,X} \|  P_{\theta,X}) \le 
D( P_{\theta^*,X} \|  P_{\theta^{(1)},X}) \}$ is compact
and $D(W\cdot P_{\theta^*,X} \| W\cdot P_{\theta,X})>0$,
there exists $\alpha>0$ such that
$\frac{D(W\cdot P_{\theta^*,X} \| W\cdot P_{\theta,X})}{D( P_{\theta^*,X} \|  P_{\theta,X})}
\ge \alpha$
for $\theta \in \{\theta \in \mathcal{M} | 
D( P_{\theta^*,X} \|  P_{\theta,X}) \le 
D( P_{\theta^*,X} \|  P_{\theta^{(1)},X}) \}$.
This condition implies the condition (B2+).
Hence, we can apply Theorem \ref{theo:conv:BBem+} instead of Theorem \ref{theo:conv:BBem}.
When $\theta^{(1)}$ is the uniform distribution on ${\cal X}$,
we obtain the precision \eqref{NHG+}
with $\frac{ \log  \log n_1  -\log \epsilon}{\log (1+\alpha)}$ iterations. 

However, each step in Algorithms \ref{protocol2F} and  \ref{protocol2} contains a minimization problem.
Unfortunately, this minimization 
requires convex minimization.
Since Arimoto-Blahut algorithm \cite{Arimoto,Blahut} has a simple procedure in each step,
the application of these methods to the classical channel capacity 
does not have an advantage over existing methods.

\begin{remark}
As shown in the end of Section 4 of Toyota \cite{Shoji},
the algorithm by Arimoto \cite{Arimoto} and Blahut \cite{Blahut}
does not use 
the inverse map $(\Pro^{(e),F}_{{\cal M}}\circ \Pro^{(m),F}_{{\cal E}}|_{{\cal M}})^{-1}$ in each iteration.
Toyota \cite{Shoji} proposed to use 
the inverse map $(\Pro^{(e),F}_{{\cal M}}\circ \Pro^{(m),F}_{{\cal E}}|_{{\cal M}})^{-1}$ in each iteration instead of the original Arimoto-Blahut algorithm,
he did not derive the exact expression of the inverse map.
\end{remark}

\subsection{Non-iterative method}
Next, we characterize the maximization \eqref{MOA} without any iterative method. 
To check Condition (B6+),
we prepare the following lemmas.
\begin{lemma}\Label{LC3B}
The relation
\begin{align}
P_{(\theta_\saa +\theta_\saa^{\dagger} ,
\theta_\sbb^{\dagger}),X}=
\bar{P}_{\theta_\saa ,X}\Label{ACU}
\end{align}
holds, where
$\theta_\saa^{\dagger}=
(\theta^{1,\dagger}, \ldots, \theta^{n_1-1,\dagger})$
is defined as
$\theta^{i,\dagger}:=-H(W_i)+H(W_{n_1})$
for $i=1\ldots,n_1-1$.
\end{lemma}

\begin{proof}
We define
$\theta_{(i)}^{\dagger}=
(\theta^{n_1-1+(i-1)(n_2-1)+1 ,\dagger}, \ldots, 
\theta^{n_1-1+i(n_2-1),\dagger}) \in \mathbb{R}^{n_1-1}$.
Since $W_i=\bar{P}_{ \theta_{(i)}^{\dagger} ,Y}$, 
we have $W_i(y)= e^{
\sum_{j=1}^{n_2-1}\theta^{n_1-1+(i-1)(n_2-1)+j,\dagger}
(f_j(y)-h_{i,j}) - F_{{\cal E},2}(\theta_{(i)}^{\dagger})} $.
Because
\begin{align}
e^{F_{{\cal E},\sbb}(\theta_{(i)}^{\dagger})}
=\sum_{y \in \Y} e^{
\sum_{j=1}^{n_2-1}\theta^{n_1-1+(i-1)(n_2-1)+j,\dagger}
(f_j(y)-h_{i,j})},
\end{align}
\eqref{XMZ} implies the relation
\begin{align}
H( W_i)= -\sum_{y} W_i(y)\log W_i(y) 
=F_{{\cal E},2}(\theta_{(i)}^{\dagger})\Label{XOS}
\end{align}
for $i=1, \ldots, n_1$.

Now, we choose an element $\theta_1' \in \mathbb{R}^{n_1-1}$ such that
\begin{align}
P_{(\theta_\saa +\theta_\saa^{\dagger} ,
\theta_\sbb^{\dagger}),X}=
\bar{P}_{\theta_\saa' ,X}.
\end{align}
Since we have
\begin{align}
P_{(\theta_\saa +\theta_\saa^{\dagger} ,
\theta_\sbb^{\dagger}),X}(n_1)=
e^{ F_{{\cal E},\sbb}(\theta_{(n_1)}^{\dagger})
-F_{{\cal M}}(\theta_\saa +\theta_\saa^{\dagger} ,\theta_\sbb^{\dagger})} ,
\end{align}
the relation $P_{(\theta_\saa +\theta_\saa^{\dagger} ,\theta_\sbb^{\dagger}),X}(n_1)=
\bar{P}_{\theta_\saa' ,X}(n_1)$ yields 
\begin{align}
e^{ F_{{\cal E},\sbb}(\theta_{(n_1)}^{\dagger})
-F_{{\cal M}}(\theta_\saa +\theta_\saa^{\dagger} ,\theta_\sbb^{\dagger})} 
=
e^{ -F_{{\cal E},\saa}(\theta_\saa')}.\Label{XOCS}
\end{align}
For $x\neq n_1$, we have 
\begin{align}
&\bar{P}_{\theta_\saa' ,X}(x)
=P_{(\theta_\saa +\theta_\saa^{\dagger} ,\theta_\sbb^{\dagger}),X}(x)=
e^{\theta_\saa^x+ \theta_\saa^{x,\dagger}
+ F_{{\cal E},\sbb}(\theta_{(x)}^{\dagger})
-F_{{\cal M}}(\theta_\saa +\theta_\saa^{\dagger} ,\theta_\sbb^{\dagger})} \nonumber \\
\overset{(a)}=&
e^{\theta_\saa^x+ \theta_\saa^{x,\dagger}
+ F_{{\cal E},\sbb}(\theta_{(x)}^{\dagger})
-F_{{\cal E},\sbb}(\theta_{(n_1)}^{\dagger})
-F_{{\cal E},\saa}(\theta_\saa')
}\overset{(b)}=e^{\theta_\saa^x-F_{{\cal E},\saa}(\theta_\saa')} ,
\end{align}
where $(a)$ and $(b)$ follow from 
\eqref{XOCS} and
the pair of \eqref{XOS} and the definition of $\theta_\saa^{x,\dagger}$, respectively.
This relation shows \eqref{ACU}.
\end{proof}

In the same way as the end of the previous subsection, 
we assume that the distributions $\{W_x\}_{x}$ are linearly independent.
Then, the rank of $H$ is $n_1-1$.
The combination of this fact and Lemma \ref{LC3B}
guarantees 
\begin{align}
\nabla^{(m)}[F_{{\cal M}}^*](\eta_\saa)-
\nabla^{(m)}[F_{{\cal E},\saa}^*](\eta_\saa)
=\theta_\saa +\theta_\saa^{\dagger}
-\theta_\saa =\theta_\saa^{\dagger},
\end{align}
which implies Condition (B6+).
We choose the parameter $\theta_\sbb^{\ddagger} \in \mathbb{R}^{n_2-1}$ such that
$H \theta_\sbb^{\ddagger}=\theta_\saa^{\dagger} $.
We choose the $(n_2-1)\times (n_2-n_1)$ matrix $G$ such that 
$\im G=\Ker H$.
Then, ${\cal E}_2$ and ${\cal M}_2$ defined in \eqref{XC1} and \eqref{XC2}
are written as
\begin{align}
{\cal E}_\sbb &=\{ \theta_\sbb^{\ddagger}
+ G \theta_\see| \theta_\see 
\in \mathbb{R}^{n_2-n_1}\},\\
{\cal M}_\sbb & = \{ \theta_\sbb \in \mathbb{R}^{n_2-1}|
\nabla^{(e)}[F_{{\cal E},\sbb}](\theta_{2})G =0 \}.
\end{align}
As explained in Subsection \ref{S47},
the intersection ${\cal E}_\sbb\cap {\cal M}_\sbb$ is composed of a unique element.
As the solution of the following minimization \eqref{XOEB},
we choose $\theta_\see^{\ddagger}$ as
\begin{align}
\theta_{\see}^{\ddagger}
:=\argmin_{\theta_\see\in \mathbb{R}^{n_2-n_1}}F_{{\cal E},2}(\theta_{\sbb}^{\ddagger}+ G \theta_\see),\Label{XOEB}
\end{align}
Then, we set
$ \bar\theta_\sbb^{\ddagger}:=\theta_\sbb^{\ddagger}+ G \theta_\see^{\ddagger} \in {\cal E}_\sbb\cap {\cal M}_\sbb$.
Then, we have the following corollary of Theorem \ref{Cor6-1}.
\begin{corollary}\Label{corB7}
When there exists 
$\theta_\saa^{\ddagger} \in \mathbb{R}^{n_1-1} $ such that
$ \nabla^{(e)}[F_{{\cal E},\sbb}](\bar\theta_{\sbb}^{\ddagger})=
\nabla^{(e)}[F_{{\cal E},\saa}](\theta_{\saa}^{\ddagger})H $,
which is equivalent to
\begin{align}
W\cdot \bar{P}_{\theta_{\saa}^{\ddagger}, X} = 
\bar{P}_{\bar\theta_\sbb^{\ddagger} ,Y}
\Label{ACA},
\end{align}
the maximizer in \eqref{CKP} is 
$(\theta_{\saa}^{\ddagger}+\theta_{\saa}^{\dagger},
{\theta}_{\sbb}^{\dagger} )\in {\cal M} $.
When the above condition holds,
the maximum \eqref{CKP} is 
\begin{align}
D(P_{(\theta_{\saa}^{\ddagger}+\theta_{\saa}^{\dagger},
\theta_{\sbb}^{\dagger} ),XY}
\|\bar{P}_{\theta_{\saa}^{\ddagger},X}\times
\bar{P}_{\bar\theta_{\sbb}^{\ddagger},Y}
)
=-H(W_{n_1})
+F_{{\cal E},\sbb}(\bar\theta_{\sbb}^{\ddagger}),\Label{XZM2}
\end{align}
\end{corollary}

The derivation of \eqref{XZM2} follows from the following calculation.
\begin{align}
&D(W_x\|\bar{P}_{\bar\theta_{\sbb}^{\ddagger},Y})
=\sum_y W_x(y) (\log W_x(y) - \log \bar{P}_{\bar\theta_{\sbb}^{\ddagger},Y}(y))\nonumber \\
=&-H(W_x)-\sum_y W_x(y) 
(\sum_{j=1}^{n_2-1}\bar\theta^{n_1-1+j,\ddagger}f_j(y)-F_{{\cal E},2}(\bar\theta_{\sbb}^{\ddagger}))\nonumber \\
=&-H(W_x)- \sum_{j=1}^{n_2-1}\bar\theta^{n_1-1+j,\ddagger}h_{x,j}
+F_{{\cal E},\sbb}(\bar\theta_{\sbb}^{\ddagger})\nonumber \\
=&-H(W_x)- \sum_{j=1}^{n_2-1}\theta^{n_1-1+j,\ddagger}h_{x,j}
+F_{{\cal E},\sbb}(\bar\theta_{\sbb}^{\ddagger})\nonumber \\
=&-H(W_x)- \theta^{x,\dagger}
+F_{{\cal E},\sbb}(\bar\theta_{\sbb}^{\ddagger})\nonumber \\
=&-H(W_x)- (-H(W_x)+H(W_{n_1}))
+F_{{\cal E},\sbb}(\bar\theta_{\sbb}^{\ddagger})\nonumber \\
=&-H(W_{n_1})
+F_{{\cal E},\sbb}(\bar\theta_{\sbb}^{\ddagger}).
\end{align}

When $n_1=n_2$, we have $l=n_1-1=n_2-1 $, which enables us to apply 
Corollary \ref{Cor6}.
In this case,
as another typical case, we can choose the functions $f_j$ such that 
$(f_j(i))_{1\le i,j\le n_2-1}$ is the inverse matrix of $(W_i(j))_{1\le i,j\le n_2-1} $.
Under this choice, $h_{i,j}$ is the identity matrix and the calculation of the maximization \eqref{MOA} based on Corollary \ref{Cor6} 
is done by Algorithm 1 in the reference \cite{exact}.
Therefore, the method based on Theorem \ref{Cor6-1}
can be considered as a generalization of Algorithm 1 in the reference \cite{exact}.
In addition, the above discussion shows that 
Algorithm 1 in the reference \cite{exact} can be characterized as finding the intersection of the exponential family ${\cal E}_2$ and the mixture family ${\cal M}_2$, which is an information geometrical characterization.


However, there is a case that no distribution $P_X $ on $\X$ satisfy \eqref{ACA}
because there does not necessarily exist 
$\theta_\saa^{\ddagger} \in \mathbb{R}^{n_1-1} $ such that
$ \nabla^{(e)}[F_{{\cal E},\sbb}](\bar\theta_{\sbb}^{\ddagger})=
\nabla^{(e)}[F_{{\cal E},\saa}](\theta_{\saa}^{\ddagger})H $.
In this case, instead of a distribution on $\X$,
there exists a function $f_X$ on $\X$ such that
\begin{align}
\sum_{x \in \X} f_X(x) W_x= 
\bar{P}_{\bar\theta_\sbb^{\ddagger} ,Y},~
\sum_{x \in \X} f_X(x)=1.
\end{align}
That is, the above function $f_X$ may take negative value(s).
Also, in this case,
there does not exist the maximum in \eqref{CKP},
and the maximum \eqref{MOA} is achieved in the boundary of $\P_\X$.
We denote the value \eqref{XZM2} by $\hat{C}(\X)$, define 
the subset
\begin{align}
{\cal N}(\X):=\{x \in \X| f_X(x) < 0\}.
\end{align}
When ${\cal N}(\X)$ is the empty set, 
$\hat{C}(\X)$ is the channel capacity.

\begin{algorithm}
\caption{Non-iterative algorithm for classical channel capacity
in the special case}
\Label{protocol1}
\begin{algorithmic}
\STATE {Step 1: Set the parameters
$h_{i,j}=\delta_{i,j}$ for $1 \le i \le n_1-1 $ and $1 \le j \le n_2-1$,
and $h_{n_1,j}=0$ for $1 \le j \le n_2-1$.
Choose $f_1, \ldots, f_{n_2-1}$ such that 
$h_{i,j}=\sum_{y} f_{j}(y)W_{i}(y)$.
Here, we use Algorithm \ref{protocol1T}.
} 
\STATE {Step 2: Set the parameter $\theta^{i,\dagger}
= -H(W_i)+H(W_{n_1})$ for $i=1, \ldots, n_1-1$.}
\STATE {Step 3: Define the function $F_{{\cal E},\sbb}(\theta_{\sbb}):=
\log \sum_{y} e^{\sum_{j=1}^{n_2 -1} \theta_\sbb^{j} f_j(y) }$ 
for
$\theta_{\sbb} \in \mathbb{R}^{n_2-1}$.}
\STATE {Step 4: Choose $\theta_{\see}^{\ddagger}\in 
\mathbb{R}^{n_2-n_1}$ as}
\begin{align}
\theta_{\see}^{\ddagger}
:=\argmin_{\theta_\see\in \mathbb{R}^{n_2-n_1}}
F_{{\cal E},\sbb}(\theta_{\saa}^{\ddagger}, \theta_\see).\Label{XOEBC}
\end{align}
\STATE {Step 5: Set
$ \bar\theta_\sbb^{\ddagger}:=(\theta_\saa^{\ddagger}, \theta_\see^{\ddagger}) \in {\cal E}_\sbb\cap {\cal M}_\sbb$, and calculate 
$P_{\bar\theta_\sbb^{\ddagger},Y}(y)$ by using \eqref{BNL}.}
\STATE {Step 6: 
Calculate $P_X$ by solving $\sum_{x}P_X(x)W_x(y)=
P_{\bar\theta_\sbb^{\ddagger},Y}(y)$ with the condition 
$\sum_{x}P_X(x)=1 $.
We output $-H(W_{n_1})+F_{{\cal E},\sbb}
(\bar\theta_{\sbb}^{\ddagger})$ 
and $\{ x \in \X| P_X(x)<0\}$ as $\hat{C}(\X)$ and ${\cal N}(\X)$, respectively.
In particular, if $P_X $ does not have a negative component, 
$\hat{C}(\X)$ is the capacity.
}
\end{algorithmic}
\end{algorithm}

\begin{algorithm}
\caption{Algorithm for finding $f_1, \ldots, f_{n_2-1}$
}
\Label{protocol1T}
\begin{algorithmic}
\STATE {Step 1: 
We reorder elements of $\Y$ such that vectors
$(W_{i}(y))_{y=1, \ldots, n_1-1}$ are linearly independent
for $i=1, \ldots, n_1-1$ and 
$W_{n_1}(n_1)>0$.}
\STATE {Step 2: 
We denote the inverse matrix of 
$(W_{i}(y)- \frac{W_i(n_1)W_{n_1}(y)}{W_{n_1}(n_1)}
)_{i,y=1, \ldots, n_1-1} $
by $c_{j,y}$, i.e., $ \sum_{y=1}^{n_1-1}c_{j,y}
(W_{i}(y)- \frac{W_i(n_1)W_{n_1}(y)}{W_{n_1}(n_1)}
)=\delta_{i,j}$.}
\STATE {Step 3: 
We set $f_{1}, \ldots f_{n_1-1}$ as
$f_{j}(y)=c_{j,y}$ for $y=1, \ldots, n_1-1$,
$f_{j}(n_1)=-\sum_{y=1}^{n_1-1}
c_{j,y}\frac{W_{n_1}(y)}{W_{n_1}(n_1)}$, 
and 
$f_{j}(y)=0$ for $y=n_1, \ldots, n_2$.}
\STATE {Step 4: 
We set $f_{n_1}, \ldots, f_{n_2-1}$ as follows.
We set $f_{j}(y)= \delta_{j+1,y}$ for $j=n_1, \ldots, n_2-1$
and $y=n_1, \ldots, n_2$.
We choose $f_{j}(y)$ for $j=n_1, \ldots, n_2-1$
and $y=1, \ldots, n_1-1$ as follows.
\begin{align}
 f_{j}(y)= -\sum_{i=1}^{n_1-1}\hat{c}_{i,y}W_i(j+1),
\Label{NX0}
\end{align}
where $(\hat{c}_{i,y})_{i,j=1, \ldots, n_1-1}$ is the inverse matrix
of $(W_i(y))_{i,j=1, \ldots, n_1-1}$.} 
\end{algorithmic}
\end{algorithm}

\subsection{Algorithms for non-iterative method}
Using Corollary \ref{corB7}, we have the following lemma.

\begin{lemma}
With the use of Algorithm \ref{protocol1T},
Algorithm \ref{protocol1} 
calculates $\hat{C}(\X)$ and ${\cal N}(\X)$.
\end{lemma}

\begin{proof}
In Algorithm \ref{protocol1},
for a simple calculation, we set the parameters $h_{i,j}$ in the way as Step 1.
The choice of functions $f_1, \ldots, f_{n_2-1}$ given in Step 1 
follows from \eqref{NUM2}.
The choice of $\theta^{i,\dagger}$ given in Step 2
follows from Lemma \ref{LC3B}.
The choice of $F_{{\cal E},\sbb}(\theta_{\sbb})$ given in Step 3
follows from \eqref{NCV2}.
The choice of $\theta_{\see}^{\ddagger}$ given in Step 4
follows from \eqref{XOEB}.
Then, Corollary \ref{corB7} guarantees that
the remaining part gives $\hat{C}({\cal X})$ and ${\cal N}({\cal X})$.

In addition, 
the output of Algorithm \ref{protocol1T} satisfies the requirement of 
Step 1 of Algorithm \ref{protocol1}, whose reason is the following.
For $i=1, \ldots, n_1$ and $j=1, \ldots, n_1$,
we have
$\sum_{y=1}^{n_2-1} f_{j}(y)W_i(y)
=\sum_{y=1}^{n_1-1} f_{j}(y)W_i(y)+f_{j}(n_1)W_i(n_1)
=\sum_{y=1}^{n_1-1} c_{j,y}W_i(y)-\sum_{y=1}^{n_1-1}
c_{j,y}\frac{W_{n_1}(y)}{W_{n_1}(n_1)}W_i(n_1)
=\delta_{i,j}$. 
For $i=1, \ldots, n_1$ and $j=n_1, \ldots, n_2-1$,
we have
$\sum_{y=1}^{n_2-1} f_{j}(y)W_i(y)
=\sum_{y=1}^{n_1-1} f_{j}(y)W_i(y)+W_i(j+1)=0$ due to \eqref{NX0}.
That is, the conditions in Step 1 of Algorithm \ref{protocol1} is satisfied.
\end{proof}

Although Algorithm \ref{protocol1} contains the minimization \eqref{XOEBC},
its objective function has a simpler form as defined in Step 3 than the mutual information.
Hence, even when the number of free parameters is large, 
 the minimization \eqref{XOEBC} can be easily calculated.

Algorithm 1 in the reference \cite{exact} covers only the case when 
$n_1=n_2$ and ${\cal N}(\X)$ is the empty set.
In this special case, 
Algorithm \ref{protocol1} coincides with 
Algorithm 1 in the reference \cite{exact}
while Step 4 of Algorithm \ref{protocol1} is a trivial procedure in this case. 

To see the case beyond Algorithm 1 in the reference \cite{exact},
we study the case when ${\cal N}(\X)$ is not the empty set.
In this case,
we need a more complicated procedure.
To handle this case, we expand the definitions of $\hat{C}(\X)$ and ${\cal N}(\X)$.
That is, we define $\hat{C}(\X_0)$ and ${\cal N}(\X_0)$ for a subset $\X_0\subset \X$ in the same way,
and they can be calculated by Algorithm \ref{protocol1}.
In this case, Algorithm \ref{protocol2Y} gives an algorithm to calculate the capacity.

In order to show this fact, we choose a subset $\X_{2,*} \subset \X_2$ as the support of the maximizer $q^{2,*}\in \P_{\X_2}$ of 
\eqref{MOA}
and denote the maximum value \eqref{MOA} by $C(\X_2)$
when $\X_2$ is substituted into $\X$.
In particular, when $\X_2=\X$, we denote $\X_{2,*} $ and $q^{2,*}$
by $\X_{*}$ and $q^{*}$, respectively.
To show the correctness of  Algorithm \ref{protocol2Y}, we prepare 
the following lemma.

\begin{lemma}\Label{LCVF}
The relation
\begin{align}
(\X_2 \setminus \X_{2,*}) \cap {\cal N}(\X_2)
\neq \emptyset
\end{align}
holds for any subset $\X_2 \subset \X$.
\end{lemma}

\begin{proof}
It is sufficient to show the desired statement for the case with $\X_2=\X$. Hence, we assume the relation $\X_2=\X$.

We define the mixture family
${\cal M}:=
\{\sum_{x \in \X_*} f_X(x) W_x \in \P_{\Y}|
\sum_{x \in \X_*} f_X(x)=1 \}$.
We denote the distribution 
$\bar{P}_{\bar\theta_\sbb^{\ddagger} ,Y}$
by $q^{**}$.
Then, we obtain $\Pro^{(e),F}_{{\cal M}}(q^{**}) \in {\cal M}$. 
Pythagorean theorem (Proposition \ref{MNL}) guarantees the relation
\begin{align}
D(W_x\| q^{*})
=
D(\Pro^{(e),F}_{{\cal M}}(q_{**})\| q^{*})+
D(W_x\| \Pro^{(e),F}_{{\cal M}}(q^{**}))
\end{align}
for $x \in X_*$.
Since $D(W_x\| \Pro^{(e),F}_{{\cal M}}(q^{**}))$
does not depend on $x \in X_*$,
$\Pro^{(e),F}_{{\cal M}}(q^{**})=q^*$. 
We choose the generator $g $ of the exponential family 
${\cal E}$
that connects
$q^{**}$ and $q^*$ as follows.
\begin{align}
q^{**}(y)=q^{*}(y)e^{g(y)-C} , \quad
\sum_{y \in \Y}q^{*}(y)g(y)=0,
\end{align}
where $ C:=\log \sum_{y \in \Y} q^{*}(y)e^{g(y)}$.
Hence, we have
\begin{align}
\sum_{y \in \Y}q^{**}(y)g(y)>0.\Label{NZK}
\end{align}

We define the hyperplane
${\cal M}_c:=\{P\in \P_\Y| \sum_{y \in \Y} g(y)P(y)=c\}$.
We denote the unique element of $ {\cal M}_c \cap {\cal E}$
by $q_c $.
Due to \eqref{NZK}, $q^{**}$ is written as $q_{t}$ with a 
positive number $t$.

Since this exponential family is orthogonal to ${\cal M}$, 
$\sum_{y \in \Y} g(y)W_x(y)=0$ for any element $x \in X_*$, i.e., 
${\cal M} \subset {\cal M}_0$.
For $x \in \X \setminus \X_*$, we choose $c(x)$ such that
$W_x \in {\cal M}_{c(x)}$.
Pythagorean theorem (Proposition \ref{MNL}) guarantees 
the relation
\begin{align}
D(W_x\| q^{*})
=
D(q_{c(x)}\| q^{*})+ D(W_x\| q_{c(x)})
\end{align}
for $x \in \X \setminus X_*$.
Since \begin{align}
D(W_x\| q^{*})
\le D(W_x\| q^{**})
=
D(q_{c(x)}\| q^{**})+ D(W_x\| q_{c(x)}),
\end{align}
we have $ D(q_{c(x)}\| q_{0})=D(q_{c(x)}\| q^{*})
\le D(q_{c(x)}\| q^{**})=D(q_{c(x)}\| q_{t})
$. Hence, $c(x)<0$ because $t>0$.

Now, we write $q^{**} $ and $q^*$ as
$q^{**}=\sum_{x \in \X} v_1(x) W_x$ and
$q^{*}=\sum_{x \in \X_*} v_2(x) W_x$ 
by using a distribution $v_2$ on $\X_*$ and 
a function $v_1$ with the condition
$\sum_{x \in \X} v_1(x)=1 $.
Using a function $v_3$, $q^{**}-q^*$ is written as 
\begin{align}
q^{**}-q^*= \sum_{x \in \X}v_3(x)(W_x-q^*)
\end{align}
by using a function $v_3$. 
Since $c(x)<0$ for $x \in \X \setminus \X_*$
and $c(x)=0$ for $x \in \X_*$, there exists an element 
$x_* \in \X_*$ such that $v_3(x_*)<0 $.
Hence,
\begin{align}
&q^{**}= q^*+\sum_{x \in \X}v_3(x)(W_x-q^*) 
=\sum_{x \in \X}v_3(x)W_x+
\Big(1-\sum_{x \in \X}v_3(x)\Big) q^*
\nonumber \\
=&\sum_{x \in \X}v_3(x)W_x+
\Big(1-\sum_{x \in \X}v_3(x)\Big) 
\sum_{x \in \X_*}v_2(x)W_x \nonumber \\
=&
\sum_{x \in \X\setminus \X_*}v_3(x)W_x+
\sum_{x \in \X_*}
\Big(v_3(x)+
\Big(1-\sum_{x \in \X}v_3(x)\Big) v_2(x)\Big)W_x,
\end{align}
which shows the desired statement.
\end{proof}
The following lemma holds for Algorithm \ref{protocol2Y}.

\begin{lemma}\Label{LCV2}
$\X_{*}$ is contained in one of sets 
$\{  \X\setminus \X_1  \}_{\X_1 \in {\cal A}_j\cup {\cal B}^j}
$ for any $j$.
Hence, when ${\cal A}_j $ is empty,
$\X_{*}$ is contained in one of sets 
$\{  \X\setminus \X_1  \}_{\X_1 \in {\cal B}^j}$, i.e.,
$\X_{*}$ equals $\X\setminus \argmax_{\X_1 \in {\cal B}^j}
\hat{C}(\X\setminus \X_1) $.
\end{lemma}

This lemma guarantees the correctness of Algorithm \ref{protocol2Y} for the calculation of the capacity.

\begin{proof}
We show the desired statement by induction for $j$.
For $j=1$, the desired statement holds as follows.
Due to Lemma \ref{LCVF}, 
$\X_{*}$ is contained in one of sets 
$\{  \X\setminus \X_1  \}_{\X_1 \in {\cal A}_1}$.

We assume that 
$\X_{*}$ is contained in one of sets 
$\{  \X\setminus \X_1  \}_{\X_1 \in {\cal A}_k\cup {\cal B}^k}
$.
If $\X_{*}$ is contained in one of sets 
$\{  \X\setminus \X_1  \}_{\X_1 \in {\cal B}^k}$,
the desired statement with $j=k+1$ holds.
If $\X_{*}$ is contained in one of sets $\{  \X\setminus \X_1  \}_{\X_1 \in {\cal A}_k}$,
we choose $\X_1 \in {\cal A}_k$ such that
$\X_{*}\subset \X\setminus \X_1 $.
Due to Lemma \ref{LCVF}, there exists an element $x \in {\cal N}(\X\setminus\X_1)$ 
such that 
$\X_{*}\subset \X \setminus (\X_1 \cup \{x\})$.
When ${\cal N}(\X \setminus (\X_1 \cup \{x\}))=\emptyset$,
$\X_{*}$ is one of subsets
$\{  \X\setminus \X_1  \}_{\X_1 \in {\cal B}^{k+1}}$.

When ${\cal N}(\X \setminus (\X_1 \cup \{x\}))\neq\emptyset$,
$\hat{C}(\X \setminus (\X_1 \cup \{x\})\ge
\hat{C}(\X_*)$ and $\hat{C}(\X_*)$ equals the capacity.
Since $\X_*$ is not contained in ${\cal B}^{k}$,
we have $\hat{C}(\X_*)> C^{k}$.
Hence, we have 
$\hat{C}(\X \setminus (\X_1 \cup \{x\})> C^{k}$.
Thus, $\X_1 \cup \{x\} \in {\cal A}_{k+1}$.
Therefore,
$\X_{*}$ is contained in one of sets 
$
\{  \X\setminus \X_1  \}_{\X_1 \in {\cal A}_{k+1}}
\subset 
\{  \X\setminus \X_1  \}_{\X_1 \in {\cal A}_{k+1}\cup{\cal B}_{k+1}}
$.
\end{proof}

\begin{algorithm}
\caption{Non-iterative algorithm for classical channel capacity in the general case}
\Label{protocol2Y}
\begin{algorithmic}
\STATE {
We apply Algorithm \ref{protocol1} to the input set $\X$.
If ${\cal N}(\X)$ is the empty set, we output $\hat{C}(\X)$ as the capacity.
Otherwise, we define the family ${\cal A}_1$ of subsets of $\X$ with cardinality 1 as
$\{   \{x\}  \}_{x \in {\cal N}(\X)}$.
Set $j=1$;} 
\REPEAT 
\STATE {We define 
families 
${\cal A}_{j+1}$,
${\cal B}_{j+1}$,
${\cal B}^{j+1}$
of subsets of $\X$ 
and the positive number $C^{j+1}$ by using Algorithm \ref{protocol1}
as follows.
\begin{align}
{\cal A}_{j+1}
&:=\left\{  \X_1 \cup \{x\} \left|
\begin{array}{l}
\X_1 \in {\cal A}_j, x \in {\cal N}(\X\setminus\X_1), 
{\cal N}(\X \setminus (\X_1 \cup \{x\})) \neq \emptyset ,\\
\hat{C}(\X \setminus (\X_1 \cup \{x\})) > C^j
\end{array}
\right \}\right. \\
{\cal B}_{j+1}&:=
 \{ \X_1 \cup \{x\} | 
\X_1 \in {\cal A}_j,x \in \X\setminus\X_1, 
{\cal N}(\X \setminus (\X_1 \cup \{x\})) = \emptyset 
\} \\
{\cal B}^{j+1}&:=
{\cal B}^{j}\cup
{\cal B}_{j+1} \\
C^{j+1}&:=\max_{
\X_1\in {\cal B}^{j+1}}\hat{C}(\X \setminus \X_1) .
\end{align}}
\UNTIL{${\cal A}_{j+1}$ is empty. When this stopping condition holds, 
we denote $j+1$ by $j_0$.} 
\end{algorithmic}
We output $C^{j_0}$ as the channel capacity.
\end{algorithm}

\begin{remark}\Label{Rem3}
Muroga \cite{Muroga} also 
considered the calculation method of the classical channel capacity.
In \cite[Section 1]{Muroga},
he derived an analytical calculation method when $n_1=n_2$.
In this special case, our method is slightly different from his method as follows.
While his method needs to calculate the inverse matrix of an $n_1\times n_1$ matrix,
our method needs only to calculate the inverse matrix of an 
$(n_1-1)\times (n_1-1)$ matrix.
Hence, our method is slightly better than his method.
When $n_2>n_1$, he presented his calculation method in \cite[Section 2]{Muroga}.
His calculation method requires to solve nonlinear characteristic equations 
\cite[(28)]{Muroga}.
Although he did not explain how to solve the characteristic equations,
the solution can be characterized by the minimizer of a certain convex function
of $n_2-n_1$ variables in a similar way to \eqref{XOEBC}.
Also,
his calculation method requires to calculate
the determinants of $n_1(n_2-n_1)+1$ $n_1\times n_1 $-matrices
while our method needs to calculate $f_j(y)$, which can be calculated by 
the inverse matrix of $(n_1-1)\times (n_1-1)$-matrix. 
The calculation of the inverse matrix of size $n_1-1$
is easier than
the determinants of $(n_1-1)^2$ $(n_1-2)\times (n_1-2)$-matrices 
and one $(n_1-1)\times (n_1-1)$-matrices 
due to Cramer's formula of the inverse matrix.
Hence, our method is slightly easier than his method.
\end{remark}

\subsection{Application of non-iterative method}
This section aims to demonstrate the advantage of 
our method over the method in \cite{exact}.
That is, applying Algorithm \ref{protocol2Y},
we make a numerical calculation of the classical channel capacity
with the following channel of $n_1=n_2=4$;
\begin{align}
W_1:=
\left(
\begin{array}{c}
0.05 \\
0.9-t \\
0.05 \\
t 
\end{array}
\right),~
W_2:=
\left(
\begin{array}{c}
0.05 \\
0.05 \\
0.9-t \\
t 
\end{array}
\right),~
W_3:=
\left(
\begin{array}{c}
0.9 \\
0.05 \\
0.05 \\
0 
\end{array}
\right),~
W_4:=
\left(
\begin{array}{c}
0.05 \\
0.05 \\
0.05 \\
0.85 
\end{array}
\right).\Label{chan1}
\end{align}

In this channel \eqref{chan1}, 
according to Algorithm \ref{protocol2Y},
we apply Algorithm \ref{protocol1} to the input set $\{1,2,3,4\}$.
As a result, we found that 
the optimal input distribution has the support $\{1,2,3,4\}$ 
when $0\le t \le 0.18$.
However, when $t \ge 0.18$, it does not have a positive probability at $X=4$, i.e., 
${\cal N}(\X)$ is not the empty set.
This case cannot be covered by 
Algorithm 1 in the reference \cite{exact}.
Hence, in this case, as the next step, we apply Algorithm \ref{protocol1} with 
${\cal X}=\{1,2,3\}$,
where we need to make the minimization \eqref{XOEBC} with one free parameter.
Its numerical calculation is done as Figs. \ref{fig-C} and \ref{fig-P}.

\begin{figure}[htbp]
\begin{center}
  \includegraphics[width=0.7\linewidth]{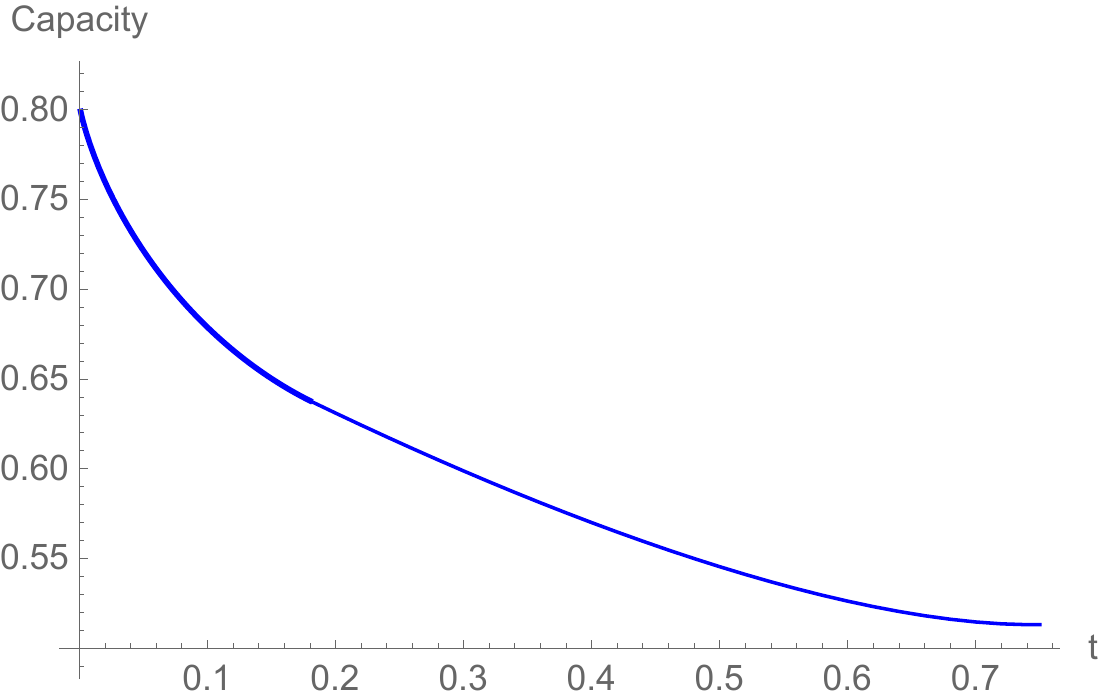}
  \end{center}
\caption{Capacity of the channel \eqref{chan1}.
For $0\le t \le 0.18$, the capacity is calculated by Algorithm \ref{protocol1} with 
${\cal X}=\{1,2,3,4\}$.
For $0.76 \ge t \ge 0.18$, 
the capacity is calculated by Algorithm \ref{protocol1} with 
${\cal X}=\{1,2,3\}$.}
\Label{fig-C}
\end{figure}   

\begin{figure}[htbp]
\begin{center}
  \includegraphics[width=0.7\linewidth]{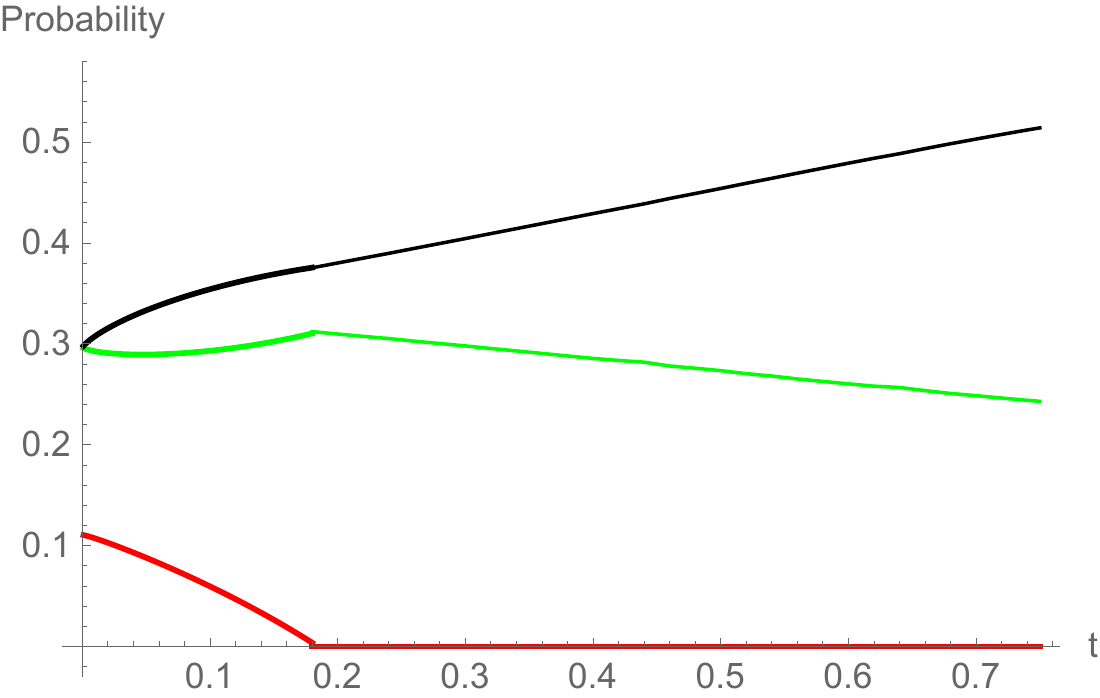}
  \end{center}
\caption{Input distribution realizing capacity:
Due to the symmetry, $P_X(1)=P_X(2)$ in the optimal input distribution.
Green curve shows $P_X(1)$ of the optimal input distribution.
Black curve shows $P_X(3)$ of the optimal input distribution.
Red curve shows $P_X(4)$ of the optimal input distribution.
This value is zero for $t \ge 0.18$.}
\Label{fig-P}
\end{figure}   

\section{Classical Secrecy Capacity}\Label{S11}
\subsection{Problem setting}
Let 
$\X:=\{1,\ldots,n_1 \}$, ${\cal Z}:=\{1,\ldots,n_2 \}$, 
and $\Y:= \{1,\ldots, n_3 \}$ 
be finite sets. 
We call a map  $ W  :\X  \rightarrow  {\cal P}_{\Z \times \Y}$ be a channel from $\X$ to 
$\Z \times \Y$. 
In this section,
we use the subscript $\saa,\sbb,\scc,\sdd,\see,
\sff,\sgg,\shh$ to express elements
of 
$\mathbb{R}^{n_1-1}$, 
$\mathbb{R}^{n_2(n_1+n_3)-n_1}$,
$\mathbb{R}^{n_2(n_1+n_3)-1}$, 
$\mathbb{R}^{n_1n_2n_3-n_1}$, 
$\mathbb{R}^{n_2(n_1+n_3)-2n_1+1}$, 
$\mathbb{R}^{n_2 n_1-n_1}$, 
$\mathbb{R}^{n_2(n_3-1)}$, 
$\mathbb{R}^{n_2(n_3-1)-n_1+1}$, 
respectively, as Tables \ref{tableC} and \ref{tableD}.

\begin{table*}[htb]
\begin{center}
\caption{Summary of subscripts for Section \ref{S11} (1)}
\begin{tabular}{|c|c|c|c|c|} \hline
Subscript &  $\saa$&$\sbb$&$\scc$&$\sdd$
\\
\hline 
Vector space & 
$\mathbb{R}^{n_1-1}$ &
$\mathbb{R}^{n_2(n_1+n_3)-n_1}$ &
$\mathbb{R}^{n_2(n_1+n_3)-1}$ &
$\mathbb{R}^{n_1n_2n_3-n_1}$ \\
\hline 
\multirow{2}{*}{Examples} &
$\Theta_{{\cal M}}, \Theta_{{\cal M}}$ &
\multirow{2}{*}{$\Theta_{{\cal E},\sbb}, \Theta_{{\cal E},\sbb}$}
& \multirow{2}{*}{$\Theta_{{\cal E}}, \Theta_{{\cal E}}$} &\\
&
$\Theta_{{\cal E},\saa}, \Theta_{{\cal E},\saa}$
& &&\\
\hline 
\end{tabular}
\par
\vspace{1ex}
\begin{flushleft}
In Section \ref{S11}, there are more types of vector spaces than in
Section \ref{Sec:BBem}. 
This table and the next table show the relation between a vector space appearing in 
Section \ref{S11}
and the subscript.
\end{flushleft}
\Label{tableC}
\end{center}

\begin{center}
\caption{Summary of subscripts for Section \ref{S11} (2)}
\begin{tabular}{|c|c|c|c|c|} \hline
Subscript &$\see$
&$\sff$&$\sgg$&$\shh$ \\
\hline 
Vector space & 
$\mathbb{R}^{n_2(n_1+n_3)-2n_1+1}$ & 
$\mathbb{R}^{n_2 n_1-n_1}$ &
$\mathbb{R}^{n_2(n_3-1)}$ &
$\mathbb{R}^{n_2(n_3-1)-n_1+1}$\\
\hline 
\end{tabular}
\Label{tableD}
\end{center}
\end{table*}

For a conditional distribution $P_{Y|Z}$
and a joint distribution $P_{X,Z}$,
we define the joint distribution
$ P_{Y|Z}\times P_{X,Z}$ on $\X\times \Z\times \Y$
as
$ P_{Y|Z}\times P_{X,Z}(x,z,y):= P_{Y|Z}(y|z) P_{X,Z}(x,z)$.
We use the notations $W_x(z,y):=W(z,y|x)$,
$W_{Z}(z|x):=\sum_{y'}W(z,y'|x)$, and 
$W_{x,z}(y):=\frac{W(z,y|x)}{W_{Z}(z|x)}$.
For $q \in {\cal P}_{\X}$ and $r \in {\cal P}_{\Y}$,
$ (W\times q)_{Y|Z}$ 
is defined by
$ (W\times q)_{Y|Z}(y|z):=\frac{(W\cdot q)(z,y)}{ (W_Z\cdot q)(z)}$.
Hence, 
$( (W\times q)_{Y|Z}\times (W_Z\times q))(x,z,y)=
(W\times q)_{Y|Z}(y|z) (W_Z\times q)(x,z)$.

When the channel $W$ satisfies Markov chain $X-Y-Z$,
the secrecy capacity of the wire-tap channel  $W$ is given by \cite{Wyner,CK79}
\begin{align}
&\max_{q \in {{\cal P}_{\X}} }
D(W _Y\times q \|  (W _Y\cdot q) \times q ) -D(W _Z\times q \|  (W _Z\cdot q) \times q ) \nonumber \\
=& \max_{q \in {{\cal P}_{\X}} }D(W \times q \| (W\times q)_{Y|Z}\times (W_Z\times q) ) \nonumber \\
=& \max_{q \in {{\cal P}_{\X}} }  \min_{Q \in {{\cal P}_{X-Z-Y} (W_{Z|X}) }} 
D(W \times q \| Q).\Label{MOA2}
\end{align}
We define the set of distributions ${{\cal P}_{X-Z-Y}  }$
on ${\cal X}\times {\cal Z}\times {\cal Y}$
to satisfy the Markov chain $X-Z-Y$.

As proven in Subsection \ref{S64},
the set ${{\cal P}_{X-Z-Y}  }$ forms an exponential subfamily ${\cal E}$ and 
the set of $W \times q  $ forms a mixture subfamily ${\cal M}$.
Hence, the maximization problem \eqref{MOA2} is a special case of the maximization \eqref{eq:Inf.st} with 
$k=n_1-1$, $l=n_1 n_2-1+ n_2(n_3-1)$, and $d=n_1 n_2 n_3-1$.
In the following, we apply Algorithm \ref{protocol2}.
For this aim, 
we need to choose a suitable coordinate to satisfy conditions (B1), (B3), (B4),
and check Condition (B2).

\subsection{Constructions of vectors $u_1, \ldots, u_{n_1 n_2 n_3-1},v_1,\ldots, v_{n_2(n_1+n_3-1)-1}$}
To choose a suitable coordinate to satisfy conditions (B3), (B4), and (B5), we need to choose 
suitable vectors $u_1, \ldots, u_{n_1n_2},v_1,\ldots, v_{n_1+n_2-1}$, which form
the matrices $U$ and $V$.
For this aim, we define various functions on ${\cal Y}$ and ${\cal X}\times {\cal Y}$.
Given $z =1, \ldots, n_2$,
we choose $n_3-1$ independent functions $f_{j,z}$ on ${\cal Y}$ with $j=1, \ldots, n_3-1$ 
to satisfy the condition that
\begin{align}
\sum_{y \in {\cal Y}} f_{j,z}(y) W_{n_1,z}(y)=0
\end{align}
and the linear space spanned by $f_{1,z}, \ldots, f_{n_2-1,z}$ does not contain a constant function.
As a typical case, $f_{j,z}$ can be chosen as follows.
\begin{align}
f_{j,z}(y):=
\left\{
\begin{array}{ll}
W_{n_1,z}(j+1) & \hbox{when } y=j \\
-W_{n_1,z}(j) & \hbox{when } y=j+1 \\
0 & \hbox{otherwise}.
\end{array}
\right.
\end{align}

Then, we define the functions $\xi_1, \ldots, \xi_{n_1n_2 n_3-1}$ on ${\cal X}\times {\cal Z}\times {\cal Y}$ as follows;
\begin{align*}
\xi_i(x,z,y):=& \delta_i(x)  \\ 
\xi_{n_1-1+(i-1)(n_2-1)+j' }(x,z,y):=& (\delta_{j'}(z)-W_{Z,i}(j'))\delta_i(x)\\
\xi_{n_1 n_2-1+(n_1-1) n_2 (n_3-1)+(z'-1)(n_3-1)+j }(x,z,y):= &f_{j,z'}(y) \delta_{n_1}(x)\delta_{z'}(z) ,
\end{align*}
and
\begin{align}
&\xi_{n_1 n_2-1+(i-1)n_2(n_3-1)+(z'-1)(n_3-1)+j }(x,z,y)\nonumber \\
:=& (f_{j,z'}(y)\!-\!h_{i,z'(n_3-1)+j })\delta_{i}(x)\delta_{z'}(z) \Label{XXK}
\end{align}
for $i=1, \ldots, n_1-1$, $z'=1, \ldots, n_2$, 
$j'=1, \ldots, n_2-1 $, and
$j=1, \ldots, n_3-1$,
where 
\begin{align}
h_{i,z'(n_3-1)+j }:=\sum_{y} f_{j,z'}(y)W_{i}(z',y).\Label{MZA}
\end{align}

Then, we define the $C^{\infty}-$strictly convex function 
$F$ on $\mathbb{R}^{n_1n_2n_3-1}$ as
\begin{align}
F(\theta):= \log \sum_{x,z,y} e^{\sum_{i=1}^{n_1 n_2 n_3-1} \theta^i \xi_i(x,z,y) }.
\end{align}
That is,
we consider the Bregman divergence system $(\mathbb{R}^{d}, F, D^F)$.
We define the distribution $P_\theta, P_{\theta,XZ},P_{\theta,X},P_{\theta,Z}, 
P_{\theta,Y|Z}$ as
\begin{align}
P_\theta(x,z,y) &:=e^{\sum_{i=1}^{n_1 n_2 n_3-1} \theta^i \xi_i(x,z,y)-F(\theta) }, \\
P_{\theta,XZ}(xz) &:=\sum_{y} e^{\sum_{i=1}^{n_1 n_2n_3-1} \theta^i \xi_i(x,z,y)-F(\theta) },\Label{CX1}\\
P_{\theta,X}(x) &:=\sum_{z,y}e^{\sum_{i=1}^{n_1 n_2n_3-1} \theta^i \xi_i(x,z,y)-F(\theta) },\Label{CX2}\\
P_{\theta,Z}(z) &:=\sum_{x,y}e^{\sum_{i=1}^{n_1 n_2n_3-1} \theta^i \xi_i(x,z,y)-F(\theta) },\\
P_{\theta,Y|Z}(y|z) &:=
\frac{\sum_{x}e^{\sum_{i=1}^{n_1 n_2n_3-1} \theta^i \xi_i(x,z,y)-F(\theta) }}{P_{\theta,Z}(z)} .
\end{align}
Then, as a special case of \eqref{MGA}, we have
\begin{align}
D^F(\theta\|\theta')= D(P_\theta\|P_{\theta'}).
\end{align}

Next, we choose the matrix $U$ as the identity matrix,
and $u_1, \ldots, u_{n_1 n_2 n_3-1}$ are chosen as its $n_1 n_2 n_3-1$ column vectors.
Then, we define vector $v_1, \ldots, v_{n_2(n_1+n_3-1)-1}$ as follows, whereas $V=(v_1, \ldots, v_{n_2(n_1+n_3-1)-1})$.
\begin{align}
v_i &:= u_i \\
v_{n_1 n_2 -1+j} &:= 
\sum_{i=1}^{n_1} u_{n_1 n_2-1+(n_1-1) n_2 (n_3-1)+j }+ \sum_{i=1}^{n_1-1} h_{i,j} u_i \Label{XLN}
\end{align}
for $ i=1, \ldots,n_1 n_2 -1$ and $ j=1, \ldots, n_2(n_3-1)$.
We define $g_j(x,z,y)$ for $j=1,\ldots, n_2(n_1+n_3-1)-1 $ as
\begin{align}
g_i(x,z,y) := &\delta_i(x) \\
g_{n_1-1+(i-1)(n_2-1)+j' }(x,z,y):=& (\delta_{j'}(z)-W_{Z,i}(j'))\delta_i(x)\Label{Eq144}\\
g_{n_1 n_2-1+(z'-1)(n_3-1)+j }(x,z,y):= &f_{j,z'}(y) \delta_{z'}(z) ,
\end{align}
for $i=1, \ldots, n_1-1$, $z'=1, \ldots, n_2$, 
$j'=1, \ldots, n_2-1 $, and
$j=1, \ldots, n_3-1$.
Then, we have
\begin{align}
\sum_{i=1}^{n_2(n_1+n_3-1)-1}\xi_i (x,z,y) v_j^i=g_j(x,z,y).\Label{COA}
\end{align}

\subsection{Parameterizations of ${\cal E}$ and ${\cal M}$}
We define the exponential subfamily ${\cal E}$
by the generator $v_1, \ldots,  v_{n_2(n_1+n_3-1)-1}$ at the point 
$0$.
Since the set $\{v_1, \ldots,  v_{n_2(n_1+n_3-1)-1}\}$ spans
 the function space spanned by functions of $\X$ and $\Z$
 and functions of $\Y$ and $\Z$,
 the exponential subfamily ${\cal E}$ is the inner of 
${\cal P}_{X-Z-Y}$.

We define the mixture family ${\cal M}$ by the constraint
$\sum_{i=1}^{n_1 n_2-1} u^i_{n_1-1+j'} \partial_i F(\theta)=0$ 
for $j'=1, \ldots, n_1 (n_2-1)$.
This constraint is equivalent to 
\begin{align}
\sum_{z,y } (\delta_j(z)-W_{Z,i}(j))P_{\theta}(i,z,y) &=0 \\
\sum_{z,y }  (f_{j,z'}(y)-h_{i, z'(n_3-1) +j }) \delta_{z'}(z) P_{\theta}(i,z,y) &=0 \\
\sum_{z,y }  f_{j,z}(y) \delta_{z'}(z) P_{\theta}(n_1,z,y)&=0 
\end{align}
for $i=1, \ldots, n_1-1$, $z'=1, \ldots, n_2$, $j=1, \ldots, n_3-1 $.
For $\theta_1 \in \mathbb{R}^{n_1-1}$ and 
$\theta_{\sbb} \in \mathbb{R}^{n_2(n_1+n_3-1)-n_1}$, the function 
$F_{{\cal E}}$ is given as
\begin{align}
F_{{\cal E}}(\theta_\saa,\theta_{\sbb})
=F( V (\theta_\saa,\theta_{\sbb})^T).
\Label{XNII}
\end{align}

The mixture family ${\cal M}$ is composed of distributions with the form $W \times q$.
Thus, the problem \eqref{MOA2} is written as the problem \eqref{eq:Inf.st} with the above defined ${\cal E}$ and ${\cal M}$.
Since the conditional probability 
$P_{\theta,ZY|X}(z,y|i)=
\frac{P_{\theta,XZY}(i,z,y)}{\sum_{z',y'}P_{\theta,XZY}(i,z',y')}$ depends only on 
$(\theta^{j})_{j=n_1}^{n_2(n_1+n_3)-1}$ for $i=1, \ldots, n_1$,
we choose $\theta_{\sdd}^{\dagger}=
(\theta^{n_1,\dagger}, \ldots, \theta^{n_1 n_2 n_3-1,\dagger})$ as
\begin{align}
W_i(z,y)=
\frac{e^{\sum_{j=n_1}^{n_1 n_2 n_3-1}  \theta^{j,\dagger } 
\xi_j(i,z,y)
}}
{\sum_{z',y'} 
e^{\sum_{j=n_1}^{n_1 n_2  n_3-1}  \theta^{j,\dagger } 
\xi_j(i,z',y')}}.\Label{HIA}
\end{align}
Since \eqref{XXK} and \eqref{MZA} guarantee the relation
\begin{align}
\sum_{z',y'} 
\xi_{j'}(i,z',y') e^{\sum_{j=n_1}^{n_1 n_2  n_3-1}  \theta^{j,\dagger } \xi_j(i,z',y')}
=0
\end{align}
for $j'=n_1, \ldots,n_1 n_2  n_3-1$, we have
\begin{align}
&H(W_i)\nonumber \\
=&
-\sum_{z,y}
\bigg(\sum_{j=n_1}^{n_1 n_2 n_3-1}  \theta^{j,\dagger } 
\xi_j(i,z,y)
-\log \Big(\sum_{z',y'} 
e^{\sum_{j=n_1}^{n_1 n_2 n_3-1}  \theta^{j,\dagger } \xi_j(i,z',y')
}\Big)\bigg)W_i(z,y)\nonumber \\
=&
\log \bigg(\sum_{z',y'} 
e^{\sum_{j=n_1}^{n_1 n_2 n_3-1}  \theta^{j,\dagger } \xi_j(i,z',y')
} \bigg)
\Label{XMT}
\end{align}

Due to \eqref{HIA}, ${\cal M}$ is written as
$\{( \theta_\saa, \theta_{\sdd}^{\dagger})| \theta_1 \in \mathbb{R}^{n_1-1}\}$
because the matrix $U$ is the identity matrix.
That is,
${\cal M}$ forms an exponential subfamily generated by 
$u_1, \ldots, u_{n_1-1}$ at $(0,\theta_{\sdd}^{\dagger}) $.
For $\theta_1 \in \mathbb{R}^{n_1-1}$, the function 
$F_{{\cal M}}$ is given as
\begin{align}
F_{{\cal M}}(\theta_\saa)=F( \theta_\saa, \theta_{\sdd}^{\dagger}).
\Label{XNIY}
\end{align}
In addition, the maximization \eqref{MOA2} is rewritten as 
\begin{align}
\max_{q \in {{\cal P}_{\X}} }  \min_{Q \in {{\cal P}_{X-Z-Y} (W_{Z|X}) }} 
D(W \times q \| Q)
=&
\max_{\theta \in \mathcal{M}} D^{F}(\theta \| \Pro^{(m),F}_{\mathcal{E}} (\theta))
\nonumber \\
=&
\max_{\theta \in \mathcal{M}} 
\min_{\theta' \in \mathcal{E}} 
D^{F}(\theta \| \theta').
\end{align}

\subsection{Check of Conditions (B2), (B3), and (B4)}\Label{S64}
Lemma \ref{LOS} guarantees Condition (B1).
We define the $(n_1-1) \times  n_2(n_3-1)$ matrix $H:=(h_{i,j})$,
where $h_{i,j}$ is defined in \eqref{MZA}.
Then, \eqref{XLN} guarantees that the 
$ (n_1-1)\times (n_2 (n_1+n_3-1)-1)$ matrix $V_1$ is $ (I,
0_{n_1-1,n_1 (n_2-1)},H)$,
where $0_{n_1-1,n_1 (n_2-1)}$ is the $(n_1-1)\times n_1 (n_2-1)$ zero matrix.
That is, $(n_1-1) \times (n_1 (n_2-1)+ n_2(n_3-1))$ matrix $V_3$ is 
$(0_{n_1-1,n_1 (n_2-1)},H)$.
Hence, conditions (B3) and (B4) hold.
Therefore, we can apply 
Algorithm \ref{protocol2}.
However, in this example, the condition (B5) does not hold, in general.

As the relation
\begin{align}
P_{\Pro^{(m),F}_{{\cal E}}  (\theta)}=P_{\theta,Y|Z}\times P_{\theta,X Z}
\end{align}
holds for any $\theta$, we have
\begin{align}
 D^{F}(\theta'\|\theta)
=&
 D(P_\theta\| P_{\theta'})
=
 D(P_{\theta,XZ}\| P_{\theta',XZ})
\nonumber \\
 \le &
 D( P_{\theta,XZ}\| P_{\theta',XZ})
+\sum_{z} P_{\theta,Z}(z)D( P_{\theta,Y|Z=z}\| P_{\theta',Y|Z=z}) \nonumber \\
=&
 D( P_{\theta,Y|Z}\times P_{\theta,XZ}\|
 P_{\theta',Y|Z}\times P_{\theta',XZ})
= D(P_{ \Pro^{(m),F}_{{\cal E}}  (\theta')}\| 
(P_{ \Pro^{(m),F}_{{\cal E}}  (\theta)} ) 
\nonumber \\
=&
 D^{F}( \Pro^{(m),F}_{{\cal E}}  (\theta')\|  \Pro^{(m),F}_{{\cal E}}  (\theta) ) 
\end{align}
for $\theta,\theta' \in {\cal M}$.
Thus, condition (B2) holds.
Therefore,
Theorem \ref{theo:conv:BBem} guarantees the global convergence.
When $\theta^{(1)}$ is the uniform distribution on ${\cal X}$,
in the same way as \eqref{CKP}, we can show that
the supremum $\sup_{\theta \in \mathcal{M}} D^F(\theta \| \theta^{(1)})$
equals $ \log n_1$.
Therefore, 
when Theorem \ref{theo:conv:BBem} is applied,
we obtain the precision \eqref{NHG}
with $ \frac{\log n_1}{\epsilon}$ iterations.

\subsection{Conversion to em-problem} 
We define the following functions;
\begin{align}
\bar{g}_i(x)&:=g_i(x,z,y) \\
\bar{g}_j(x,z)&:=g_j(x,z,y) \Label{Eq159}\\
\bar{g}_{j'}(z,y)&:=g_{j'}(x,z,y) 
\end{align}
for $i=1, \ldots, n_1-1$, $j=n_1, \ldots, n_1n_2-1$, and 
$j'= n_1n_2,\ldots, n_2(n_1+n_3-1)-1$.

For 
$\theta_\saa=(\theta^1,\ldots,\theta^{n_1-1})\in \mathbb{R}^{n_1-1} $
and
$\theta_\sff=(\theta^{n_1},\ldots,\theta^{n_1n_2-1})\in \mathbb{R}^{n_1(n_2-1)} $,
we define 
$\hat{P}_{\theta_\saa, X}$,
$\hat{P}_{\theta_\saa,\theta_\sff,XZ}$, and
$\hat{P}_{\theta_\saa,\theta_\sff,Z|X}$ as
\begin{align}
\hat{P}_{\theta_\saa,X}(x)&:=
\frac{e^{
\sum_{j=1}^{n_1-1} \theta^j \bar{g}_j(x)
}}
{\sum_{x'}
e^{\sum_{j=1}^{n_1-1} \theta^j \bar{g}_j(x')}} \Label{XC1T}\\
\hat{P}_{\theta_\saa,\theta_\sff,XZ}(x,z)&:=
\frac{e^{
\sum_{j=1}^{n_1-1} \theta^j \bar{g}_j(x)
+\sum_{j=n_1}^{n_1n_2-1} \theta^j \bar{g}_j(x,z)
}}
{\sum_{x',z'}
e^{\sum_{j=1}^{n_1-1} \theta^j \bar{g}_j(x')
+\sum_{j=n_1}^{n_1n_2-1} \theta^j \bar{g}_j(x',z')
}} \Label{XC2T}\\
\hat{P}_{\theta_\sff,Z|X}(z|x)&:=
\frac{e^{
\sum_{j=n_1}^{n_1n_2-1} \theta^j \bar{g}_j(x,z)
}}
{\sum_{z'}
e^{\sum_{j=n_1}^{n_1n_2-1} \theta^j \bar{g}_j(x,z')
}} .\Label{XC3T}
\end{align}
Then, we have
\begin{align}
\hat{P}_{\theta_\saa,\theta_\sff,Z|X}(z|x)=
\frac{\hat{P}_{\theta_\saa,\theta_\sff,XZ}(x,z)}
{\sum_{z'}\hat{P}_{\theta_\saa,\theta_\sff,XZ}(x,z')}.
\end{align}
For $\theta_\sff=(\theta^{n_1},\ldots,\theta^{n_1n_2-1})\in \mathbb{R}^{n_1(n_2-1)} $,
we choose $\psi_\saa(\theta_\sff) =(\psi^1,\ldots,\psi^{n_1-1})\in \mathbb{R}^{n_1-1} $ and $C_X(\theta_\sff)$ as
\begin{align}
\sum_{j=1}^{n_1-1} \psi^j \bar{g}_j(x)
=-\log \sum_{z'} e^{\sum_{j=n_1}^{n_1n_2-1} \theta^j \bar{g}_j(x,z')}+C_X(\theta_\sff).
\Label{XZP}
\end{align}
Also, we define $\theta_\saa^\dagger
=(\theta^{1,\dagger}, \ldots, \theta^{n_1-1,\dagger}) \in \mathbb{R}^{n_1-1}$ and $C_{X,\dagger}$
as
\begin{align}
\sum_{j=1}^{n_1-1} \theta^{j,\dagger} \bar{g}_j(x)
=-H(W_x) +C_{X,\dagger}.
\Label{XZP2}
\end{align}
The relation \eqref{XZP2} gives the unique definition of $\theta_1^\dagger$ because 
the functions $\{\bar{g}_j\}_j$ and the constant form a basis of the function space over ${\cal X}$.

\begin{lemma}\Label{XCL}
We have the following relations
\begin{align}
\hat{P}_{\theta_1,X}(x)=&
\sum_{z}\hat{P}_{\theta_\saa+ \psi_1(\theta_\sff),\theta_\sff,XZ}(x,z) \Label{XM1}\\
\hat{P}_{\theta_\saa,X}(x)=&
P_{\theta_\saa+\theta_{\saa}^{\dagger},
\theta_{\sdd}^{\dagger},X}(x).
\Label{XM2} 
\end{align}
\end{lemma}

\begin{proof}
The relation \eqref{XM1} is shown as follows;
\begin{align}
&\sum_{z}\hat{P}_{\theta_\saa+ \psi_\saa(\theta_\sff),
\theta_\sff,XZ}(x,z) \nonumber \\
=&
\sum_{z}\frac{e^{
\sum_{j=1}^{n_1-1} (\theta^j +\psi^j)\bar{g}_j(x)
+\sum_{j=n_1}^{n_1n_2-1} \theta^j \bar{g}_j(x,z)
}}
{\sum_{x',z'}
e^{\sum_{j=1}^{n_1-1} (\theta^j +\psi^j) \bar{g}_j(x')
+\sum_{j=n_1}^{n_1n_2-1} \theta^j \bar{g}_j(x',z')
}}\nonumber \\
\overset{(a)}=&
\frac{e^{
\sum_{j=1}^{n_1-1} \theta^j \bar{g}_j(x)
+C_X(\theta_\sff)}}
{\sum_{x'}
e^{
\sum_{j=1}^{n_1-1} \theta^j \bar{g}_j(x')
+C_X(\theta_\sff)}
}\nonumber \\
=&
\frac{e^{
\sum_{j=1}^{n_1-1} \theta^j \bar{g}_j(x)}}
{\sum_{x'}
e^{
\sum_{j=1}^{n_1-1} \theta^j \bar{g}_j(x')}
}=\hat{P}_{\theta_\saa,X}(x),
\end{align}
where $(a)$ follows from \eqref{XZP}. 

The relation \eqref{XM2} is shown as follows;
\begin{align}
P_{\theta_\saa+\theta_{\saa}^{\dagger},\theta_{\sdd}^{\dagger},X}(x)
=&
\frac{e^{
\sum_{j=1}^{n_1-1} (\theta^j +\theta^{j,\dagger})\bar{g}_j(x)}
\sum_{z,y}
e^{\sum_{j=n_1}^{n_1 n_2 n_3-1}  \theta^{j,\dagger } \xi_j(x,z,y)}}
{\sum_{x'} e^{
\sum_{j=1}^{n_1-1} (\theta^j +\theta^{j,\dagger})\bar{g}_j(x')}
\sum_{z,y}
e^{\sum_{j=n_1}^{n_1 n_2 n_3-1}  \theta^{j,\dagger } \xi_j(x',z,y)}}
\nonumber \\
\overset{(a)}=&
\frac{e^{
\sum_{j=1}^{n_1-1} (\theta^j +\theta^{j,\dagger})\bar{g}_j(x)+H(W_x)}
}
{\sum_{x'} e^{
\sum_{j=1}^{n_1-1} (\theta^j +\theta^{j,\dagger})\bar{g}_j(x')+H(W_{x'})}
}
\nonumber \\
\overset{(b)}=&
\frac{e^{
\sum_{j=1}^{n_1-1} \theta^j \bar{g}_j(x)+C_{X,\dagger}}
}
{\sum_{x'} e^{
\sum_{j=1}^{n_1-1} \theta^j \bar{g}_j(x')+C_{X,\dagger}}
}
=
\frac{e^{
\sum_{j=1}^{n_1-1} \theta^j \bar{g}_j(x)}
}
{\sum_{x'} e^{
\sum_{j=1}^{n_1-1} \theta^j \bar{g}_j(x')}
}
=
\hat{P}_{\theta_\saa,X}(x),
\end{align}
where $(a)$ and $(b)$ follows from \eqref{XMT} and \eqref{XZP2}, respectively.
\end{proof}

To check Condition (B6),
we define 
$\theta_\saa^{\ddagger}=(\theta^{1,\ddagger},\ldots,
\theta^{n_1-1,\ddagger})$, $C_{X,\ddagger}$,
and
$\theta_\sff^{\ddagger}=(\theta^{n_1,\ddagger},\ldots,
\theta^{n_2 n_1-1,\ddagger})$ as
\begin{align}
\sum_{j=1}^{n_1-1} \theta^{j,\ddagger} \bar{g}_j(x)
&=-H(W_{Z|x}) +C_{X,\ddagger}
\Label{XZP3} \\
W_{Z|i}(z)&=\frac{ 
e^{\sum_{j=1}^{n_2-1} \theta^{n_1-1+(i-1)(n_2-1)+j,\ddagger } 
\bar{g}_{(i-1)(n_2-1)+j}(i,z) }
}{
\sum_{z'}e^{\sum_{j=1}^{n_2-1} \theta^{n_1-1+(i-1)(n_2-1)+j ,\ddagger} 
\bar{g}_{(i-1)(n_2-1)+j}(i,z') }
}
\Label{ACG}
\end{align}
for $i=1, \ldots, n_1$.
Since the function $\bar{g}_j$ is defined by 
\eqref{Eq144} and \eqref{Eq159},
\eqref{ACG} is rewritten as
\begin{align}
W_{Z|x}(z)=\frac{e^{
\sum_{j=n_1}^{n_1n_2-1} 
\theta^{j,\ddagger}\bar{g}_j(x,z)}
}
{\sum_{z'}e^{
\sum_{j=n_1}^{n_1n_2-1} 
\theta^{j,\ddagger}\bar{g}_j(x,z')}
}.
\Label{ACGL}
\end{align}

\begin{lemma}
We have the following relations.
\begin{align}
\hat{P}_{\theta_\sff^{\ddagger},Z|X}(z|x)
=&W_{Z|x}(z)
\Label{XMPG} \\
\theta_\saa^{\ddagger}=& \psi_\saa(\theta_{\sff}^{\ddagger}) \Label{CMR}\\
\hat{P}_{(\theta_\saa+\theta_\saa^{\ddagger},\theta_\sff^{\ddagger}), XZ}
=&P_{(\theta_\saa+\theta_\saa^{\dagger},
\theta_{\sdd}^{\dagger}), XZ}.
\Label{XMP4B}
\end{align}
\end{lemma}

\begin{proof}
The relation \eqref{XMPG} follows from \eqref{ACGL} and \eqref{XC3T}.
The combination of \eqref{XZP}, \eqref{XZP3}, and \eqref{ACGL} yields \eqref{CMR}.

The relation \eqref{CMR} is shown as follows;
Similar to \eqref{XMT}, since \eqref{Eq144} and \eqref{Eq159} guarantee
$\sum_{z}\bar{g}_j(x,z)W_{Z|x}(z)=0$, \eqref{ACGL} implies
\begin{align}
H(W_{Z|x})=-\log \sum_{z'} e^{\sum_{j=n_1}^{n_1n_2-1} \theta^j \bar{g}_j(x,z')}.\Label{AMR}
\end{align}
The combination of \eqref{XZP}, \eqref{XZP3}, and \eqref{AMR} yields 
\begin{align}
\sum_{j=1}^{n_1-1} \psi^j \bar{g}_j(x)-C_X(\theta_\sff)
=
\sum_{j=1}^{n_1-1} \theta^{j,\ddagger} \bar{g}_j(x)-C_{X,\ddagger}.
\end{align}
Since the functions $\{\bar{g}_j\}_j$ and the constant
are linearly independent, we obtain
\eqref{CMR}.

The relation \eqref{XMP4B} is shown as follows;
\begin{align}
&P_{\theta_\saa+\theta_{\saa}^{\dagger},
\theta_{\sdd}^{\dagger},XZ}(x)
\nonumber \\
=&
\frac{e^{
\sum_{j=1}^{n_1-1} (\theta^j +\theta^{j,\dagger})\bar{g}_j(x)}
\sum_{y}e^{\sum_{j=n_1}^{n_1 n_2 n_3-1}  \theta^{j,\dagger } \xi_j(x,z,y)}}
{\sum_{x',y',z'} e^{
\sum_{j=1}^{n_1-1} (\theta^j +\theta^{j,\dagger})\bar{g}_j(x')}
e^{\sum_{j=n_1}^{n_1 n_2 n_3-1}  \theta^{j,\dagger } \xi_j(x',z',y')}}
\nonumber \\
=&
\frac{
e^{
\sum_{j=1}^{n_1-1} (\theta^j +\theta^{j,\dagger})\bar{g}_j(x)}
\sum_{z',y'}e^{\sum_{j=n_1}^{n_1 n_2 n_3-1}  \theta^{j,\dagger } \xi_j(x,z',y')}
}
{\sum_{x',y',z'} e^{
\sum_{j=1}^{n_1-1} (\theta^j +\theta^{j,\dagger})\bar{g}_j(x')}
e^{\sum_{j=n_1}^{n_1 n_2 n_3-1}  \theta^{j,\dagger } \xi_j(x',z',y')}}
\nonumber \\
&\cdot \frac{\sum_{y}e^{\sum_{j=n_1}^{n_1 n_2 n_3-1}  \theta^{j,\dagger } \xi_j(x,z,y)}}
{\sum_{z',y'}e^{\sum_{j=n_1}^{n_1 n_2 n_3-1}  \theta^{j,\dagger } \xi_j(x,z',y')}}
\nonumber \\
\overset{(a)}=&
\hat{P}_{\theta_\saa,X}(x)\sum_{y}W_x(z,y)
=\hat{P}_{\theta_\saa,X}(x)W_{Z|x}(z)
\overset{(b)}=\hat{P}_{\theta_\saa,X}(x)\hat{P}_{\theta_\sff^{\ddagger},Z|X}(z|x)\nonumber \\
\overset{(c)}=&
\frac{e^{
\sum_{j=1}^{n_1-1} \theta^j \bar{g}_j(x)}
}
{\sum_{x'} e^{
\sum_{j=1}^{n_1-1} \theta^j \bar{g}_j(x')}
}
\frac{e^{
\sum_{j=n_1}^{n_1n_2-1} \theta^{j,\ddagger} \bar{g}_j(x,z)
}}
{\sum_{z'}
e^{\sum_{j=n_1}^{n_1n_2-1} \theta^{j,\ddagger} \bar{g}_j(x,z')
}}\nonumber \\
\overset{(d)}=&
\frac{e^{
\sum_{j=1}^{n_1-1} (\theta^j+\psi_1^j(\theta_\sff^{\ddagger})) \bar{g}_j(x)}
e^{
\sum_{j=n_1}^{n_1n_2-1} \theta^{j,\ddagger} \bar{g}_j(x,z)
}}
{\sum_{x'} e^{
\sum_{j=1}^{n_1-1} \theta^j \bar{g}_j(x')}
e^{C_X(\theta_\sff^{\ddagger})}}
\nonumber \\
\overset{(e)}=&
\frac{e^{
\sum_{j=1}^{n_1-1} (\theta^j+\theta^{j,\ddagger}) \bar{g}_j(x)
+\sum_{j=n_1}^{n_1n_2-1} \theta^{j,\ddagger} \bar{g}_j(x,z)
}}
{\sum_{x',z'}
e^{\sum_{j=1}^{n_1-1} (\theta^j+\theta^{j,\ddagger}) \bar{g}_j(x')
+\sum_{j=n_1}^{n_1n_2-1} \theta^{j,\ddagger} \bar{g}_j(x',z')
}} \nonumber \\
=&\hat{P}_{(\theta_1+\theta_1^{\ddagger},\theta_2^{\ddagger}), XZ}(x,z),
\end{align}
where each step can be shown as follows.
$(a)$ follows from \eqref{HIA} and \eqref{XMP4B}.
$(b)$ follows from \eqref{XMPG}.
$(c)$ follows from \eqref{XC1T} and \eqref{XC3T}.
$(d)$ follows from \eqref{XZP2}.
$(e)$ follows from \eqref{CMR} and the fact that the denominator
$\sum_{x'} e^{
\sum_{j=1}^{n_1-1} \theta^j \bar{g}_j(x')}
e^{C_X(\theta_\sff^{\ddagger})}$ is a constant that does not depend on $x,z$.
\end{proof}

For any element $\theta_\saa=$ $(\theta^1,\ldots, $ $\theta^{n_1-1})^T
\in \Theta_{{\cal M}}$, 
we define 
$\phi_\sgg(\theta_\saa)=(\phi^{n_1n_2}, \ldots,$ $ \phi^{n_2(n_1+n_3-1)-1})^T$
as follows.
\begin{align}
P_{(\theta_\saa+\theta_\saa^{\dagger},
\theta_{\sdd}^{\dagger}),Y|Z}(y|z) 
=
\frac{e^{\sum_{j=n_1n_2}^{n_2(n_1+n_2-1)-1}\phi^{j} 
\bar{g}_j(z,y)}}{
\sum_{y' \in\Y} e^{\sum_{j=n_1n_2}^{n_2(n_1+n_2-1)-1}\phi^{j} 
\bar{g}_j(z,y')}}
\Label{XPQ2}
\end{align}
for $z \in \Z, y \in \Y$.
Then, for any element $\theta_\saa=(\theta^1,\ldots, \theta^{n_1-1})^T
\in \Theta_{{\cal M}}$, 
we choose
$\phi_\sff(\theta_\saa)=(\phi^{n_1}, \ldots, \phi^{n_1 n_2-1})^T$
and $C_{XZ}(\theta_\saa)$ such that
\begin{align}
&\sum_{j=n_1}^{n_1n_2-1} 
\phi^j \bar{g}_j(x,z)\nonumber \\
=&-\log \Big(
\sum_{y' \in\Y} e^{\sum_{j=n_1n_2}^{n_2(n_1+n_2-1)-1}\phi^{j} \bar{g}_j(z,y')}\Big)+C_{XZ}(\theta_1)
\Label{XPQ1}
\end{align}
for $x \in \X,z \in \Z$.
Then, we prepare the following lemma.
\begin{lemma}\Label{LC3BT}
The relation
\begin{align}
P_{(\theta_\saa+\theta_\saa^{\dagger},
\theta_{\sdd}^{\dagger}),Y|Z}
\times P_{(\theta_\saa+\theta_\saa^{\dagger},\theta_\sff^{\dagger}),XZ}
=P_{V(\theta_\saa+\theta_\saa^{\ddagger}, 
\theta_\sff^{\ddagger}+\phi_\sff(\theta_\saa),
\phi_\sgg(\theta_\saa) )^T}
\Label{XMP6}
\end{align}
holds for 
$\theta_\saa\in \Theta_{{\cal M}}$.
\end{lemma}

\begin{proof}
We have
\begin{align}
&P_{V(\theta_\saa+\theta_\saa^{\ddagger}, 
\theta_\sff^{\ddagger}+\phi_\sff(\theta_\saa),
\phi_\sgg(\theta_\saa) )^T,XZ}(x,z) \nonumber\\
\overset{(a)}=&\frac{e^{
\sum_{j=1}^{n_1-1} (\theta^j+\theta^{j,\ddagger}) \bar{g}_j(x)
+\sum_{j=n_1}^{n_1n_2-1} 
(\theta^{j,\ddagger}+ \phi^j)\bar{g}_j(x,z)}
\sum_{y}  e^{\sum_{j=n_1n_2}^{n_2(n_1+n_2-1)-1}\phi^{j} 
\bar{g}_j(z,y)}
}
{\sum_{x',z'}e^{
\sum_{j=1}^{n_1-1} (\theta^j+\theta^{j,\ddagger}) \bar{g}_j(x')
+\sum_{j=n_1}^{n_1n_2-1} 
(\theta^{j,\ddagger}+ \phi^j)\bar{g}_j(x',z')}
\sum_{y'}  e^{\sum_{j=n_1n_2}^{n_2(n_1+n_2-1)-1}\phi^{j} 
\bar{g}_j(z',y')}
}\nonumber \\
\overset{(b)}=&\frac{e^{
\sum_{j=1}^{n_1-1} (\theta^j+\theta^{j,\ddagger}) \bar{g}_j(x)
+\sum_{j=n_1}^{n_1n_2-1} 
\theta^{j,\ddagger}\bar{g}_j(x,z)+C_{XZ}(\theta_1)}
}
{\sum_{x',z'}e^{
\sum_{j=1}^{n_1-1} (\theta^j+\theta^{j,\ddagger}) \bar{g}_j(x')
+\sum_{j=n_1}^{n_1n_2-1} 
\theta^{j,\ddagger}\bar{g}_j(x',z')+C_{XZ}(\theta_1)}
}\nonumber \\
=&\frac{e^{
\sum_{j=1}^{n_1-1} (\theta^j+\theta^{j,\ddagger}) \bar{g}_j(x)
+\sum_{j=n_1}^{n_1n_2-1} 
\theta^{j,\ddagger}\bar{g}_j(x,z)}
}
{\sum_{x',z'}e^{
\sum_{j=1}^{n_1-1} (\theta^j+\theta^{j,\ddagger}) \bar{g}_j(x')
+\sum_{j=n_1}^{n_1n_2-1} 
\theta^{j,\ddagger}\bar{g}_j(x',z')}
}\nonumber \\
\overset{(c)}=&
\hat{P}_{\theta_\saa+\theta_\saa^{\ddagger}, \theta_{2}^{\ddagger},XZ}(x,z)
\Label{XMP3},
\end{align}
where $(a)$, $(b)$, and $(c)$ follow from \eqref{COA}, \eqref{XPQ1}, and \eqref{XC2T}, respectively.
The combination of \eqref{XMP4B} and \eqref{XMP3} yields that
\begin{align}
P_{V(\theta_\saa+\theta_\saa^{\ddagger}, 
\theta_\sff^{\ddagger}+\phi_\sff(\theta_\saa),
\phi_\sgg(\theta_\saa) )^T,XZ}
=P_{(\theta_\saa+\theta_\saa^{\dagger},
\theta_{\sdd}^{\dagger}), XZ}.
\Label{XMP4}
\end{align}

In the same way as $(a)$ of \eqref{XMP3}, we have
\begin{align}
&P_{V(\theta_\saa+\theta_\saa^{\ddagger}, 
\theta_\sff^{\ddagger}+\phi_\sff(\theta_\saa),
\phi_\sgg(\theta_\saa) )^T,ZY}(z,y) \nonumber\\
=&\frac{\sum_{x}  e^{
\sum_{j=1}^{n_1-1} (\theta^j+\theta^{j,\ddagger}) \bar{g}_j(x)
+\sum_{j=n_1}^{n_1n_2-1} 
(\theta^{j,\ddagger}+ \phi^j)\bar{g}_j(x,z)}
e^{\sum_{j=n_1n_2}^{n_2(n_1+n_2-1)-1}\phi^{j} 
\bar{g}_j(z,y)}
}
{\sum_{x',z',y'}e^{
\sum_{j=1}^{n_1-1} (\theta^j+\theta^{j,\ddagger}) \bar{g}_j(x')
+\sum_{j=n_1}^{n_1n_2-1} 
(\theta^{j,\ddagger}+ \phi^j)\bar{g}_j(x',z')}
  e^{\sum_{j=n_1n_2}^{n_2(n_1+n_2-1)-1}\phi^{j} 
\bar{g}_j(z',y')}
}
\nonumber\\
=&
\tau(z) e^{\sum_{j=n_1n_2}^{n_2(n_1+n_2-1)-1}\phi^{j} 
\bar{g}_j(z,y)},
\Label{XMP3G}.
\end{align}
where $\tau(z)$ is a constant that depends only on $z$.
Hence,
\begin{align}
&P_{V(\theta_\saa+\theta_\saa^{\ddagger}, 
\theta_\sff^{\ddagger}+\phi_\sff(\theta_\saa),
\phi_\sgg(\theta_\saa) )^T,Y|Z}(y|z) \nonumber\\
=&\frac{ 
e^{\sum_{j=n_1n_2}^{n_2(n_1+n_2-1)-1}\phi^{j} 
\bar{g}_j(z,y)}
}{
\sum_{y' \in\Y}
e^{\sum_{j=n_1n_2}^{n_2(n_1+n_2-1)-1}\phi^{j} 
\bar{g}_j(z,y')}
}\Label{XMP5H}.
\end{align}
Thus, \eqref{XPQ2} and \eqref{XMP5H} yield that
\begin{align}
P_{V(\theta_\saa+\theta_\saa^{\ddagger}, 
\theta_\sff^{\ddagger}+\phi_\sff(\theta_\saa),
\phi_\sgg(\theta_\saa) )^T,Y|Z}
=P_{(\theta_\saa+\theta_\saa^{\dagger},
\theta_{\sdd}^{\dagger}),Y|Z}.
\Label{XMP5}
\end{align}
Therefore, the combination of \eqref{XMP5} and \eqref{XMP4} implies \eqref{XMP6}.
\end{proof}

For $i=1, \ldots, n_1-1$, we have
\begin{align}
&\nabla_i^{(e)}[F](V(\theta_\saa+\theta_\saa^{\ddagger}, 
\theta_\sff^{\ddagger}+\phi_\sff(\theta_\saa),
\phi_\sgg(\theta_\saa) )^T)\nonumber \\
=&\sum_{x,y,z} g_i(x,y,z)
P_{V(\theta_\saa+\theta_\saa^{\ddagger}, 
\theta_\sff^{\ddagger}+\phi_\sff(\theta_\saa),
\phi_\sgg(\theta_\saa) )^T}(x,y,z)\nonumber \\
=&\sum_{x,y,z} \delta_i(x)
P_{V(\theta_\saa+\theta_\saa^{\ddagger}, 
\theta_\sff^{\ddagger}+\phi_\sff(\theta_\saa),
\phi_\sgg(\theta_\saa) )^T}(x,y,z)\nonumber \\
\overset{(a)}{=}&\sum_{x,y,z} \delta_i(x)
P_{(\theta_\saa+\theta_\saa^{\dagger},
\theta_{\sdd}^{\dagger}),Y|Z}(y|z)
P_{(\theta_\saa+\theta_\saa^{\dagger},
\theta_\sff^{\dagger}),XZ}
(x,z)\nonumber \\
=&\sum_{x,z} \delta_i(x)P_{\theta_\saa+\theta_\saa^{\dagger}, \theta_{\sdd}^{\dagger},XZ}
(x,z)\nonumber \\
=&\sum_{x,y,z} g_i(x,y,z)P_{\theta_\saa+\theta_\saa^{\dagger}, \theta_{\sdd}^{\dagger}}
(x,y,z)
=\nabla_i^{(e)}[F_{{\cal M}}](\theta_\saa+\theta_\saa^{\dagger}),\Label{SMY}
\end{align}
where $(a)$ follows from Lemma \ref{LC3BT}.

For $i'=n_1, \ldots, n_2(n_1+n_3-1)-1$, we have
\begin{align}
&\nabla_{i'}^{(e)}[F](V(\theta_\saa+\theta_\saa^{\ddagger}, 
\theta_\sff^{\ddagger}+\phi_\sff(\theta_\saa),
\phi_\sgg(\theta_\saa) )^T)\nonumber \\
=&\sum_{x,y,z} g_{i'}(x,y,z)
P_{V(\theta_\saa+\theta_\saa^{\ddagger}, 
\theta_\sff^{\ddagger}+\phi_\sff(\theta_\saa),
\phi_\sgg(\theta_\saa) )^T}(x,y,z)=0.\Label{SMY2}
\end{align}
Combining \eqref{XNII}, \eqref{SMY}, and \eqref{SMY2}, 
for $j=1,\ldots, n_2(n_1+n_3-1)-1$, we have
\begin{align}
&\nabla_j^{(e)}[F_{{\cal E}}](\theta_\saa+\theta_\saa^{\ddagger}, 
\theta_\sff^{\ddagger}+\phi_\sff(\theta_\saa),\phi_\sgg(\theta_\saa) )\nonumber \\
=&
\sum_{i=1}^{n_2(n_1+n_3-1)-1}
\nabla_i^{(e)}[F](V(\theta_\saa+\theta_\saa^{\ddagger}, 
\theta_\sff^{\ddagger}+\phi_\sff(\theta_\saa),
\phi_\sgg(\theta_\saa) )^T)v_{j}^i\nonumber \\
=&
\sum_{i=1}^{n_1-1}
\nabla_i^{(e)}[F_{{\cal M}}](\theta_\saa+\theta_\saa^{\dagger})
v_{j}^i
=(\nabla^{(e)}[F_{{\cal M}}](\theta_\saa+\theta_\saa^{\dagger}) V_1)_j. \Label{SMY3}
\end{align}
When $\eta_\saa=  \nabla^{(e)}[F_{{\cal M}}](\theta_\saa
+\theta_\saa^{\dagger})$,
\eqref{SMY3} guarantees that $\eta_\saa V_1=\nabla^{(e)}[F_{{\cal E}}](\theta_\saa+\theta_\saa^{\ddagger}, 
\theta_\sff^{\ddagger}+\phi_\sff(\theta_\saa),
\phi_\sgg(\theta_\saa) )$, which implies that
$\nabla_\saa^{(m)}[F_{{\cal E}}^*](\eta_\saa V_1 )
=\theta_\saa+\theta_\saa^{\ddagger}$.
Thus, we have
\begin{align}
\nabla^{(m)}[F_{{\cal M}}^*](\eta_\saa)-
\nabla_\saa^{(m)}[F_{{\cal E}}^*](\eta_\saa V_1 )
=\theta_\saa+\theta_\saa^{\dagger}
-(\theta_\saa+\theta_\saa^{\ddagger})
=\theta_\saa^{\dagger} - \theta_\saa^{\ddagger},\Label{XPS}
\end{align}
which implies Condition (B6).

In the following, we assume that 
the rank of $H$ is $n_1-1$.
We choose the parameter 
$\theta_{\sbb}^{\ddagger} \in \mathbb{R}^{n_1(n_2-1)+n_2(n_3-1)}$ 
such that
$(0_{n_1-1,n_1 (n_2-1)},H) \theta_{\sbb}^{\ddagger}=\theta_1^{\ddagger} $.
We choose $ (n_1(n_2-1)+n_2(n_3-1)) \times (n_1(n_2-2)+n_2(n_3-1)+1)$ 
matrix $G$ such that 
$\im G=\Ker (0_{n_1-1,n_1 (n_2-1)},H) $.
Then, $\overline{\cal E}$ and $\overline{\cal M}$ defined in \eqref{XC3} and \eqref{XC4}
are written as
\begin{align}
\overline{\cal E}=& \{ (\theta_\saa,
\theta_{\sbb}^{\ddagger}+ G \theta_\see)^T |
\theta_1 \in \mathbb{R}^{n_1-1}, \theta_\see \in \mathbb{R}^{
n_1(n_2-2)+n_2(n_3-1)+1} \} \\
\overline{\cal M}=& 
\left\{ (\theta_\saa,\theta_{\sbb})^T \left|
\begin{array}{l}
(\eta_\saa,\eta_\sbb)= 
\nabla^{(e)}[F_{{\cal E}}]((\theta_\saa,\theta_{\sbb})^T),\\
\eta^1 (0_{n_1-1,n_1 (n_2-1)},H)- \eta_\sbb=0 
\end{array}
\right.\right\} .
\end{align}

We choose the parameter $\theta_\sgg^{\ddagger} \in \mathbb{R}^{n_2(n_3-1)}$ 
such that
$H \theta_\sgg^{\ddagger}=\theta_1^{\ddagger} $.
We choose $ n_2(n_3-1) \times (n_2(n_3-1)+1-n_1)$ 
matrix $G$ such that 
$\im G=\Ker H $.
Then, $\overline{\cal E}$ and $\overline{\cal M}$ defined in \eqref{XC3T} and \eqref{XC4}
are written as
\begin{align}
\overline{\cal E}=& \{ (\theta_\saa,\theta_\sff,\theta_{\sgg}^{\ddagger}+ G \theta_\shh)^T |
\theta_\saa \in \mathbb{R}^{n_1-1} ,
\theta_\sff \in \mathbb{R}^{n_1(n_2-1)} ,
\theta_\shh \in \mathbb{R}^{n_2(n_3-1)-n_1+1} \} \\
\overline{\cal M}=& \{ (\theta_\saa,\theta_\sff,\theta_\sgg)^T |
(\eta_\saa,\eta_\sff,\eta_\sgg)= \nabla^{(e)}[F_{{\cal E}}]((\theta_\saa,\theta_\sff,\theta_\sgg)^T),~
\eta_\sff=0,~
\eta_\saa H= \eta_\sgg \} .
\end{align}

Therefore, 
when the intersection $\overline{\cal E} \cap \overline{\cal M}$ is not empty, 
due to Corollary \ref{Cor4},
the maximization \eqref{eq:Inf.st} is written 
by using an element $(\theta_{\saa,*},\theta_{\sbb,*})^T\in \overline{\cal E} \cap \overline{\cal M}$ as 
\begin{align}
C_{\sup}({\cal M},{\cal E})=D^F(\phi^{(e)}_{{\cal M}}(
\Gamma_{{\cal M}}^{(m),F}((\bar\theta_{\saa,*},\bar\theta_{\sbb,*})^T) \|
\phi^{(e)}_{{\cal E}}((\bar\theta_{\saa,*},\bar\theta_{\sbb,*})^T).
\Label{SXOY}
\end{align}

\section{Capacity of classical-quantum channel}\Label{S12}
\subsection{Problem setting}
Next, we discuss a classical-quantum channel from the classical system ${\cal C}:=\{1,\ldots, n_1\}$ to the quantum system
${\cal H}_A$ with dimension $n_2$, which is given as a set of density matrices $\{W_j\}_{j=1}^{n_1}$.
Under this classical-quantum channel,
given an input probability distribution $(p_j)$ on the classical system ${\cal C}$,
we define the classical-quantum state
$\rho[p]:= \sum_{j=1}^{n_1} p_j |j \rangle \langle j| \otimes W_j$ on ${\cal H}_A\otimes {\cal H}_C$,
where ${\cal H}_C$ is spanned by $\{ |i\rangle\}_{i=1}^{n_1}$.
Then, we denote the partial trace for ${\cal C}$ and ${\cal H}_A$ by $\Tr_C$ and $\Tr_A$, respectively.
The Hovelo quantity is defined as
\begin{align}
\sum_{j=1}^k p_j D \bigg(W_j\bigg\| \sum_{j'=1}^k p_{j'} W_{j'}\bigg)
=D\big(\rho[p]\big\| (\Tr_C \rho[p])\otimes (\Tr_A \rho[p])  \big) .
\end{align}
The capacity of the classical-quantum channel 
$\{W_j\}_{j=1}^{k}$ is defined as the maximum
\begin{align}
\max_{p \in {\cal P}_{{\cal C}}} \sum_{j=1}^k p_j D\bigg( W_j \bigg\| \sum_{j'=1}^k p_{j'} W_{j'}\bigg)
=\max_{p \in {\cal P}_{{\cal C}}} \min_{\rho \in {\cal S}({\cal H}_A), q \in {\cal P}_{{\cal C}} }
D(\rho[p]\| \rho \otimes q ) ,\Label{MOA3}
\end{align}
where the distribution $q$ is identified with the state $\sum_{i=1}^{n_1} q_i |i\rangle \langle i| $.
This quantity expresses the maximum transmission rate of classical information when 
we employ the classical-quantum channel  $\{W_j\}_{j=1}^{k}$ \cite{Holevo,SW}.

The set of $\rho'' \otimes q$ forms an exponential subfamily ${\cal E}$ and 
the set of $\rho[p]$ forms a mixture subfamily  ${\cal M}$.
Hence, the maximization problem \eqref{MOA3} is a special case of the maximization \eqref{eq:Inf.st} with
$k=n_1-1$, $l=n_1+n_2^2-2$, and $d=n_1 n_2^2-1$.
As shown as Lemma \ref{LOS3}, condition (B1) holds.
In the following, we apply Algorithm \ref{protocol2}.

\subsection{Constructions of vectors $u_1, \ldots, u_{n_1n_2^2-1},v_1,\ldots, v_{n_1+n_2^2-2}$}
For this aim, 
we need to choose a suitable coordinate to satisfy conditions (B3), (B4), and (B5)
and check Condition (B2).
For this aim, we choose $n_2^2-1$ linearly independent Hermitian matrices $X_j$ on ${\cal H}_A$ 
for $j=1, \ldots, n_2^2-1$ to satisfy the condition that
\begin{align}
\Tr X_j W_{n_1}=0
\end{align}
and the linear space spanned by $X_1, \ldots, X_{n_2^2-1}$ does not contain the identity matrix.
Then, we define the Hermitian matrices $\xi_1, \ldots, \xi_{n_1n_2^2-1}$ on 
${\cal H}_A\otimes {\cal H}_C$ as follows.
We define $\xi_i := I_A \otimes |i\rangle \langle i|$ for $i=1, \ldots, n_1-1$.
We define $\xi_{n_1-1+(i-1)(n_2-1)+j }:= 
(X_j -h_{i,j} I_A)\otimes |i\rangle \langle i| $
for $i=1, \ldots, n_1-1$ and $j=1, \ldots, n_2-1$, where
$h_{i,j}:=\Tr X_j W_{i}$.
We define $\xi_{(n_1-1)n_2+j }(x,y):=
X_j \otimes |n_1\rangle \langle n_1| $ for $j=1, \ldots, n_2-1$.
Then, we define the $C^{\infty}-$strictly convex function 
$F$ on $\mathbb{R}^{n_1n_2-1}$ as
\begin{align}
F(\theta):= \log 
\Tr \exp \Big(  \sum_{i=1}^{n_1 n_2^2-1} \theta^i \xi_i \Big).
\end{align}
We define the density matrices $\rho_\theta, \rho_{\theta,C}, \rho_{\theta,A}$ as
\begin{align}
\rho_\theta &:=e^{\sum_{i=1}^{n_1 n_2^2-1} \theta^i \xi_i-F(\theta) }, \\
\rho_{\theta,C} &:=\Tr_A e^{\sum_{i=1}^{n_1 n_2^2-1} \theta^i \xi_i-F(\theta) },\\
\rho_{\theta,A} &:=\Tr_C e^{\sum_{i=1}^{n_1 n_2^2-1} \theta^i \xi_i-F(\theta) }. 
\end{align}
Then, as a special case of \eqref{MGA2}, we have
\begin{align}
D^F(\theta\|\theta')= D(\rho_\theta\|\rho_{\theta'}).
\end{align}

Next, we choose the Matrix $U$ as the identity matrix,
and $u_1, \ldots, u_{n_1 n_2^2-1}$ are chosen as its $n_1 n_2^2-1$ column vectors.
Then, we define vector $v_1, \ldots, v_{n_1+n_2^2-2}$ as follows, 
whereas $V=(v_1, \ldots, v_{n_1+n_2^2-2})$.
We define $v_i:= u_i$ for $i=1, \ldots,n_1-1$.
We define $v_{n_1-1+j}:= 
\sum_{i=1}^{n_1} u_{n_1-1+(i-1)(n_2^2-1)+j }+ \sum_{i=1}^{n_1-1} h_{i,j} u_i$ for $j=1, \ldots,n_2^2-1$.
Then, we have
\begin{align}
\sum_{i=1}^{n_1-1}\xi_i  v_j^i=
\left\{
\begin{array}{ll}
I_A \otimes |j \rangle \langle j| & \hbox{when }j=1, \ldots, n_{1}-1 \\
X_{j-n_1+1 }\otimes I_C & \hbox{when }j=n_1, \ldots, n_{1}+n_2^2-1 .
\end{array}
\right.
\end{align}

\subsection{Parameterizations of ${\cal E}$ and ${\cal M}$}
Using 
$F_{{\cal E},\saa}(\theta^1, \ldots,\theta^{n_1-1}):= 
\log \sum_{x} e^{\sum_{i=1}^{n_1 -1} {\theta}^i \delta_i(x)}$
and
$F_{{\cal E},\sbb}(\theta^{n_1}, \ldots,\theta^{n_1+n_2^2-2})
:= \log \Tr e^{\sum_{j=1}^{n_2^2 -1} {\theta}^{n_1-1+j} X_j }
$,
we define the distributions on $\X$ and $\Y$ as
\begin{align}
\bar{\rho}_{\theta_\saa,C}:=&e^{\sum_{i=1}^{n_1 -1} 
{\theta}^i 
|i\rangle \langle i| -
F_{{\cal E},\saa}(\theta_\saa)} \\
\bar{\rho}_{\theta_\sbb,A}:=&
e^{\sum_{j=1}^{n_2 -1} {\theta}^{n_1-1+j} X_j
-F_{{\cal E},\sbb}(\theta_\sbb)}
\end{align}
for 
$\theta_\saa:=(\theta^1, \ldots,\theta^{n_1-1})$
and 
$\theta_\sbb:=(\theta^{n_1}, \ldots,\theta^{n_1+n_2^2-2})$.
Then, we have
\begin{align}
\rho_{\sum_{j=1}^{n_1-1}\theta_\saa^j v_j+\sum_{j'=1}^{n_2^2-1}\theta_\sbb^j 
v_{n_1-1+j}}
=\bar{\rho}_{\theta_\saa,C}\otimes 
\bar{\rho}_{\theta_\sbb,A}.
\end{align}
Hence, the set of product states is written as the exponential subfamily
${\cal E}:=\{ \rho_{\sum_{j=1}^{n_1+n_2^2-2} {\theta}^j v_j }\}$
generated by $v_1, \ldots, v_{n_1+n_2-2}$ at the point $(0, \ldots, 0)$.

We define the mixture family ${\cal M}$ by the constraint
$\sum_{i=1}^{n_1 n_2^2-1} u^i_{n_1-1+j'} \partial_i F(\theta)=0$ 
for $j'=1, \ldots, n_1 (n_2^2-1)$.
This constraint is equivalent to 
\begin{align}
\Tr \big((X_j -h_{i,j}I )\otimes |i\rangle \langle i|  \big)  \rho_{\theta}=0 , \quad
\Tr \big( X_j \otimes |n_1 \rangle \langle n_1| \big)  \rho_{\theta} =0 
\end{align}
for $i=1, \ldots, n_1-1$ and $j=1, \ldots, n_2^2-1$.
Hence, the mixture family ${\cal M}$ is composed of density matrices with the form $W \times q$.
Thus, the problem \eqref{MOA} is written as the problem \eqref{eq:Inf.st} with the above defined ${\cal E}$ and ${\cal M}$.

we choose $\theta_\sbb^{\dagger}=
(\theta^{n_1,\dagger}, \ldots, \theta^{n_1n_2^2-1,\dagger})$ as
\begin{align}
W_i=
\frac{ e^{\sum_{j=1}^{n_2^2-1}  \theta^{n_1-1+ (i-1) (n_2^2-1)+j,\dagger } 
(X_j-h_{i,j})
}}
{\Tr e^{\sum_{j'=1}^{n_2^2-1}  \theta^{n_1-1+ (i-1) (n_2^2-1)+j',\dagger } 
(X_{j'}-h_{i,j'})
}}.
\end{align}
In this choice, we have
\begin{align}
\log \Tr e^{\sum_{j'=1}^{n_2^2-1}  \theta^{n_1-1+ (i-1) (n_2^2-1)+j',\dagger } 
(X_{j'}-h_{i,j'})
}
=H(W_i)\Label{XMZ3}.
\end{align}
Then, ${\cal M}$ is written as
$\{( \theta_\saa, \theta_\sbb^{\dagger})| \theta_1 \in \mathbb{R}^{n_1-1}\}$.
That is,
${\cal M}$ forms an exponential subfamily generated by $u_1, \ldots, u_{n_1-1}$.
Hence, the maximization \eqref{MOA3} is rewritten as 
\begin{align}
\max_{p \in {\cal P}_{{\cal C}}} \min_{\rho \in {\cal S}({\cal H}_A), q \in {\cal P}_{{\cal C}} }
D(\rho[p]\| \rho \otimes q ) 
=&
\max_{\theta \in \mathcal{M}} D^{F}(\theta \| \Pro^{(m),F}_{\mathcal{E}} (\theta))
\nonumber \\
=&
\max_{\theta \in \mathcal{M}} 
\min_{\theta' \in \mathcal{E}} 
D^{F}(\theta \| \theta').\Label{CKP3}
\end{align}

\subsection{Check of Conditions (B2), (B3), (B4), and (B5)}
We define the $(n_1-1) \times (n_2^2-1)$ matrix $H:=(h_{i,j})$.
Then, we find that the $ (n_1-1)\times (n_1+n_2^2-2)$ matrix $V_1$ is $ (I,H)$.
That is, the $(n_1-1) \times (n_2^2-1)$ matrix $V_3$ is $H$.
Hence, conditions (B3) and (B4) hold.
In the exponential family ${\cal E}$, we have 
$F_{{\cal E}}(\bar{\theta})=
\log \sum_{x} e^{\sum_{i=1}^{n_1 -1} {\theta}^i \delta_i(x)}
+\log \Tr e^{\sum_{j=1}^{n_2^2 -1} {\theta}^{n_1-1+j} X_j }$.
Hence, the condition (B5) holds.
Therefore, we can apply 
Algorithm \ref{protocol2} with Condition (B5).

Since we have
\begin{align}
\rho_{\Pro^{(m),F}_{{\cal E}}  (\theta)}= \rho_{\theta,C}\otimes \rho_{\theta,A}.
\end{align}
for any $\theta$, we have
\begin{align}
 D^{F}(\theta'\|\theta)
=&D(\rho_\theta\|\rho_{\theta'})
= D(\rho_{\theta,C}\|\rho_{\theta',C})
 \le
 D(\rho_{\theta,C}\|\rho_{\theta',C}) + D(\rho_{\theta,A}\|\rho_{\theta',A})
 \nonumber \\
=&
 D(\rho_{\theta,A}\otimes \rho_{\theta,C}\|\rho_{\theta',A}\otimes \rho_{\theta',C})
=
 D(\rho_{ \Pro^{(m),F}_{{\cal E}}  (\theta')}\| (\rho_{ \Pro^{(m),F}_{{\cal E}}  (\theta)} ) 
\nonumber \\
=&
 D^{F}( \Pro^{(m),F}_{{\cal E}}  (\theta')\|  \Pro^{(m),F}_{{\cal E}}  (\theta) ) 
\end{align}
for $\theta,\theta' \in {\cal M}$.
Thus, the condition (B2) holds.
Therefore, Theorem \ref{theo:conv:BBem} guarantees the global convergence.
When $\theta^{(1)}$ is the uniform distribution on ${\cal X}$,
in the same way as \eqref{CKP}, we can show that
the supremum $\sup_{\theta \in \mathcal{M}} D^F(\theta \| \theta^{(1)})$
equals $ \log n_1$.
Therefore, 
when Theorem \ref{theo:conv:BBem} is applied,
we obtain the precision \eqref{NHG}
with $ \frac{\log n_1}{\epsilon}$ iterations.

Now, with the above choice of $\theta^{(1)}$, we consider the case 
when the density matrices $\{W_x\}_{x}$ are linearly independent.
In the same way as Section \ref{S10},
there exists $\alpha>0$ to satisfy the condition (B2+).
Hence, we can apply Theorem \ref{theo:conv:BBem+} instead of Theorem \ref{theo:conv:BBem}.
When $\theta^{(1)}$ is the uniform distribution on ${\cal X}$,
we obtain the precision \eqref{NHG+}
with $\frac{ \log  \log n_1  -\log \epsilon}{\log (1+\alpha)}$ iterations. 

\subsection{Non-iterative method}
Next, we characterize the maximization \eqref{MOA3} without an iterative method. 
To check Condition (B6+),
we prepare the following lemmas.
\begin{lemma}\Label{LC3B3}
The relation
\begin{align}
\rho_{(\theta_\saa +\theta_\saa^{\dagger},
\theta_\sbb^{\dagger}),C}=
\bar{\rho}_{\theta_\saa ,C}.\Label{ACU3}
\end{align}
holds, where
$\theta_\saa^{\dagger}=
(\theta^{1,\dagger}, \ldots, \theta^{n_1-1,\dagger})$
is defined as
$\theta^{i,\dagger}:=-H(W_i)+H(W_{n_1})$
for $i=1\ldots,n_1-1$.
\end{lemma}

\begin{proof}
We define
$\theta_{(i)}^{\dagger}=
(\theta^{n_1-1+(i-1)(n_2^2-1)+1 ,\dagger}, \ldots, 
\theta^{n_1-1+i(n_2^2-1),\dagger}) \in \mathbb{R}^{n_1-1}$.
Since $W_i=\bar{\rho}_{ \theta_{(i)}^{\dagger} ,A}$, 
we have $W_i= e^{
\sum_{j=1}^{n_2^2-1}\theta^{n_1-1+(i-1)(n_2^2-1)+j,\dagger}
(X_j-h_{i,j}) - F_{{\cal E},\sbb}(\theta_{(i)}^{\dagger})} $.
Because
\begin{align}
e^{F_{{\cal E},\sbb}(\theta_{(i)}^{\dagger})}=
\Tr e^{\sum_{j=1}^{n_2^2-1}\theta^{n_1-1+(i-1)(n_2^2-1)+j,\dagger}
(X_j -h_{i,j})},
\end{align}
\eqref{XMZ3} implies the relation
\begin{align}
H( W_i)= -\Tr W_i \log W_i  
=F_{{\cal E},\sbb}(\theta_{(i)}^{\dagger})\Label{XOS3}
\end{align}
for $i=1, \ldots, n_1$.

Now, we choose $\theta_1' \in \mathbb{R}^{n_1-1}$ such that
\begin{align}
\rho_{(\theta_\saa +\theta_\saa^{\dagger} ,\theta_\sbb^{\dagger}),C}=
\bar{\rho}_{\theta_\saa' ,C}.
\end{align}
Since we have
\begin{align}
\langle n_1 |\rho_{(\theta_\saa +\theta_\saa^{\dagger} ,\theta_\sbb^{\dagger}),C}
|n_1\rangle =
e^{ F_{{\cal E},\saa}(\theta_{(n_1)}^{\dagger})
-F_{{\cal M}}(\theta_\saa +\theta_\saa^{\dagger} ,\theta_\sbb^{\dagger})} ,
\end{align}
the relation $P_{(\theta_\saa +\theta_\sbb^{\dagger} ,\theta_\sbb^{\dagger}),X}(n_1)=
\bar{P}_{\theta_1' ,X}(n_1)$ yields 
\begin{align}
e^{ F_{{\cal E},\sbb}(\theta_{(n_1)}^{\dagger})
-F_{{\cal M}}(\theta_\saa +\theta_\saa^{\dagger} ,\theta_\sbb^{\dagger})} 
=
e^{ -F_{{\cal E},\saa}(\theta_\saa')}.\Label{XOCS3}
\end{align}
For $x\neq n_1$, we have 
\begin{align}
&\langle x| \bar{\rho}_{\theta_\saa' ,C}|x\rangle
=\langle x| \rho_{(\theta_\saa +\theta_\saa^{\dagger} ,\theta_\sbb^{\dagger}),C}|x\rangle
=
e^{\theta_\saa^x+ \theta_\saa^{x,\dagger}+ F_{{\cal E},\sbb}(\theta_{(x)}^{\dagger})
-F_{{\cal M}}(\theta_\saa +\theta_\saa^{\dagger} ,\theta_\sbb^{\dagger})} \nonumber \\
\overset{(a)}=&
e^{\theta_\saa^x+ \theta_\saa^{x,\dagger}
+ F_{{\cal E},\sbb}(\theta_{(x)}^{\dagger})
-F_{{\cal E},\sbb}(\theta_{(n_1)}^{\dagger})
-F_{{\cal E},\saa}(\theta_1')
}\overset{(b)}=e^{\theta^x-F_{{\cal E},\saa}(\theta_\saa')} ,
\end{align}
where $(a)$ and $(b)$ follow from 
\eqref{XOCS3} and
the pair of \eqref{XOS3} and the definition of 
$\theta_\saa^{x,\dagger}$, respectively.
This relation shows \eqref{ACU3}.
\end{proof}

In the same way as the end of the previous subsection, 
we assume that the distributions $\{W_x\}_{x}$ are linearly independent.
Then, the rank of $H$ is $n_1-1$.
The combination of this fact and Lemma \ref{LC3B}
guarantees 
\begin{align}
\nabla^{(m)}[F_{{\cal M}}^*](\eta^1)-
\nabla^{(m)}[F_{{\cal E},1}^*](\eta^1)
=\theta_\saa +\theta_\saa^{\dagger}-\theta_\saa 
=\theta_\saa^{\dagger},
\end{align}
which implies Condition (B6+).
We choose the parameter $\theta_\sbb^{\ddagger} \in \mathbb{R}^{n_2^2-1}$ such that
$H \theta_\sbb^{\ddagger}=\theta_\saa^{\dagger} $.
We choose $(n_2^2-1)\times (n_2^2-n_1)$ matrix $G$ such that 
$\im G=\Ker H$.
Then, ${\cal E}_\sbb$ and ${\cal M}_\sbb$ defined in \eqref{XC1T} and \eqref{XC2T}
are written as
\begin{align}
{\cal E}_\sbb &=\{ \theta_\sbb^{\ddagger}
+ G \theta_\see| \theta_\see 
\in \mathbb{R}^{n_2^2-n_1}\},\\
{\cal M}_\sbb & = \{ \theta_\sbb \in \mathbb{R}^{n_2^2-1}|
\nabla^{(e)}[F_{{\cal E},\sbb}](\theta_{2})G =0 \}.
\end{align}
As explained in Subsection \ref{S47},
the intersection ${\cal E}_\sbb\cap {\cal M}_\sbb$ is composed of a unique element.
As the solution of the following minimization \eqref{XOE2},
we choose $\theta_\see^{\ddagger}$ as
\begin{align}
\bar\theta_{\see}^{\ddagger}
:=\argmin_{\theta_\see}F_{{\cal E},\sbb}(\theta_{\sbb}^{\ddagger}+ G \theta_\see).
\Label{XOE2}
\end{align}
Then, we set
$ \bar\theta_\sbb^{\ddagger}:=\theta_\sbb^{\ddagger}
+ G \theta_\see^{\ddagger} \in {\cal E}_\sbb\cap 
{\cal M}_\sbb$.
Due to Theorem \ref{Cor6-1}, when there exists 
$\theta_\saa^{\ddagger} \in \mathbb{R}^{n_1-1} $ such that
$ \nabla^{(e)}[F_{{\cal E},\sbb}](\bar\theta_{\sbb}^{\ddagger})
=
\nabla^{(e)}[F_{{\cal E},\saa}](\theta_{\saa}^{\ddagger})H $, which is equivalent to
\begin{align}
W\cdot \bar{\rho}_{\theta_{\saa}^{\ddagger}, C} 
= \bar{\rho}_{\bar\theta_\sbb^{\ddagger} ,A}.
\Label{ACA3}
\end{align}
the maximizer in \eqref{CKP3} is 
$(\theta_{\saa}^{\ddagger}+\theta_{\saa}^{\dagger},
\theta_{\sbb}^{\ddagger} )\in {\cal M} $.
In addition, the maximum \eqref{CKP3} is 
\begin{align}
D(\rho_{(\theta_{\saa}^{\ddagger}+\theta_{\saa}^{\dagger},\theta_{\sbb}^{\dagger} )}
\|\bar{\rho}_{\theta_{\saa}^{\ddagger},C}\times
\bar{\rho}_{\bar\theta_{\sbb}^{\ddagger},A}
)
=-H(W_{n_1})
+F_{{\cal E},\sbb}(\bar\theta_{\sbb}^{\ddagger}),
\end{align}
because
\begin{align}
&D(W_x\|\bar{\rho}_{\bar\theta_{\sbb}^{\ddagger},A})
=\Tr W_x (\log W_x - 
\log \bar{\rho}_{\bar\theta_{\sbb}^{\ddagger},A})\nonumber \\
=&-H(W_x)-\Tr W_x
\Big(\sum_{j=1}^{n_2^2-1}\bar\theta^{n_1-1+j,\ddagger}X_j -F_{{\cal E},\sbb}(\bar\theta_{\sbb}^{\ddagger})\Big)\nonumber \\
=&-H(W_x)- \sum_{j=1}^{n_2^2-1}\bar\theta^{n_1-1+j,\ddagger}h_{x,j}
+F_{{\cal E},\sbb}(\bar\theta_{\sbb}^{\ddagger})\nonumber \\
=&-H(W_x)- \sum_{j=1}^{n_2^2-1}\theta^{n_1-1+j,\ddagger}h_{x,j}
+F_{{\cal E},2\sbb}(\bar\theta_{\sbb}^{\ddagger})\nonumber \\
=&-H(W_x)- \theta^{x,\dagger}
+F_{{\cal E},\sbb}(\bar\theta_{\sbb}^{\ddagger})\nonumber \\
=&-H(W_x)- (-H(W_x)+H(W_{n_1}))
+F_{{\cal E},\sbb}(\bar\theta_{\sbb}^{\ddagger})\nonumber \\
=&-H(W_{n_1})
+F_{{\cal E},\sbb}(\bar\theta_{\sbb}^{\ddagger}).
\end{align}

When $n_1=n_2^2$, we have $l=n_1-1=n_2^2-1 $, which enables us to apply 
Corollary \ref{Cor6}.
In this case,
as another typical case, we can choose the matrices $X_j$ such that 
$(\Tr X_j W_i)_{1\le i,j\le n_2^2-1}$ is is the identity matrix
Under this choice,  
the calculation of the maximization \eqref{MOA} based on Corollary \ref{Cor6} 
is done by Algorithm 1 in the reference \cite{exact}.
Therefore, the method based on Theorem \ref{Cor6-1}
can be considered as a generalization of Algorithm 2 in the reference \cite{exact}.


However, there is a case that no distribution $P_X $ on $\X$ satisfies \eqref{ACA}.
In this case, instead of a distribution on $\X$,
there exists a function $f_X$ on $\X$ such that
\begin{align}
\sum_{x \in \X} f_X(x) W_x= \bar{\rho}_{\bar\theta_\sbb^{\ddagger} ,A},~
\sum_{x \in \X} f_X(x)=1.
\end{align}
Also, 
there does not exist the maximum in \eqref{CKP3},
and 
the maximum \eqref{MOA3} is achieved in the boundary of $\P_\X$.
When we remove an element $x\in \X$, we have a subset 
$P_{\X \setminus \{x\}}$ of the boundary.
That is, the boundary is composed of this type of subsets.
Hence, to obtain the maximum \eqref{MOA3}, we need to apply the method in this subsection to the case when the channel is defined in the above type of subset.

In summary, in the same way as the capacity of the classical channel, 
the capacity of the classical-quantum channel can be calculated
with an algorithm similar to Algorithm \ref{protocol1}.

\section{Conclusion}\Label{Sec:conc}
In our study, we have tackled the reverse em-problem within the general framework of Bregman divergence. 
We have formulated this problem as the maximization of the minimum divergence between a mixture family and an exponential family, and proposed various methods to address it.

Our first method involves the development of the reverse em-algorithm using Bregman divergence. 
We have shown the convergence of this algorithm to the true value and analyzed its convergence speed under conditions that align with information-theoretical problem settings. 
We have applied this approach to problems related to channel capacity, including quantum settings. 
This method was initially proposed by Toyota in the context of calculating the classical channel capacity \cite{Shoji}. 
However, Toyota's work did not establish the existence of the inverse map of the map $\Pro^{(e),F}{{\cal M}}\circ \Pro^{(m),F}{{\cal E}}|_{{\cal M}}$. In Theorem \ref{theo:conv:BBem1}, 
we have shown that the inverse map uniquely exists under our Condition (B3) within the general framework of Bregman divergence. Furthermore, in Section \ref{S10}, 
we have shown that the case of classical channel capacity satisfies our Condition (B3). 
Consequently, we have successfully solved the problem originally proposed by Toyota \cite{Shoji}. 
Theorem \ref{theo:conv:BBem1} also provides the form of the inverse map through the minimization of a convex function. Moreover, in Section \ref{S46}, we have derived a simpler form of the inverse map under additional conditions.

In the second method, we have successfully transformed the reverse em-problem into em-problems by imposing 
the conditions introduced earlier. 
In this method, the reverse em-problem is converted to finding the intersection between an exponential family
and a mixture family.
The intersection is characterized by solving the em-problem between the exponential family and the mixture family. 

In the third method, we have strengthened the conditions and achieved an even more simplified approach. Under these stronger conditions, the reverse em-problem is converted into a convex minimization problem. The convex function involved in this minimization is a part of the function used to define the exponential family. Importantly, this convex function is simpler compared to the objective function $D^{F}(\theta | \Pro^{(m),F}_{\mathcal{E}} (\theta))$ that needs to be maximized in the original reverse em-problem. Notably, in specific cases where the original reverse em-problem satisfies certain conditions, this problem can be solved without requiring the additional minimization step. 
When applied to the classical channel capacity, this special case 
coincides with the algorithm proposed in the recent paper 
\cite{exact}. Consequently, this method can be regarded as a generalization of the approach presented in that paper \cite{exact}.

In the subsequent sections, we have shown that various concrete models, including those in the quantum setting, satisfy the conditions introduced in Section \ref{Sec:BBem}. Furthermore, we have established that these models also fulfill several conditions presented in this paper. Additionally, we have provided a detailed algorithm for calculating the classical channel capacity, which serves as a generalization of the method proposed in the recent paper \cite{exact}. Moreover, we have performed numerical calculations using this algorithm for cases that cannot be handled by the existing method \cite{exact}.

As an additional contribution, in Subsection \ref{S45}, 
we have introduced the quadratic approximation in each iteration of our proposed algorithm, Algorithm \ref{protocol2M}. 
However, we have not extensively discussed the convergence speed or computational complexity in various applications. This analysis is a topic for future research, which includes comparing our method with existing approaches.

The results obtained illustrate the effectiveness of information geometry as a conversion method for optimization problems. A key aspect of information geometry lies in the choice of parameterization associated with an exponential family and a mixture family. By leveraging this structure, we have successfully derived alternative characterizations of the original problems. Consequently, we can anticipate that the application of information geometry will lead to further valuable conversions in important optimization problems. 
In this way, our findings shed light on this novel application of information geometry, expanding its potential uses.
For example, we can consider the application of our result to 
the channel capacity of channels with Markovian memory.
This topic was studied in the preceding studies \cite{Kavcic,Vontobel,Wu}.
Since information geometry of Markovian process can be handled as a special case
of Bgregman divergence system \cite{Nakagawa,Nagaoka,HW}, 
our method can be expected to applied this topic.

\section*{Acknowledgments}
The author was supported in part by the National Natural Science Foundation of China (Grant No. 62171212). 
The author is very grateful to Mr. Shoji Toyota for helpful
discussions and for explaining the achievements of the reference \cite{Shoji}.
In particular, he explained to the author what problems were not solved in the reference \cite{Shoji}.
In addition, he pointed out that 
the secrecy capacity can be written as the reverse em-algorithm in a similar way as the channel capacity \cite{Toyota}.

\section*{Data availability}
Data sharing is not applicable to this article as no datasets were generated
or analyzed during the current study.

\section*{Conflict of interest}

There are no competing interests.

\appendix

\section{Proof of Theorem \ref{theo:conv:BBem}}\Label{A1}
Let $\theta_{(t)} $ be $\Pro^{(m),F}_{{\cal E}}(\theta^{(t)} ) $.
For any $\epsilon_1>0$, we choose an element $\theta(\epsilon_1) $ of ${\cal M}$ 
such that 
$D^F( \theta(\epsilon_1)\| \Pro^{(m),F}_{{\cal E}}(\theta(\epsilon_1)) )
\ge C_{\sup}(\mathcal{M},\mathcal{E})-\epsilon_1$.
Also, let $\theta(\epsilon_1)_* $ be $\Pro^{(m),F}_{{\cal E}}(\theta(\epsilon_1) ) $.


\begin{figure}[htbp]
\begin{center}
  \includegraphics[width=0.7\linewidth]{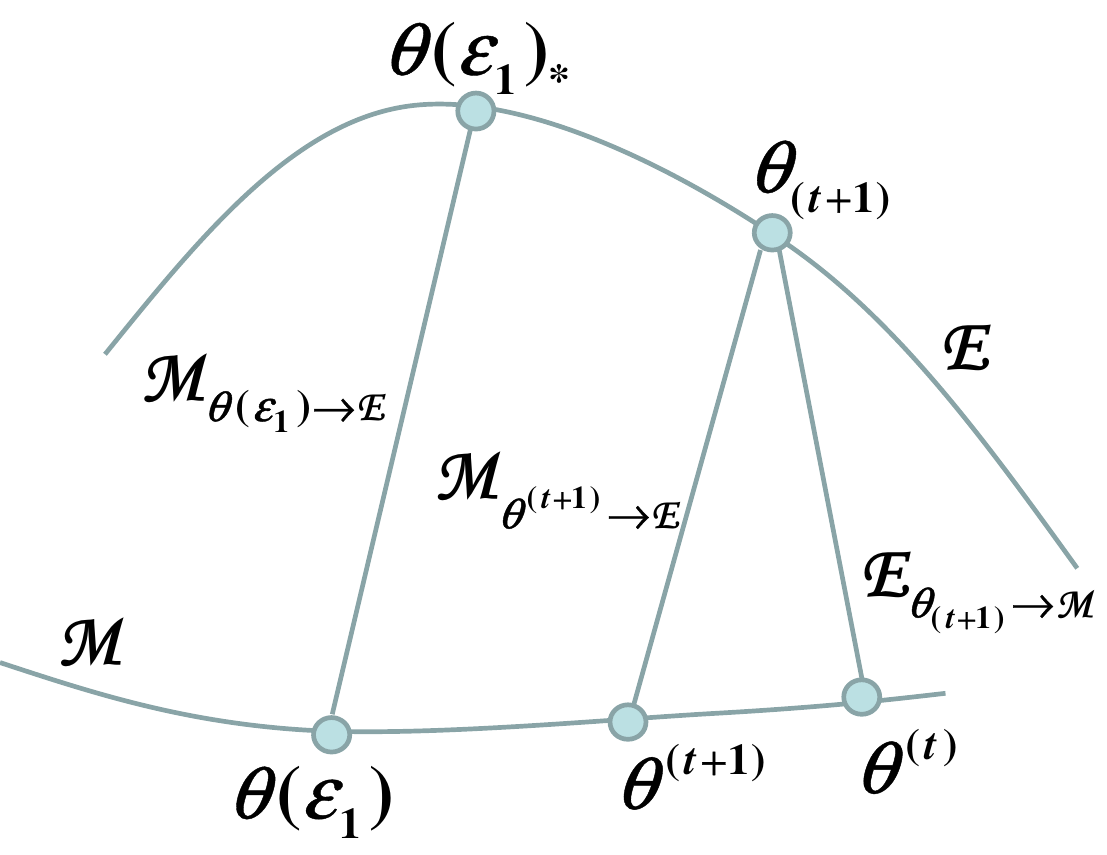}
  \end{center}
\caption{
Algorithms \ref{protocol2} and \ref{protocol2M}: 
This figure shows the topological relation among
$\theta(\epsilon_1 )_*$, $\theta(\epsilon_1 )$,
$\theta_{(t+1 )}$, $\theta^{(t+1 )}$, and $\theta^{(t)}$,
which is used in the
application of Pythagorean theorem (Proposition \ref{MNL}).
$\mathcal{M}_{\theta(\epsilon_1)\to \mathcal{E}}$
and 
$\mathcal{M}_{\theta^{(t+1)} \to \mathcal{E}}$
are the mixture subfamilies to project $\theta(\epsilon_1)$
and $\theta^{(t+1)}$
to the exponential subfamily $\mathcal{E}$, respectively.
$\mathcal{E}_{\theta_{(t+1)}\to \mathcal{M}}$
is the exponential subfamily to project $\theta_{(t+1)}$
to the mixture subfamily $\mathcal{M}$.
}
\Label{FF2}
\end{figure}   

As explained in Fig. \ref{FF2}, 
Pythagorean theorem (Proposition \ref{MNL}) guarantees that
the divergence $D^F(\theta(\epsilon_1)\| \theta_{(t+1)}) $ can be written in the following two ways;
\begin{align}
& D^F(\theta(\epsilon_1)\| \theta^{(t)})
+D^F(\theta^{(t)}\| \theta_{(t+1)})
=
D^F(\theta(\epsilon_1)\| \theta_{(t+1)}) \nonumber\\
=&
D^F(\theta(\epsilon_1)\| \theta(\epsilon_1)_*)
+D^F(\theta(\epsilon_1)_*\| \theta_{(t+1)}).
\end{align}
Hence,
\begin{align}
& 
C_{\sup}(\mathcal{M},\mathcal{E})-\epsilon_1- D^{F}(\theta^{(t)}\| \Pro^{(m),F}_{{\cal E}}  (\theta^{(t)}) )  \nonumber \\
=& 
D^F(\theta(\epsilon_1)\| \theta(\epsilon_1)_*) 
-D^{F}(\theta^{(t)}\| \Pro^{(m),F}_{{\cal E}}  (\theta^{(t)}) ) \nonumber \\
=& D^F(\theta(\epsilon_1)\| \theta(\epsilon_1)_*) 
-D^{F}(\theta^{(t)}\| \theta_{(t+1)} ) \nonumber\\
=& D^F(\theta(\epsilon_1)\| \theta^{(t)})
-D^F(\theta(\epsilon_1)_*\| \theta_{(t+1)}) \nonumber \\
=& D^F(\theta(\epsilon_1)\| \theta^{(t)})-D^F(\Pro^{(m),F}_{{\cal E}} (\theta(\epsilon_1))\| 
\Pro^{(m),F}_{{\cal E}} (\theta^{(t+1)})) \nonumber\\
\le &
 D^F(\theta(\epsilon_1)\| \theta^{(t)})-D^F(\theta(\epsilon_1)\| \theta^{(t+1)}) ,
\Label{PIT}
\end{align}
where the final inequality follows from condition (B2).
Thus,
\begin{align}
&\sum_{i=1}^{t}C_{\sup}(\mathcal{M},\mathcal{E})-\epsilon_1
- D^{F}(\theta^{(i)}\| \Pro^{(m),F}_{{\cal E}}  (\theta^{(i)}) )  \nonumber \\
\le &
\sum_{i=1}^{t}
 D^F(\theta(\epsilon_1)\| \theta^{(i)})-D^F(\theta(\epsilon_1)\| \theta^{(i+1)}) \nonumber\\
= &
 D^F(\theta(\epsilon_1)\| \theta^{(1)})-D^F(\theta(\epsilon_1)\| \theta^{(t+1)})
\le  D^F(\theta(\epsilon_1)\| \theta^{(1)}) \nonumber \\
\le & \sup_{\theta \in \mathcal{M}}  D^F(\theta\| \theta^{(1)}) .\Label{suma}
\end{align}
Taking the limit $\epsilon_1 \to 0$, we have
\begin{align}
\sum_{i=1}^{t}C_{\sup}(\mathcal{M},\mathcal{E})- D^{F}(\theta^{(i)}\| \Pro^{(m),F}_{{\cal E}}  (\theta^{(i)}) ) 
\le \sup_{\theta \in \mathcal{M}}  D^F(\theta\| \theta^{(1)}) .\Label{suma80}
\end{align}
Since
\begin{align}
&D^{F}(\theta^{(i+1)}\| \Pro^{(m),F}_{{\cal E}}  (\theta^{(i)}) ) 
=
D^{F}(\theta^{(i+1)}\| \theta_{(i+1)} ) 
\nonumber \\
\ge &
D^{F}(\theta^{(i)}\| \theta_{(i+1)} ) 
\ge
D^{F}(\theta^{(i)}\| \Pro^{(m),F}_{{\cal E}}(\theta^{(i)}) ) ,
\end{align}
for $i \le t$, we have
\begin{align}
C_{\sup}(\mathcal{M},\mathcal{E})- D^{F}(\theta^{(t)}\| \Pro^{(m),F}_{{\cal E}}  (\theta^{(t)}) ) 
\le 
C_{\sup}(\mathcal{M},\mathcal{E})- D^{F}(\theta^{(i)}\| \Pro^{(m),F}_{{\cal E}}  (\theta^{(i)}) ) . \Label{MLG1}
\end{align}
The combination of \eqref{suma80} and \eqref{MLG1} implies that
\begin{align}
C_{\sup}(\mathcal{M},\mathcal{E})- D^{F}(\theta^{(t)}\| \Pro^{(m),F}_{{\cal E}}  (\theta^{(t)}) )
\le \frac{1}{t}\sup_{\theta \in \mathcal{M}}  D^F(\theta\| \theta^{(1)}) ,\Label{suma8}
\end{align}
which implies \eqref{NHG}.

When the inequality
\begin{align}
C_{\sup}(\mathcal{M},\mathcal{E})- D^{F}(\theta^{(t)}\| \Pro^{(m),F}_{{\cal E}}  (\theta^{(t)}) ) 
\ge c(\frac{1}{t})
\end{align}
holds with a constant $c>0$, 
\eqref{suma80} yields
\begin{align}
\infty = \sum_{t=1}^{\infty} c(\frac{1}{t}) 
\le \sup_{\theta \in \mathcal{M}}D^F(\theta\| \theta^{(1)}),
\end{align}
which implies the contradiction. 
Hence, we obtain \eqref{mma}.

Indeed, when the maximum in \eqref{eq:Inf.st} exists, i.e., 
$\theta^*(\mathcal{M},\mathcal{E})$ exists,
the supremum $\sup_{\theta \in \mathcal{M}} D^F(\theta \| \theta^{(1)})$ 
in the above evaluation is replaced by
$D^F(\theta^*(\mathcal{M},\mathcal{E}) \| \theta^{(1)})$
because $\theta(\epsilon_1)$ is replaced by $\theta^*(\mathcal{M},\mathcal{E})$.

\section{Proof of Theorem \ref{theo:conv:BBem+}}\Label{A2}
We use the same notation as the proof of Theorem \ref{theo:conv:BBem}.
We denote $\theta^*(\mathcal{M},\mathcal{E})$ 
and $\Pro^{(m),F}_{{\cal E}}(\theta^*(\mathcal{M},\mathcal{E}))$
by $\theta^*$ and $\theta_*$, respectively.
Also, set $\alpha:= \alpha(\theta_{(1)})$.
Using \eqref{PIT} with $\epsilon_1=0$, we have 
\begin{align}
0 \le & 
C_{\sup}(\mathcal{M},\mathcal{E})- D^{F}(\theta^{(t)}\| \Pro^{(m),F}_{{\cal E}}  (\theta^{(t)}) )  \nonumber \\
=& D^F(\theta^*\| \theta^{(t)})-D^F(\Pro^{(m),F}_{{\cal E}} (\theta^*) \| 
\Pro^{(m),F}_{{\cal E}} (\theta^{(t+1)})) \Label{BK2}\\
\le &
 D^F(\theta^*\| \theta^{(t)})-  D^F(\theta^*\| \theta^{(t+1)}),\Label{BK1}
 \end{align}
which implies $ D^F(\theta^*\| \theta^{(t+1)}) \le D^F(\theta^*\| \theta^{(1)})$.
Thus, the condition (B2+) implies $(1+\alpha)D^F(\theta^*\| \theta^{(t+1)}) \le
D^F(\Pro^{(m),F}_{{\cal E}} (\theta^*) \| 
\Pro^{(m),F}_{{\cal E}} (\theta^{(t+1)}))$.
Combining \eqref{BK2}, we have
$0\le D^F(\theta^*\| \theta^{(t)})-  (1+\alpha)D^F(\theta^*\| \theta^{(t+1)}). $
Thus, we have
\begin{align}
D^F(\theta^*\| \theta^{(t+1)})
 \le
(1+\alpha)^{-1} D^F(\theta^*\| \theta^{(t)}), 
 \end{align}
which implies that
\begin{align}
D^F(\theta^*\| \theta^{(t)})
 \le
(1+\alpha)^{-(t-1)} D^F(\theta^*\| \theta^{(1)}). 
\end{align}
Using \eqref{PIT} with $\epsilon_1=0$, we have 
\begin{align}
& 
C_{\sup}(\mathcal{M},\mathcal{E})- D^{F}(\theta^{(t)}\| \Pro^{(m),F}_{{\cal E}}  (\theta^{(t)}) )  \nonumber \\
\le &
 D^F(\theta^*\| \theta^{(t)})-D^F(\theta^*\| \theta^{(t+1)}) 
\le
 D^F(\theta^*\| \theta^{(t)})
 \le
(1+\alpha)^{-(t-1)} D^F(\theta^*\| \theta^{(1)}).
 \end{align}
 Hence, we obtain \eqref{mma+}.

\section{Proofs of Lemmas \ref{BXPR} and \ref{BXP}}\Label{AL3}
We show Lemmas \ref{BXPR} and \ref{BXP}.
\begin{proofof}{Lemma \ref{BXPR}}
Eq. \eqref{MAF} follows from the following relation.
\begin{align}
&\eta_\saa V_1
=\nabla_{\saa}^{(m)}[F](\psi_{\mathcal{M}}^{(m)} (\eta_\saa))V_1
=\nabla^{(m)}[F](\psi_{\mathcal{M}}^{(m)} (\eta_\saa))V \nonumber\\
\stackrel{(a)}{=}&
\nabla^{(m)}[F]
(\Pro^{(m),F}_{{\cal E}}  
(\psi_{\mathcal{M}}^{(m)} (\eta_\saa)))V\nonumber\\
= &
\nabla^{(m)}[F]
(V(\psi_{\mathcal{E}}^{(e)})^{-1}
(\Pro^{(m),F}_{{\cal E}}  
(\psi_{\mathcal{M}}^{(m)} (\eta_\saa))))V\nonumber\\
\stackrel{(b)}{=}&
\nabla^{(m)}[F_{\cal E}]
((\psi_{\mathcal{E}}^{(e)})^{-1}
(\Pro^{(m),F}_{{\cal E}}  
(\psi_{\mathcal{M}}^{(m)} (\eta_\saa))))\nonumber\\
= &
(\psi_{\mathcal{E}}^{(m)})^{-1}
(\Pro^{(m),F}_{{\cal E}}  
(\psi_{\mathcal{M}}^{(m)} (\eta_\saa)))
= 
\Pro^{(m),F}_{\Xi_{{\cal M}} \to\Xi_{\cal E}}  
(\eta_\saa),
\end{align}
where
$(a)$ follows from Lemma \ref{Th5}, and $(b)$ follows from
$F_{\cal E}(\eta_{\scc})=F(V\theta_{\scc})$.
\end{proofof}

\begin{proofof}{Lemma \ref{BXP}}
To show the equivalence between (i) and (ii),
we apply the condition (C1) of Lemma \ref{Th6}. 
Hence, (i) means that
$\psi_{\mathcal{E}}^{(m)} (\eta_\saa V_1)\in {\cal E}
$ and 
$ \psi_{\mathcal{M}}^{(e)}(\theta_\saa)\in {\cal M}$
belong to the same exponential family generated by 
$u_{k+1}, \ldots, u_{d}$.
That is, these two elements have the same $k$ coefficients
on the $k$ vectors $u_{1}, \ldots, u_{k}$.
The $k$ coefficients of 
$ \psi_{\mathcal{M}}^{(e)}(\theta_\saa)\in {\cal M}$
is $\theta_\saa$.
The $k$ coefficients of 
$ \psi_{\mathcal{E}}^{(e)}(\theta_\scc)\in {\cal E}$
is $V_1\theta_\scc$.
That is, the intersection between ${\cal E}$ and the above 
exponential family is the exponential subfamily 
$
\psi_{\mathcal{E}}^{(e)} (
\{ \theta_\scc \in \Theta_\mathcal{E}
| \theta_\saa=V_1 \theta_\scc\})$.
Hence, the condition (i) is equivalent to 
$\psi_{\mathcal{E}}^{(m)} (\eta_\saa V_1)
\in \psi_{\mathcal{E}}^{(e)} (
\{ \theta_\scc \in \Theta_\mathcal{E}
| \theta_\saa=V_1 \theta_\scc\})$, i.e., the condition (ii).

Since $\nabla^{(m)}[F_{\mathcal{E}}^*]( \eta_{\saa} V_1)
=\phi_{\mathcal{E}}^{(e)}\circ\psi_{\mathcal{E}}^{(m)}( \eta_{\saa} V_1)$,
the condition (ii) is equivalent to 
\begin{align}
\theta_\saa
=V_1 \nabla^{(m)}[F_{\mathcal{E}}^*]( \eta_\saa V_1)  .
\end{align}
Since the relation \eqref{MDF} guarantees that
\begin{align}
&V_1 \nabla^{(m)}[F_{\mathcal{E}}^*]( \eta_\saa V_1)  
\nonumber \\
=& L[V_1 ]\circ
\nabla^{(m)}[F_{\mathcal{E}}^*]\circ R[V_1]( \eta_\saa)
\nonumber \\
=&\nabla^{(m)}[ F_{\mathcal{E}}^* \circ R[V_1]]( \eta_\saa ),\Label{NMJ}
\end{align}
the conditions (ii) and (iii) are equivalent.
The relation between \eqref{M1} and \eqref{M2} guarantees the equivalence between the conditions (iii) and (iv).

Since $F_{\mathcal{E}}^* \circ R[V_1 ]$ is a convex function, 
the condition (iv)
is equivalent to the condition 
\begin{align}
\eta_\saa =\argmin_{\eta_\saa' \in \mathbb{R}^{k} }F_{\mathcal{E}}^* \circ R[V_1 ]( \eta_\saa')
- \langle\eta_\saa' , \theta_\saa \rangle\Label{NAD}.
\end{align}
Since $F_{{\cal E}}^*({\eta}_\saa V_1)=
 F_{\mathcal{E}}^* \circ R[V_1 ]( \eta_\saa )$,
the condition (iv)
is equivalent to the condition (v).
\end{proofof}

\section{Proof of Theorem \ref{theo:conv:BBem1}}\Label{A3}
\begin{proofof}{Theorem \ref{theo:conv:BBem1}}
To show the statement (i) of
Theorem \ref{theo:conv:BBem1},
we choose two elements 
$\theta^{(t)},\theta^{(t+1)}\in \mathcal{E}$
as
$\theta^{(t)}=\Pro^{(e),F}_{{\cal M}}\circ \Pro^{(m),F}_{{\cal E}}(\theta^{(t+1)})$.
the input element is characterized by the mixture parameter
$\hat\eta_\saa (\theta^{(t+1)})$
and the output element is characterized by the natural parameter
$\hat\theta_\saa (\theta^{(t)})$
with respect to $\mathcal{M}$.
Then,
$\Pro^{(m),F}_{{\cal E}}(\theta^{(t+1)})$
has the mixture parameter 
$ \hat{\eta}_\saa(\theta^{(t+1)}) V_1$
with respect to ${\cal E}$ due to Lemma \ref{Th5}.
Due to the equivalence between the conditions (i) and (iv) of Lemma \ref{BXP}, 
the mixture parameter $\hat\eta_\saa (\theta^{(t+1)})$ 
and 
the natural parameter $\hat\theta_\saa (\theta^{(t)})$ 
satisfies the following condition;
\begin{align}
\hat{\eta}_\saa(\theta^{(t+1)}) 
 = \nabla^{(e)}
[(F_{\cal E}^*\circ R[V_1])^*]
( \hat{\theta}_\saa(\theta^{(t)})).\Label{M3}
\end{align}
Since ${\cal M}$ is also an exponential subfamily, the function $F_{\cal M}$ is defined.
Hence, the relation \eqref{M3} is rewritten with the natural parameter in ${\cal M}$ as
\begin{align}
\hat{\theta}_\saa(\theta^{(t+1)}) 
 = 
\nabla^{(m)} [F_{\cal M}^*] \circ
 \nabla^{(e)} [(F_{\cal E}^*\circ R[V_1])^*]
 (\hat{\theta}_\saa(\theta^{(t)})) .\Label{AKO}
\end{align}
The condition \eqref{AKO} 
is equivalent to the condition that
$\theta^{(t)}=\Pro^{(e),F}_{{\cal M}}\circ \Pro^{(m),F}_{{\cal M}\to {\cal E}}(\theta^{(t+1)})$
for $\theta^{(t)},\theta^{(t+1)} \in {\cal M}$.
Hence, 
for any $\theta^{(t)} \in \mathcal{M}$, there uniquely exists 
an element $\theta^{(t+1)} \in {\cal M}$ to satisfy the condition 
$\theta^{(t)}=\Pro^{(e),F}_{{\cal M}}\circ 
\Pro^{(m),F}_{{\cal M}\to {\cal E}}(\theta^{(t+1)})$.
Thus, $\nabla^{(m)} [F_{\cal M}^*]\circ \nabla^{(e)} [(F_{\cal E}^*\circ R[V_1])^*]$
is the unique inverse map of 
$\Pro^{(e),F}_{{\cal M}}\circ \Pro^{(m),F}_{{\cal M}\to {\cal E}}$, and is defined in ${\cal M}$.
Hence, 
$\Pro^{(e),F}_{{\cal M}}\circ \Pro^{(m),F}_{{\cal M}\to {\cal E}}$ is a bijective map from ${\cal M}$ to ${\cal M}$.
The statement (i) is obtained.

The statement (ii) follows from the equivalence between the conditions (iv) and (v) of Lemma \ref{BXP}.
\end{proofof}

The key point of the above proof is the following;
Since ${\cal M}$ is an exponential subfamily as well as a mixture subfamily,
the natural parameter is written as the Legendre transform of the mixture parameter,
which is stated as \eqref{AKO}.
 
\section{Proof of Theorem \ref{conv:BBem}}\Label{A4}
(Step 0)
We prepare several relations that are used in this proof.
In this proof, 
we use the notation
$\hat\eta_\saa^{(t+1),*}:=
\argmin_{\hat\eta_\saa \in \mathbb{R}^{k}} 
F_{{\cal E}}^*( \hat{\eta}_\saa V_1) 
-\langle\hat\eta_\saa , \hat{\theta}^{(t)}\rangle$.
We define elements $
\theta_*:=\Pro^{(m),F}_{{\cal E}}  (\theta^*), 
\theta_{(t+1)}:=
\psi_{\mathcal{E}}^{(m)}(\hat\eta_\saa^{(t+1)} V_1),
\theta_{(t+1),*}:=
\psi_{\mathcal{E}}^{(m)}(\hat\eta_\saa^{(t+1),*} V_1)
\in \mathcal{E} $.

\begin{figure}[htbp]
\begin{center}
  \includegraphics[width=0.7\linewidth]{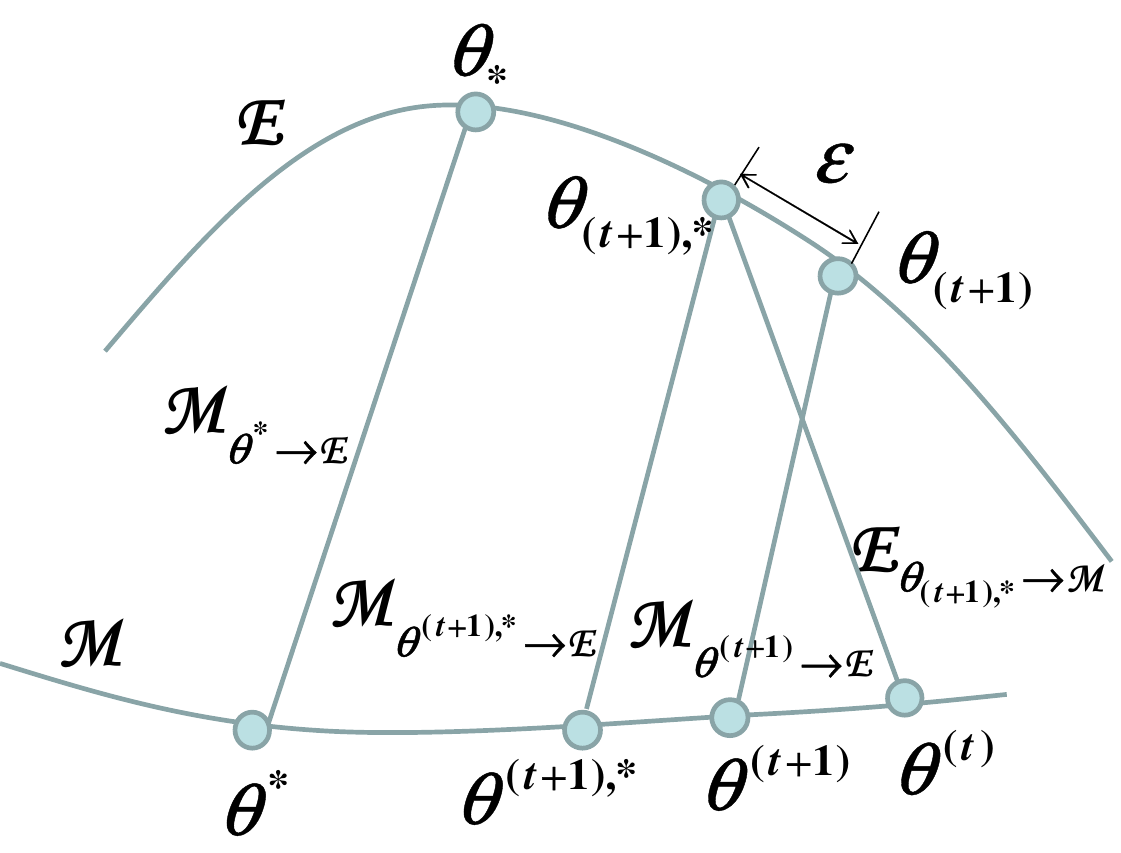}
  \end{center}
\caption{
Algorithm \ref{protocol2F}: 
This figure shows the topological relation among
$\theta_*$, $\theta^*$,
$\theta_{(t+1 )}$, $\theta^{(t+1 )}$,
$\theta_{(t+1 ),*}$, $\theta^{(t+1 ),*}$, and $\theta^{(t)}$,
which is used in the
application of Pythagorean theorem (Proposition \ref{MNL}).
$\mathcal{M}_{\theta^*\to \mathcal{E}}$,
$\mathcal{M}_{\theta^{(t+1)} \to \mathcal{E}}$,
and 
$\mathcal{M}_{\theta^{(t+1),*} \to \mathcal{E}}$
are the mixture subfamilies to project $\theta(\epsilon_1)$,
$\theta^{(t+1),*}$,
and $\theta^{(t+1)}$
to the exponential subfamily $\mathcal{E}$, respectively.
$\mathcal{E}_{\theta_{(t+1),*}\to \mathcal{M}}$
is the exponential subfamily to project $\theta_{(t+1),*}$
to the mixture subfamily $\mathcal{M}$.
}
\Label{FF2D}
\end{figure}   

\noindent{\bf (Step 1):}\quad
The aim of the first step is to show the inequality
\begin{align}
D^F(\theta^{(t+1)} \| \theta^{(t+1),*}) \le \epsilon.\Label{SOP}
\end{align}

We define the mixture subfamily 
\begin{align}
\bar{\mathcal{M}}:=
\{ \psi_{{\cal E}}^{(m)}(\hat\eta_\saa V_1)\in \mathcal{E} | 
\hat\eta_\saa \in \mathbb{R}^k\}.
\end{align}
In this mixture subfamily $\bar{\mathcal{M}}$,
we employ the mixture parameter $\hat\eta_\saa$.
That is, we have 
$\phi_{\bar{\mathcal{M}}}^{(m)}\circ 
\psi_{{\cal E}}^{(m)}(\hat\eta_\saa V_1)
=\hat\eta_\saa$.
Hence, we choose $F_{\bar{\mathcal{M}}}^*$ as
\begin{align}
F_{\bar{\mathcal{M}}}^*(\hat\eta_\saa)
=F_{\mathcal{E}}^*(\hat\eta_\saa V_1)
=F_{\mathcal{E}}^*\circ R[V_1] (\hat\eta_\saa)\Label{BXO}.
\end{align}
Since $\Pro^{(e),F}_{{\cal M}}  (\theta_{(t+1),*})=
\psi_{\mathcal{M}}^{(e)}(\hat\theta^{(t)}_\saa)$,
we have
\begin{align}
\nabla^{(m)}[F_{\bar{\mathcal{M}}}^*](\hat\eta^{(t+1),*}_\saa)
\stackrel{(a)}{=}
\nabla^{(m)}[F_{\mathcal{E}}^*\circ R[V_1]](\hat\eta^{(t+1),*}_\saa)
\stackrel{(b)}{=}
\hat\theta^{(t)}_\saa, \Label{XZL2}
\end{align}
where
$(a)$ and $(b)$ follow from \eqref{BXO} and 
the equivalence between the conditions (i) and (iii) of Lemma \ref{BXP},
 respectively.
Hence, using \eqref{AM8}, we have
\begin{align}
&F_{\bar{\mathcal{M}}}^*(\hat\eta^{(t+1)}_\saa)
-\langle\hat\eta_\saa^{(t+1)}, \hat\theta_\saa^{(t)}\rangle
=
F_{{\mathcal{E}}}^*(\hat\eta_\saa^{(t+1)}V_1)
-\langle\hat\eta_\saa^{(t+1)}, \hat\theta_\saa^{(t)}\rangle \nonumber \\
\le &
F_{{\mathcal{E}}}^*(\hat\eta_\saa^{(t+1),*}V_1)
-\langle\hat\eta_\saa^{(t+1),*}, \hat\theta_\saa^{(t)}\rangle+\epsilon
=F_{\bar{\mathcal{M}}}^*(\hat\eta_\saa^{(t+1),*})
-\langle\hat\eta_\saa^{(t+1),*}, \hat\theta_\saa^{(t)}\rangle
+\epsilon. \Label{XLP}
\end{align}
Therefore,
we have
\begin{align}
& D^F(\theta^{(t+1)} \| \theta^{(t+1),*}) \stackrel{(a)}{=}
D^{F_{\mathcal{E}}^*}(\eta^{(t+1),*}V_1 \|
\eta^{(t+1)}V_1 ) \nonumber \\
\stackrel{(b)}{=}& 
D^{F_{\bar{\mathcal{M}}}^*}(\eta^{(t+1),*} \|
\eta^{(t+1)}) \nonumber \\
\stackrel{(c)}{=}& 
\langle\nabla^{(m)}[F_{\bar{\mathcal{M}}}^*](\hat\eta_\saa^{(t+1),*})
, (\hat\eta_\saa^{(t+1),*}-\hat\eta_\saa^{(t+1)})\rangle
- F_{\bar{\mathcal{M}}}^*(\hat\eta_\saa^{(t+1),*})
+F_{\bar{\mathcal{M}}}^*(\hat\eta_\saa^{(t+1)}) \nonumber \\
\stackrel{(d)}{=}& 
\langle\hat\theta_{\saa,(t),*}
, (\hat\eta_\saa^{(t+1),*}-\hat\eta_\saa^{(t+1)})\rangle
- F_{\bar{\mathcal{M}}}^*(\hat\eta_\saa^{(t+1),*})
+F_{\bar{\mathcal{M}}}^*(\hat\eta_\saa^{(t+1)}) \nonumber \\
\stackrel{(e)}{\le}& \epsilon,
\end{align}
where $(a)$, $(b)$, $(c)$, $(d)$, and $(e)$  
follow from the combination of \eqref{XI1} and \eqref{NBSO},
the application of \eqref{NBSO} to the substitution
of $F_{\mathcal{E}}^*$ and $F_{\bar{\mathcal{M}}}^*=F_{\mathcal{E}}^*\circ R[V_1]$
into $F$ and $F_{\mathcal{E}}$,
\eqref{XZL}, 
\eqref{XZL2}, 
and
\eqref{XLP}, respectively. 
Hence, we obtain \eqref{SOP}.

\noindent{\bf (Step 2):}\quad
The aim of this step is showing 
\begin{align}
 D^F(\theta^* \| \theta^{(t),*}  )-
 D^F(\theta^* \| \theta^{(t+1),*}  )
 \ge 0 \Label{BXV}
\end{align}
for $t =2, \ldots,  t_0$, by induction
when we assume that $t_0$ satisfies the following condition with 
for $t =2, \ldots,  t_0$;
\begin{align}
D^F(\theta^*\|\theta_{*}) 
-
 D^F(\theta^{(t)} \| \theta_{(t+1),*}  )
\ge 
2\gamma \sqrt{D^F(\theta_{*,-}\|\theta_{(1)})
\epsilon}+
\gamma \epsilon.\Label{BLA7}
\end{align}

Applying the Pythagorean Theorem (Proposition \ref{MNL})
to $D^F(\theta^* \|\theta_{(t+1),*})$,
for $t =1, \ldots,  t_0$
we have
\begin{align}
& D^F(\theta^* \|\theta_{(t+1),*})
=
D^F(\theta^* \|\theta_*)
+D^F(\theta_* \|\theta_{(t+1),*})
=
D^F(\theta^* \|\theta^{(t)})
+D^F(\theta^{(t)} \|\theta_{(t+1),*}) 
\Label{XAC}. 
\end{align}
Thus, we have
\begin{align}
&D^F(\theta^* \|\theta^{(t)}) -D^F(\theta_* \|\theta_{(t+1),*})
=
D^F(\theta^* \|\theta_*) -D^F(\theta^{(t)} \|\theta_{(t+1),*})
\nonumber \\
\ge &
D^F(\theta^* \|\theta_*)  -D^F(\theta^{(t+1),*} \|\theta_{(t+1),*})
\ge 0.\Label{NMC}
\end{align}
Due to the assumption of induction,
we have 
\begin{align}
D^F(\theta^* \|\theta^{(t),*}) \le D^F(\theta^*\| \theta^{(2),*})
\stackrel{(a)}{\le}
D^F(\theta^* \|\theta^{(1)})
\Label{XAO},
\end{align}
where $(a)$ follows from \eqref{NMC} with $t=1$.

Since the set $\mathcal{M}_0$ is a star subset for $\theta^*$,
we can apply Theorem \ref{XAM} to 
the set $\mathcal{M}_0$ as a star subset of $\mathcal{M}$ for $\theta_{*,-}$.
Hence, using the above relations, for $t=2, \ldots, t_0$,
we have
\begin{align}
&
D^F(\theta^*\|\theta_{*}) - D^F(\theta^{(t)} \| \theta_{(t+1),*}  )
\stackrel{(a)}{=}  D^F(\theta^* \|\theta^{(t)}) -D(\theta_* \|\theta_{(t+1),*}) \nonumber \\
\stackrel{(b)}{\le} &
D^F(\theta^* \|\theta^{(t),*}) 
+2\gamma \sqrt{D^F(\theta^* \|\theta^{(t),*}) 
D^F(\theta^{(t)} \|\theta^{(t),*})}
\nonumber \\
&+\gamma D^F(\theta^{(t)} \|\theta^{(t),*})
-D^F(\theta_* \|\theta_{(t+1),*}) \nonumber \\
\stackrel{(c)}{\le} &
D^F(\theta^* \|\theta^{(t),*}) 
+2\gamma \sqrt{D^F(\theta^* \|\theta^{(1)}) 
\epsilon
}
+\gamma \epsilon
-D^F(\theta_* \|\theta_{(t+1),*}) \nonumber \\
\stackrel{(d)}{\le} &
D^F(\theta^* \|\theta^{(t),*}) 
+2\gamma \sqrt{D^F(\theta^* \|\theta^{(1)}) 
\epsilon
}
+\gamma \epsilon
-D^F(\theta^* \|\theta^{(t+1),*}) ,
\Label{NO}
\end{align}
where
$(a)$, $(b)$, $(c)$, and $(d)$ follow from 
\eqref{XAC}, 
Theorem \ref{XAM}, the combination of \eqref{XAO} and \eqref{SOP}, and
the condition (B2), respectively.

Thus,
\begin{align}
& D^F(\theta^*\|\theta_{*}) - D^F(\theta^{(t)} \| \theta_{(t+1),*}  )
-2\gamma \sqrt{D^F(\theta^* \|\theta^{(1)}) 
\epsilon}
-\gamma \epsilon\nonumber \\
{\le} &
 D^F(\theta^* \| \theta^{(t),*}  )-
 D^F(\theta^* \| \theta^{(t+1),*}  ).\Label{NXL}
\end{align}
The combination of \eqref{BLA7} and \eqref{NXL} implies the relation \eqref{BXV}.

\noindent{\bf (Step 3):}\quad
The aim of this step is showing 
\begin{align}
&
D^F(\theta^*\|\theta_{*}) - D^F(\theta_f^{(t_1)}\| \Pro^{(m),F}_{{\cal E}} (\theta_f^{(t_1)} ) ) 
\nonumber  \\
\le & 
\max \Big(
 \frac{D^F(\theta^{*}\| \theta^{(1)}  )}{t_1-1}
 + 2\gamma \sqrt{D^F(\theta^{*}\|\theta^{(1)})
\epsilon}+
(\gamma+1) \epsilon  ,\nonumber \\
&\qquad 2\gamma \sqrt{D^F(\theta^{*}\|\theta^{(1)})
\epsilon}+
(\gamma+1) \epsilon
\Big).
\end{align}
To this aim, it sufficient to show 
\begin{align}
&
D^F(\theta^*\|\theta_{*})- 
D^F(\theta_f^{(t_1)}\| \Pro^{(m),F}_{{\cal E}} (\theta_f^{(t_1)} ) ) \nonumber \\
\le & 
 \frac{D^F(\theta^{*}\| \theta^{(1)}  )}{t_1-1}
 + 2\gamma \sqrt{D^F(\theta^{*}\|\theta^{(1)})
\epsilon}+
(\gamma+1) \epsilon    \Label{XOA2}
\end{align}
under the assumption 
\begin{align}
D^F(\theta^*\|\theta_{*})- 
D^F(\theta_f^{(t_1)}\| \Pro^{(m),F}_{{\cal E}} (\theta_f^{(t_1)} ) ) 
\ge 
2\gamma \sqrt{D^F(\theta^{*}\|\theta^{(1)})\epsilon}+(\gamma+1) \epsilon
\Label{JBX}.
\end{align}

The assumption \eqref{JBX} implies that
\begin{align}
D^F(\theta^*\|\theta_{*})- 
D^F(\theta^{(t)}\| \theta_{(t)}  ) 
\ge 
2\gamma \sqrt{D^F(\theta^{*}\|\theta^{(1)})\epsilon}+(\gamma+1) \epsilon
\Label{JBX2}.
\end{align}
for $t=2, \ldots, t_1$.
We have the following relation with $t=1, \ldots, t_1-1$;
\begin{align}
& D^F(\theta^{(t+1)}\| \theta_{(t+1)}  ) +\epsilon
\stackrel{(a)}{\ge} 
D^F(\theta^{(t+1)}\| \theta_{(t+1)}  ) 
+D^F(\theta^{(t+1)}\| \theta_{(t+1),*}  ) 
\nonumber \\
\stackrel{(b)}{=} &
D^F(\theta^{(t+1)}\| \theta_{(t+1),*}  ) 
\stackrel{(c)}{\ge} 
D^F(\theta^{(t)}\| \theta_{(t+1),*}  ) ,
\Label{JBX3}
\end{align}
where $(a)$, $(b)$, and $(c)$ follow from 
\eqref{SOP}, 
Pythagorean theorem (Proposition \ref{MNL}), and
the fact that $\theta^{(t)}=\Pro^{(e),F}_{{\cal M}}(\theta_{(t+1),*}) $,
respectively.
The combination of \eqref{JBX2} and \eqref{JBX3} implies the condition 
\eqref{BLA7} with $t =2, \ldots,  t_1-1$.
Due to the conclusion of (Step 2), we have \eqref{BXV}
for $t=2, \ldots, t_1-1$.
Since (Step 2) derived the relation \eqref{NO} with the same condition, 
the relation \eqref{NO} holds with $t=2, \ldots, t_1-1$.
Therefore, we have
\begin{align}
& D^F(\theta^*\|\theta_{*})- D^F(\theta^{(t+1)}\| \theta_{(t+1)}  ) \nonumber  \\
\stackrel{(a)}{\le}  &
 D^F(\theta^*\|\theta_{*})- D^F(\theta^{(t)}\| \theta_{(t+1),*}  ) +\epsilon\nonumber  \\
\stackrel{(b)}{\le} &
D^F(\theta^* \|\theta^{(t),*}) 
+2\gamma \sqrt{D^F(\theta^* \|\theta^{(1)}) 
\epsilon
}
+(\gamma+1) \epsilon
-D^F(\theta^* \|\theta^{(t+1),*}) ,
\Label{NO2}
\end{align}
where $(a)$ and $(b)$ follow from \eqref{JBX3} and \eqref{NO}, respectively.

Taking the sum for \eqref{NO2}, we have
\begin{align}
&
(t_1-1) \Big(
D^F(\theta^*\|\theta_{*})- 
D^F(\theta_f^{(t_1)}\| \Pro^{(m),F}_{{\cal E}} (\theta_f^{(t_1)} ) ) 
\Big)\nonumber \\
\stackrel{(a)}{\le} &
\sum_{t=1}^{t_1-1}
\Big(
D^F(\theta^*\|\theta_{*})- 
D^F(\theta^{(t+1)}\| \theta_{(t+1)}  ) 
\Big)\nonumber  \\
= & 
D^F(\theta^*\|\theta_{*})- 
D^F(\theta^{(2)}\| \theta_{(2)}  ) 
+\sum_{t=2}^{t_1-1}
\Big(
D^F(\theta^*\|\theta_{*})- 
D^F(\theta^{(t+1)}\| \theta_{(t+1)}  ) 
\Big)\nonumber  \\
\stackrel{(b)}{\le} &
D^F(\theta^*\| \theta^{(1)})
-D^F(\theta^*\| \theta^{(2),*})
\nonumber \\
&+ \sum_{t=2}^{t_1-1}
\Big( D^F(\theta^* \|\theta^{(t),*}) 
+2\gamma \sqrt{D(\theta^* \|\theta^{(1)}) 
\epsilon}
+(\gamma+1) \epsilon
-D^F(\theta^* \|\theta^{(t+1),*})  \Big) \nonumber \\
= & 
 D^F(\theta^* \| \theta^{(1)}  )-  D^F(\theta^* \|\theta^{(t_1),*})
 + 2(t_1-2)\gamma \sqrt{D^F(\theta^*\|\theta^{(1)})
\epsilon}+
(t_1-2)(\gamma+1) \epsilon \nonumber  \\
\le & 
 D^F(\theta^* \| \theta^{(1)}  )
 + 2(t_1-2)\gamma \sqrt{D^F(\theta^*\|\theta^{(1)})
\epsilon}+
(t_1-2)(\gamma+1) \epsilon  ,
\end{align}
where $(a)$ and $(b)$ follow from
the relation $
D^F(\theta_f^{(t_1)}\| \Pro^{(m),F}_{{\cal E}} (\theta_f^{(t_1)} ) ) 
\le D^F(\theta^{(t+1)}\| \theta_{(t+1)}  ) $
and \eqref{NO2}, respectively.
Hence, we have \eqref{XOA2}.

\noindent{\bf (Step 4):}\quad
Finally, we show \eqref{NAC} from \eqref{ANC}.
The condition 
$t_1-1 \ge \frac{2 D^F(\theta^{*} \| \theta^{(1)})}{\epsilon'}$
implies 
$ \frac{D^F(\theta^{*}\| \theta^{(1)}  )}{t_1-1}
\le \frac{\epsilon'}{2}$.
The condition
$\epsilon \le 
\frac{{\epsilon'}^2}{4
(3\gamma+1)^2 {D^F(\theta^{*}\|\theta^{(1)})}}$
implies
$
(3\gamma+1) \sqrt{D^F(\theta^{*}\|\theta^{(1)})
\epsilon}
\le \frac{\epsilon'}{2}$.
Since $D^F(\theta^{*}\|\theta^{(1)})
\ge \epsilon$, we have
$2\gamma \sqrt{D^F(\theta^{*}\|\theta^{(1)}) \epsilon}
+(\gamma+1) \epsilon
\le \frac{\epsilon'}{2}$.
Hence, we have \eqref{NAC}.

\section{Proof of Theorem \ref{VCF}}\Label{A5}
To characterize $\nabla^{(e)} [(F_{\cal E}^*\circ R[V_1])^*](\hat{\theta})$ for 
$\hat{\theta}_\saa \in \Theta_{{\cal M}}$,
we apply \eqref{VCA}.
For $(\theta_\saa,\theta_\sbb)^T\in \Theta_{{\cal M}}$ with $
\theta_\saa \in \mathbb{R}^{k}, \theta_\sbb\in \mathbb{R}^{l-k}$, 
the condition $\hat{\theta}_\saa=(I,V_3) (\theta_\saa,\theta_\sbb)^T$ is equivalent to 
$\theta_\saa= \hat{\theta}_\saa -V_3 \theta_\sbb$.
Hence, \eqref{VCA} implies that
\begin{align}
 (F_{{\cal E}}^* \circ R[(I,V_3)])^* (\hat{\theta}_\saa)
=\min_{\theta_\sbb} F_{\cal E}( \hat{\theta}_\saa 
-V_3 \theta_\sbb,\theta_\sbb) .
\end{align}
The element $\theta_\sbb^*= \argmin _{\theta_\sbb} 
F_{\cal E}( \hat{\theta}_\saa-V_3 \theta_\sbb,\theta_\sbb) $
satisfies the following;
\begin{align}
\nabla^{(e)} [F_{\cal E}]( \hat{\theta}_\saa-V_3 \theta_\sbb,\theta_\sbb)
\left(
\begin{array}{c}
-V_3\\
I_{l-k}
\end{array}
\right)
=0.\Label{SKO}
\end{align}
That is, when the element $\theta_\sbb$ satisfying \eqref{SKO} is written as $\theta_\sbb^*(\hat{\theta}_\saa)$,
we have
\begin{align}
 (F_{{\cal E}}^* \circ R[(I,V_3)])^* (\hat{\theta}_\saa)
= F_{\cal E}( \hat\theta_\saa-V_3 \theta_\sbb^*(\hat{\theta}_\saa),\theta_\sbb^*(\hat{\theta}_\saa)) .\Label{VGK}
\end{align}
Taking the derivative for $\hat\theta_\saa$ in \eqref{VGK} 
and using the relation \eqref{SKO}, we have 
\begin{align}
 \nabla^{(e)} [(F_{\cal E}^*\circ R[V_1])^*](\hat{\theta}_\saa) 
 =&\nabla^{(e)}[F_{\cal E}](\bar{\theta}_\scc)
  \left( 
\begin{array}{c}
 I_k \\
 0
\end{array}
 \right),\Label{VBI2}
 \end{align}
which implies \eqref{VBI}.


\begin{thebibliography}{99}
\bibitem{Amari}
S. Amari, 
``Information geometry of the EM and em-algorithms for neural networks,''
{\em Neural Networks} 8(9): 1379 -- 1408 (1995). 

\bibitem{Fujimoto}
Y. Fujimoto and N. Murata,
``A modified EM algorithm for mixture models based on Bregman divergence,''
{\em Annals of the Institute of Statistical Mathematics},
vol. 59, 3 -- 25 (2007).

\bibitem{Allassonniere}
S. Allassonni\`{e}re and J. Chevallier,
``A New Class of EM Algorithms. Escaping Local Minima and Handling Intractable Sampling,''. 
{\em Computational Statistics \& Data Analysis}, Elsevier, vol. 159(C). (2019)

\bibitem{Amari-Nagaoka} 
S. Amari and H. Nagaoka, 
{\em Methods~of~Information~Geometry} (AMS and Oxford, 2000).

\bibitem{Amari-Bregman} 
S. Amari,
``$\alpha$-Divergence Is Unique, Belonging to Both f-Divergence and Bregman Divergence Classes,''
{\em IEEE Trans. Inform. Theory}, 
vol. 55, 4925 -- 4931 (2009).

\bibitem{Shoji} S. Toyota,``Geometry of Arimoto algorithm,'' 
{\em Information Geometry}, vol. 3, 183 -- 198 (2020) 

\bibitem{Shannon} C. E.  Shannon,``A Mathematical Theory of Communication,'' 
{\em Bell System Technical Journal}, 
vol.27, 379--423 and 623--656 (1948).

\bibitem{Arimoto} 
S. Arimoto, ``An algorithm for computing the capacity of arbitrary discrete memoryless channels,'' 
{\em IEEE Trans. Inform. Theory}, vol. 18, no. 1, 14 -- 20 (1972).

\bibitem{Blahut} 
R. Blahut, 
``Computation of channel capacity and rate-distortion functions,''
{\em IEEE Trans. Inform. Theory}, 
vol. 18, no. 4, 460 -- 473 (1972).

\bibitem{Matz} G. Matz and P. Duhamel, 
``Information geometric formulation and interpretation of accelerated Blahut-Arimoto-Type algorithms,'' 
in 
{\em Proc. Information Theory Workshop}, 
24--29, San Antonio, Texas, October, (2004).

\bibitem{Yu} Yaming Yu, ``Squeezing the Arimoto-Blahut algorithm for faster convergence,'' 
{\em IEEE Trans. Inform. Theory}, 
vol. 56, 3149 -- 3157 (2010).

\bibitem{SSML}
T. Sutter, D. Sutter, P. Mohajerin Esfahani, and  J. Lygeros,
``Efficient Approximation of Channel Capacities,''
{\em IEEE Trans. Inform. Theory}, 
vol. 61, no. 4, 1649 -- 1666 (2015). 

\bibitem{NWS}
K. Nakagawa, K. Watabe, and T. Sabu,
``On the Search Algorithm for the Output Distribution That Achieves the Channel Capacity,''
{\em IEEE Trans. Inform. Theory}, 
vol.  63, 1043 -- 1062 (2017).

\bibitem{exact} 
M. Hayashi, ``Analytical algorithm for capacities of classical and classical-quantum channels,''
{\em IEEE Trans. Inform. Theory}, 
vol. 69, 1680 -- 1694 (2023).
arXiv:2201.02450 
(2022).

\bibitem{Muroga}
S. Muroga,
``On the Capacity of a Discrete Channel. I Mathematical expression of capacity of a channel which is disturbed by noise in its every one symbol and expressible in one state diagram,''
{\em Journal of the Physical Society of Japan}, 
{\bf 8}, 484-494 (1953). 

\bibitem{Wyner}
A. D. Wyner, ``The wire-tap channel,'' 
{\em Bell. Sys. Tech. Jour.}, {\bf 54} 1355 -- 1387 (1975).

\bibitem{CK79} I. Csisz\'{a}r and J. K\"{o}rner,
``Broadcast channels with confidential messages,'' 
{\em IEEE Trans. Inform. Theory}, 
vol. 24, no. 3, 339 -- 348 (1978).

\bibitem{Holevo}
A.S. Holevo, ``The capacity of the quantum channel with general signal states,'' 
{\em IEEE Trans. Inform. Theory}, 
vol. 44, 269 (1998)

\bibitem{SW}
B. Schumacher, and M.D. Westmoreland, ``Sending classical information via noisy quantum channelsm''
{\em Phys. Rev. A} vol.  56, 131 (1997)

\bibitem{Yasui}
K. Yasui, T. Suko, and T. Matsushima,
``An Algorithm for Computing the Secrecy Capacity of Broadcast Channels with Confidential Messages,''
{\em Proc. 2007 IEEE Int. Symp. Information Theory (ISIT 2007)}, 
Nice, France, 24-29 June 2007, pp. 936 -- 940.

\bibitem{Nagaoka} 
H. Nagaoka,
``Algorithms of Arimoto-Blahut type for computing quantum channel capacity,''
{\em Proc. 1998 IEEE Int. Symp. Information Theory (ISIT 1998)}, 
Cambridge, MA, USA, 16-21 Aug. 1998, pp. 354.

\bibitem{Dupuis} 
F. Dupuis, W. Yu, and F. Willems, 
``Blahut-Arimoto algorithms for computing channel capacity and rate-distortion with side information,''
{\em Proc. 2014 IEEE Int. Symp. Information Theory (ISIT 2014)}, 
Chicago, IL, USA, 27 June-2 July 2004, pp. 179.

\bibitem{Sutter} 
D. Sutter, T. Sutter, P. M. Esfahani, and R. Renner, 
``Efficient approximation of quantum channel capacities,'' 
{\em IEEE Trans. Inform. Theory}, 
vol. 62,  578 -- 598 (2016).

\bibitem{Li} 
H. Li and N. Cai,
``A Blahut-Arimoto Type Algorithm for Computing Classical-Quantum Channel Capacity,''
{\em Proc. 2019 IEEE Int. Symp. Information Theory (ISIT 2019)}, 
Paris, France, 7-12 July 2019, pp. 255--259.

\bibitem{RISB}
N. Ramakrishnan, R. Iten. V. B. Scholz, and M. Berta,
``Computing Quantum Channel Capacities,''
{\em IEEE Trans. Inform. Theory}, 
vol. 67, 946 -- 960 (2021).

\bibitem{em-only}
M. Hayashi, ``Bregman divergence based em-algorithm 
and its application to classical and quantum rate distortion theory,''
{\em IEEE Trans. Inform. Theory}, 
vol. 69, no. 6, 3460 -- 3492 (2023).

\bibitem{hayashi} M. Hayashi, {\em Quantum Information Theory: Mathematical Foundation,} 
Graduate Texts in Physics, Springer-Verlag, (2017). 

\bibitem{Toyota}
S. Toyota, Private communication (2019).

\bibitem{Kavcic}
A. Kav\v{c}i\'{c}, 
``On the capacity of Markov sources over noisy channels,''
in {\em Proc. IEEE Global Telecommun. Conf.,} San Antonio, TX, USA,
Nov. 2001, pp. 2997 -- 3001.

\bibitem{Vontobel}
P. O. Vontobel, A. Kav\v{c}i\'{c}, D. M. Arnold, and H.-A. Loeliger, 
``A generalization of the Blahut–Arimoto algorithm to finite-state channels,'' 
{\em IEEE Trans. Inform. Theory}, vol. 54, no. 5, pp. 1887 -- 1918, May 2008.

\bibitem{Wu}
C. Wu, G. Han, V. Anantharam and B. Marcus, 
``A Deterministic Algorithm for the Capacity of Finite-State Channels,'' 
{\em IEEE Trans. Inform. Theory}, 
vol. 68, no. 3, pp. 1465 -- 1479, March 2022,

\bibitem{Nakagawa}
K. Nakagawa  and F. Kanaya, 
``On the converse theorem in statistical hypothesis testing for Markov chains,'' 
{\em IEEE Trans. Inform. Theory}, 
vol. 39, 629 -- 633 (1993).

\bibitem{Nagaoka}
H. Nagaoka
``The exponential family of Markov chains and its information geometry,''
In {\em Proceedings of the 28th Symposium on Information Theory and Its Applications} (SITA2005), Okinawa, Japan (2005). 

\bibitem{HW}
M. Hayashi and S. Watanabe, `` Information Geometry Approach to Parameter Estimation in Markov Chains,'' 
{\em Annals of Statistics}, vol. 44, no. 4, 1495 -- 1535 (2016). 


\end{thebibliography}
\end{document}